\documentclass[11pt]{article}
\usepackage[margin=.8in,left=.8in]{geometry}
\usepackage{amsmath}
\usepackage{amsfonts}
\usepackage{amssymb}
\usepackage{amsthm}
\usepackage{mathtools}
\usepackage{subcaption}
\usepackage{float}
\usepackage{graphicx}
\usepackage{color}
\usepackage{xcolor}
\usepackage{xfrac}
\usepackage{stmaryrd}
\usepackage{rotating}
\usepackage[utf8]{inputenc}
\usepackage{hyperref}
\usepackage[all]{xy}

\usepackage{tikz}
\usepackage{tikz-cd}
\usepackage{lipsum}
\usepackage{adjustbox}

\definecolor{redi}{RGB}{255,38,0}
\definecolor{redii}{RGB}{200,50,30}
\definecolor{yellowi}{RGB}{255,251,0}
\definecolor{bluei}{RGB}{0,150,255}
\definecolor{blueii}{RGB}{135,247,210}
\definecolor{blueiii}{RGB}{91,205,250}
\definecolor{blueiv}{RGB}{115,244,253}
\definecolor{bluev}{RGB}{1,58,215}
\definecolor{orangei}{RGB}{220,160, 20}
\definecolor{yellowii}{RGB}{222,247,100}
\definecolor{greeni}{RGB}{85,102,0}
\definecolor{navy}{RGB}{17, 10, 102}
\definecolor{brown}{RGB}{85, 102, 0}
\definecolor{oxford}{RGB}{0, 0, 100}

\usepackage{multirow}

\usepackage{stmaryrd}    

\usepackage{enumerate} 

\usepackage{helvet}   

\usepackage{mathptmx}
\usepackage{amsmath}
\usepackage{graphicx}
\DeclareRobustCommand{\coprod}{\mathop{\text{\fakecoprod}}}
\newcommand{\fakecoprod}{%
  \sbox0{$\prod$}%
  \smash{\raisebox{\dimexpr.9625\depth-\dp0}{\scalebox{1}[-1]{$\prod$}}}%
  \vphantom{$\prod$}%
}

\usepackage{floatflt}  
\usepackage{wrapfig} 
\usepackage{array}     
\newcolumntype{L}[1]{>{\raggedright\let\newline\\\arraybackslash\hspace{0pt}}m{#1}}
\newcolumntype{C}[1]{>{\centering\let\newline\\\arraybackslash\hspace{0pt}}m{#1}}
\newcolumntype{R}[1]{>{\raggedleft\let\newline\\\arraybackslash\hspace{0pt}}m{#1}}

\usepackage{diagbox}  

\makeatletter
\newcommand{\raisemath}[1]{\mathpalette{\raisem@th{#1}}}
\newcommand{\raisem@th}[3]{\raisebox{#1}{$#2#3$}}
\makeatother

\usepackage{tabularx}   

\usepackage[new]{old-arrows}   

\newdir{> }{{}*!/10pt/@{>}}

\DeclareRobustCommand{\rchi}{{\mathpalette\irchi\relax}}
\newcommand{\irchi}[2]{\raisebox{\depth}{$#1\chi$}} 

\makeatletter
\newif\if@sup
\newtoks\@sups
\def\append@sup#1{\edef\act{\noexpand\@sups={\the\@sups #1}}\act}%
\def\reset@sup{\@supfalse\@sups={}}%
\def\mk@scripts#1#2{\if #2/ \if@sup ^{\the\@sups}\fi \else%
  \ifx #1_ \if@sup ^{\the\@sups}\reset@sup \fi {}_{#2}%
  \else \append@sup#2 \@suptrue \fi%
  \expandafter\mk@scripts\fi}
\def\tensor#1#2{\reset@sup#1\mk@scripts#2_/}
\def\multiscripts#1#2#3{\reset@sup{}\mk@scripts#1_/#2%
  \reset@sup\mk@scripts#3_/}
\makeatother

\makeatletter
\newbox\slashbox \setbox\slashbox=\hbox{$/$}
\def\itex@pslash#1{\setbox\@tempboxa=\hbox{$#1$}
  \@tempdima=0.5\wd\slashbox \advance\@tempdima 0.5\wd\@tempboxa
  \copy\slashbox \kern-\@tempdima \box\@tempboxa}
\def\slash{\protect\itex@pslash}
\makeatother

\def\clap#1{\hbox to 0pt{\hss#1\hss}}
\def\mathllap{\mathpalette\mathllapinternal}
\def\mathrlap{\mathpalette\mathrlapinternal}
\def\mathclap{\mathpalette\mathclapinternal}
\def\mathllapinternal#1#2{\llap{$\mathsurround=0pt#1{#2}$}}
\def\mathrlapinternal#1#2{\rlap{$\mathsurround=0pt#1{#2}$}}
\def\mathclapinternal#1#2{\clap{$\mathsurround=0pt#1{#2}$}}

\let\oldroot\root
\def\root#1#2{\oldroot #1 \of{#2}}
\renewcommand{\sqrt}[2][]{\oldroot #1 \of{#2}}

\DeclareSymbolFont{symbolsC}{U}{txsyc}{m}{n}
\SetSymbolFont{symbolsC}{bold}{U}{txsyc}{bx}{n}
\DeclareFontSubstitution{U}{txsyc}{m}{n}

\DeclareSymbolFont{stmry}{U}{stmry}{m}{n}
\SetSymbolFont{stmry}{bold}{U}{stmry}{b}{n}

\DeclareFontFamily{OMX}{MnSymbolE}{}
\DeclareSymbolFont{mnomx}{OMX}{MnSymbolE}{m}{n}
\SetSymbolFont{mnomx}{bold}{OMX}{MnSymbolE}{b}{n}
\DeclareFontShape{OMX}{MnSymbolE}{m}{n}{
    <-6>  MnSymbolE5
   <6-7>  MnSymbolE6
   <7-8>  MnSymbolE7
   <8-9>  MnSymbolE8
   <9-10> MnSymbolE9
  <10-12> MnSymbolE10
  <12->   MnSymbolE12}{}

\makeatletter
\def\Decl@Mn@Delim#1#2#3#4{%
  \if\relax\noexpand#1%
    \let#1\undefined
  \fi
  \DeclareMathDelimiter{#1}{#2}{#3}{#4}{#3}{#4}}
\def\Decl@Mn@Open#1#2#3{\Decl@Mn@Delim{#1}{\mathopen}{#2}{#3}}
\def\Decl@Mn@Close#1#2#3{\Decl@Mn@Delim{#1}{\mathclose}{#2}{#3}}
\Decl@Mn@Open{\llangle}{mnomx}{'164}
\Decl@Mn@Close{\rrangle}{mnomx}{'171}
\Decl@Mn@Open{\lmoustache}{mnomx}{'245}
\Decl@Mn@Close{\rmoustache}{mnomx}{'244}
\makeatother

\makeatletter
\DeclareRobustCommand\widecheck[1]{{\mathpalette\@widecheck{#1}}}
\def\@widecheck#1#2{%
    \setbox\z@\hbox{\m@th$#1#2$}%
    \setbox\tw@\hbox{\m@th$#1%
       \widehat{%
          \vrule\@width\z@\@height\ht\z@
          \vrule\@height\z@\@width\wd\z@}$}%
    \dp\tw@-\ht\z@
    \@tempdima\ht\z@ \advance\@tempdima2\ht\tw@ \divide\@tempdima\thr@@
    \setbox\tw@\hbox{%
       \raise\@tempdima\hbox{\scalebox{1}[-1]{\lower\@tempdima\box
\tw@}}}%
    {\ooalign{\box\tw@ \cr \box\z@}}}
\makeatother

\makeatletter
\def\udots{\mathinner{\mkern2mu\raise\p@\hbox{.}
\mkern2mu\raise4\p@\hbox{.}\mkern1mu
\raise7\p@\vbox{\kern7\p@\hbox{.}}\mkern1mu}}
\makeatother

\newcommand{\gt}{>}
\newcommand{\lt}{<}

\renewcommand{\(}{\begin{equation}}
\renewcommand{\)}{\end{equation}}
\newcommand{\bea}{\begin{eqnarray*}}
\newcommand{\eea}{\end{eqnarray*}}

\usepackage{cleveref}

\crefformat{section}{\S#2#1#3} 
\crefformat{subsection}{\S#2#1#3}
\crefformat{subsubsection}{\S#2#1#3}

\theoremstyle{italics}
\newtheorem{theorem}{Theorem}[section]
\newtheorem{lemma}[theorem]{Lemma}
\newtheorem{prop}[theorem]{Proposition}

\theoremstyle{definition}
\newtheorem{defn}[theorem]{Definition}

\newtheorem{example}[theorem]{Example}

\newtheorem{remark}[theorem]{Remark}
\newtheorem{note[theorem]}{Note}

\usepackage{amsfonts}

\begin{document}

\title{
  Twisted Cohomotopy implies
  M-theory
  \\
  anomaly cancellation on 8-manifolds
}

\author{Domenico Fiorenza, \; Hisham Sati, \; Urs Schreiber}
%

\maketitle

\vspace{-10mm}
\begin{center}
{\it \footnotesize To Mike Duff on the occasion of his 70th birthday}
\end{center}
\begin{abstract}
  We consider the hypothesis
  that the C-field 4-flux and 7-flux forms in M-theory
  are in the image of the non-abelian Chern character map
  from the non-abelian generalized cohomology theory
  called J-twisted Cohomotopy theory.
  We prove for M2-brane backgrounds in M-theory on 8-manifolds
  that such charge quantization of the C-field in Cohomotopy theory
  implies a list of expected anomaly cancellation conditions,
  including:
  shifted C-field flux quantization and C-field tadpole cancellation,
  but also the DMW anomaly cancellation
  and the C-field's integral equation of motion.
\end{abstract}

\vspace{-.6cm}

\setcounter{tocdepth}{2}

\tableofcontents

\newpage

\section{Introduction and survey}
\label{IntroductionAndSurvey}

\noindent We consider the following hypothesis, which we
make precise as Def. \ref{HypothesisHFor8Manifolds} below,
based on details developed in \cref{Cohomotopy},
see \cref{Conclusions} for background, motivation and outlook:

\medskip

\hypertarget{HypothesisH}{}
\noindent {\bf Hypothesis H}:
{\it
The C-field 4-flux and 7-flux forms in M-theory
are subject
to
\emph{charge quantization in J-twisted Cohomotopy cohomology theory}
in that they are in the image of the non-abelian Chern character
map from J-twisted Cohomotopy theory.}

\medskip

\noindent In support of Hypothesis H, we prove
in \cref{CancellationFromCohomotopy} that it implies the following phenomena,
expected for M2-brane backgrounds in M-theory on 8-manifolds
(recalled in Remark \ref{MTheoryOn8Manifolds} below):

\medskip

{\small
\hypertarget{Table1}{}
\hspace{-1cm}
\begin{tabular}{|cc|cc|c|}
  \hline
  \multicolumn{2}{|c}{
  $
  \mathclap{\phantom{a^{\vert^{\vert^{\vert}}}}}
  \mathclap{\phantom{a_{\vert_{\vert_{\vert}}}}}
  $
  {\color{oxford} \bf Cohomotopy theory}
  }
  &
  \multicolumn{2}{c|}{
    Expression
  }
  &
  {\color{oxford} \bf M-theory}
  \\
  \hline
  \hline
  \multirow{2}{*}{
  \hspace{-.2cm}
  \cref{W7Cancellation}
  \hspace{-.5cm}
  }
  &
  \multirow{2}{*}{
  \hspace{-.3cm}
  \begin{tabular}{c}
    $\mathclap{\phantom{a^{\vert^{\vert^{\vert}}}}}$
    Compatible twisting
    \\
    $\mathclap{\phantom{a_{\vert_{\vert_{\vert}}}}}$
    on 4- \& 7-Cohomotopy theory
  \end{tabular}
  \hspace{-.3cm}
  }
  &
  $\mathclap{\phantom{ \vert \atop \vert }}$
  $W_7(T X) = 0$
  &
  \eqref{SP1Sp2Structure}
  &
  \hspace{-.4cm}
  \begin{tabular}{c}
    $\mathclap{\phantom{a^{\vert^{\vert^{\vert}}}}}$
    DMW anomaly cancellation condition
    \\
    $\mathclap{\phantom{a_{\vert_{\vert_{\vert}}}}}$
    \cite{DMWa}\cite[6]{DMWb}
  \end{tabular}
  \hspace{-.4cm}
  \\
  \cline{3-5}
  &
  &
  $\mathclap{\phantom{ \vert \atop \vert }}$
  $
    \begin{aligned}
    &
    \mathclap{\phantom{a^{\vert^{\vert^{\vert}}}}}
    \tfrac{1}{24} \rchi_8(T X)
    \;=\;
    I_8(T X)
    \\
    &
    \mathclap{\phantom{a_{\vert_{\vert_{\vert}}}}}
    :=\;
    \tfrac{1}{48}
    \big(
      p_2(T X)
      -
      \tfrac{1}{4}(p_1(T X))^2
    \big)
    \end{aligned}
  $
  &
  \eqref{SP1Sp2Structure}
  &
  \hspace{-.4cm}
  \begin{tabular}{c}
    one-loop anomaly polynomial
    \\
    \cite{DuffLiuMinasian95}\cite{VafaWitten95}
  \end{tabular}
  \hspace{-.4cm}
  \\
  \hline
  \hspace{-.2cm}
  \cref{TwistedCohomotopyInDegreeSeven}
  \hspace{-.5cm}
  &
  \hspace{-.3cm}
  \begin{tabular}{c}
   $\mathclap{\phantom{a^{\vert^{\vert^{\vert}}}}}$
   Any cocycle
   \\
   $\mathclap{\phantom{a_{\vert_{\vert_{\vert}}}}}$
   in J-twisted 7-Cohomotopy
  \end{tabular}
  \hspace{-.3cm}
  &
  $\mathrm{Spin}(7)$-structure $g$
  &
  \eqref{CohomotopyAndSpin7}
  &
  \hspace{-.4cm}
  \begin{tabular}{c}
    $\geq \sfrac{1}{8}$ BPS M2-brane background
    \\
    \cite{IP}\cite{IPW}\cite{Tsimpis06}
  \end{tabular}
  \hspace{-.4cm}
  \\
  \hline
  \hspace{-.2cm}
  \cref{TwistedCohomotopyInDegrees}
  \hspace{-.5cm}
  &
  \hspace{-.3cm}
  \begin{tabular}{c}
    $\mathclap{\phantom{a^{\vert^{\vert^{\vert}}}}}$
    Any cocycle in
    \\
    $\mathclap{\phantom{a_{\vert_{\vert_{\vert}}}}}$
    compatibly twisted 4\&7-Cohomotopy
  \end{tabular}
  \hspace{-.3cm}
  &
  $\mathrm{Sp}(1)\cdot \mathrm{Sp}(1)$-structure $\tau$
  &
  \eqref{CohomotopyAndSp2Sp1}
  &
  \hspace{-.4cm}
  \begin{tabular}{c}
    $\sfrac{4}{8}$ BPS M2-brane background
    \\
    \cite[7.3]{MF10}
  \end{tabular}
  \hspace{-.4cm}
  \\
  \hline
  \hline
  \hspace{-.2cm}
  \cref{CFieldBackgroundCharge}
  \hspace{-.5cm}
  &
  \hspace{-.3cm}
  \begin{tabular}{c}
    Chern character of
    \\
    rationally twisted 4-Cohomotopy
  \end{tabular}
  \hspace{-.3cm}
  &
  $
  \begin{aligned}
    \mathclap{\phantom{a^{\vert^{\vert^{\vert}}}}}
    d\, G_4 & = \; 0
    \\
    d\, G_7
    & =\;
    -\tfrac{1}{2}
    G_4 \wedge G_4
    +
    L_8
    \mathclap{\phantom{a_{\vert_{\vert_{\vert}}}}}
  \end{aligned}
  $
  &
  \eqref{RationalTwists}
  &
  \hspace{-.4cm}
  \begin{tabular}{c}
    $\mathclap{\phantom{a^{\vert^{\vert^{\vert}}}}}$
    C-field Bianchi identity with
    \\
    generic higher curvature correction
    \\
    $\mathclap{\phantom{a_{\vert_{\vert_{\vert}}}}}$
    \cite{SoueresTsimpis16}
  \end{tabular}
  \hspace{-.4cm}
  \\
  \hline
  \hspace{-.2cm}
  \cref{CFieldBackgroundCharge}
  \hspace{-.5cm}
  &
  \hspace{-.3cm}
  \begin{tabular}{c}
    Chern character of compatibly
    \\
    rationally twisted 4\&7-Cohomotopy
  \end{tabular}
  \hspace{-.3cm}
  &
  \raisebox{0pt}{
  $
  \begin{aligned}
    \mathclap{\phantom{a^{\vert^{\vert^{\vert}}}}}
    d\, \widetilde G_4 & =\; 0
    \\
    d\, G_7
    &=\;
    -\tfrac{1}{2}
    \big(
      \widetilde G_4 - \tfrac{1}{4}P_4
    \big)
      \wedge
    \widetilde G_4
    +
    K_8
    \mathclap{\phantom{a_{\vert_{\vert_{\vert}}}}}
  \end{aligned}
  $
  }
  \raisebox{6pt}{
    \hspace{-.9cm}
  }
  &
  \eqref{CompatibleRationalTwists}
  &
  \hspace{-.4cm}
  \begin{tabular}{c}
    $\mathclap{\phantom{a^{\vert^{\vert^{\vert}}}}}$
    Shifted C-field Bianchi identity with
    \\
    generic higher curvature correction
    \\
    \cite{Tsimpis04}
  \end{tabular}
  \hspace{-.4cm}
  \\
  \hline
  \hline
  \multirow{2}{*}{
  \hspace{-.2cm}
  \cref{HalfIntegralCFieldFluxQuantization}
  \hspace{-.5cm}
  }
  &
  \hspace{-.3cm}
  \multirow{4}{*}{
  \begin{tabular}{c}
    Chern character 4-form of
    \\
    $\mathrm{Sp}(2)$-twisted 4-Cohomotopy
  \end{tabular}
  }
  \hspace{-.3cm}
  &
  $
  \mathclap{\phantom{a^{\vert^{\vert^{\vert}}}}}
  \widetilde G_4 \;=\; G_4 + \tfrac{1}{4}p_1(\nabla)
  \mathclap{\phantom{a_{\vert_{\vert_{\vert}}}}}
  $
  &
  \eqref{ShiftByBackgroundCharge}
  &
  \hspace{-.4cm}
  \begin{tabular}{c}
    $\mathclap{\phantom{a^{\vert^{\vert^{\vert}}}}}$
    C-field shift
    \\
    \cite{Witten96a}\cite{Witten96b}\cite{Tsimpis04}
    $\mathclap{\phantom{a_{\vert_{\vert_{\vert}}}}}$
  \end{tabular}
  \hspace{-.4cm}
  \\
  \cline{3-5}
  &
  &
  $
  \mathclap{\phantom{a^{\vert^{\vert^{\vert}}}}}
  [\widetilde G_4] \;\in\; H^4(X, \mathbb{Z})
  \mathclap{\phantom{a_{\vert_{\vert_{\vert}}}}}
  $
  &
  \eqref{ShiftedFluxQuantizationCondition}
  &
  \hspace{-.4cm}
  \begin{tabular}{c}
    $\mathclap{\phantom{a^{\vert^{\vert^{\vert}}}}}$
    Shifted C-field flux quantization
    \\
    $\mathclap{\phantom{a_{\vert_{\vert_{\vert}}}}}$
    \cite{Witten96a}\cite{Witten96b}\cite{DFM03}\cite{HopkinsSinger05}
  \end{tabular}
  \hspace{-.4cm}
  \\
  \cline{3-5}
  \hspace{-.2cm}
  \cref{BackgroundCharge}
  \hspace{-.5cm}
  &
  &
  $
  \mathclap{\phantom{a^{\vert^{\vert^{\vert}}}}}
  (G_4)_0 \;=\; \tfrac{1}{2}p_1(\nabla)
  \mathclap{\phantom{a_{\vert_{\vert_{\vert}}}}}
  $
  &
  \eqref{BackgroundChargeValue}
  &
  \hspace{-.4cm}
  \begin{tabular}{c}
    $\mathclap{\phantom{a^{\vert^{\vert^{\vert}}}}}$
    Background charge
    \\
    $\mathclap{\phantom{a_{\vert_{\vert_{\vert}}}}}$
    \cite[p. 11]{Freed09}\cite{Freed00}
  \end{tabular}
  \hspace{-.4cm}
  \\
  \cline{3-5}
  \hspace{-.2cm}
  \cref{IntegralEquationOfMotionFromCohomotopy}
  \hspace{-.5cm}
  &
  &
  $
  \mathclap{\phantom{a^{\vert^{\vert^{\vert}}}}}
  \mathrm{Sq}^3\big( [\widetilde G_4] \big) \;=\; 0
  \mathclap{\phantom{a_{\vert_{\vert_{\vert}}}}}
  $
  &
  \eqref{SteenrodSquareVanishes}
  &
  \hspace{-.4cm}
  \begin{tabular}{c}
    $\mathclap{\phantom{a^{\vert^{\vert^{\vert}}}}}$
    Integral equation of motion
    \\
    $\mathclap{\phantom{a_{\vert_{\vert_{\vert}}}}}$
    \cite{DMWa}\cite[5]{DMWb}
  \end{tabular}
  \hspace{-.4cm}
  \\
  \hline
  \multirow{3}{*}{
  \hspace{-.2cm}
  \cref{PageCharge}
  \hspace{-.5cm}
  }
  &
  \hspace{-.3cm}
  \multirow{3}{*}{
  \begin{tabular}{c}
    Chern character 7-form of compatibly
    \\
    $\mathrm{Sp}(2)$-twisted 4\&7-Cohomotopy
  \end{tabular}
  }
  \hspace{-.3cm}
  &
  $
  \mathclap{\phantom{a^{\vert^{\vert^{\vert}}}}}
  \widetilde G_7 \;=\; G_7 + \tfrac{1}{2}H_3 \wedge \widetilde G_4
  \mathclap{\phantom{a_{\vert_{\vert_{\vert}}}}}
  $
  &
  \eqref{PageCharge7Form}
  &
  \multirow{2}{*}{
  \hspace{-.4cm}
  \begin{tabular}{c}
    $\mathclap{\phantom{a^{\vert^{\vert^{\vert}}}}}$
    Page charge
    \\
    $\mathclap{\phantom{a_{\vert_{\vert_{\vert}}}}}$
    \cite[(8)]{Page83}\cite[(43)]{DuffStelle91}\cite{Moore05}
  \end{tabular}
  \hspace{-.4cm}
  }
  \\
  \cline{3-4}
  &
  &
  $
  \mathclap{\phantom{a^{\vert^{\vert^{\vert}}}}}
  d\, \tilde G_7
    \;=\;
  - \tfrac{1}{2}\rchi_{8}(\nabla)
  \mathclap{\phantom{a_{\vert_{\vert_{\vert}}}}}
  $
  &
  \eqref{PageChargeFormIsClosed}
  &
  \\
  \cline{3-5}
  &
  &
  $
  \mathclap{\phantom{a^{\vert^{\vert^{\vert}}}}}
  2\int_{{}_{S^7}} i^\ast \widetilde G_7 \;\in\; \mathbb{Z}
  \mathclap{\phantom{a_{\vert_{\vert_{\vert}}}}}
  $
  &
  \eqref{HalfIntegral7Flux}
  &
  \hspace{-.3cm}
  \begin{tabular}{c}
    $\mathclap{\phantom{a^{\vert^{\vert^{\vert}}}}}$
    Level quantization of Hopf-WZ term
    \\
    $\mathclap{\phantom{a_{\vert_{\vert_{\vert}}}}}$
    \cite{Intriligator00}
  \end{tabular}
  \hspace{-.3cm}
  \\
  \hline
  \hline
  \hspace{-.2cm}
  \cref{M2BraneTadpoleCancellation}
  \hspace{-.5cm}
  &
  \hspace{-.4cm}
  \begin{tabular}{c}
    $\mathclap{\phantom{a^{\vert^{\vert^{\vert}}}}}$
    Integrated Chern character of compatibly
    \\
    $\mathclap{\phantom{a_{\vert_{\vert_{\vert}}}}}$
    $\mathrm{Sp}(2)$-twisted 4\&7-Cohomotopy
  \end{tabular}
  \hspace{-.4cm}
  &
  $
  \mathclap{\phantom{a^{\vert^{\vert^{\vert}}}}}
  N_{M2} \;=\; - I_8
  \mathclap{\phantom{a_{\vert_{\vert_{\vert}}}}}
  $
  &
  \eqref{NIsI}
  &
  \begin{tabular}{c}
    C-field tadpole cancellation
    \\
    \cite{SethiVafaWitten96}
  \end{tabular}
  \\
  \hline
\end{tabular}

}

\vspace{.2cm}

 ~~~~~~~~~~~~~~~~{\bf Table 1 --}{\it  Implications of C-field charge quantization in J-twisted Cohomotopy.}

\newpage

\vspace{-2mm}
\paragraph{Organization of the paper.}

\vspace{-1mm}
\begin{list}{--}{}
\vspace{-3mm}
\item In \cref{IntroductionAndSurvey} we
survey our constructions and results.
\vspace{-3mm}
\item  In  \cref{Cohomotopy} we introduce twisted Cohomotopy
  theory, and prove some fundamental facts about it.
\vspace{-3mm}
\item
In \cref{CancellationFromCohomotopy} we use these results
to explains and prove the statements in \hyperlink{Table1}{\it Table 1}.
\vspace{-3mm}
\item
In \cref{Conclusions} we comment on background and implications.
\end{list}

\medskip

\noindent {\bf Generalized abelian cohomology.} Before we start,
we briefly say a word on ``generalized'' cohomology theories,
recalling some basics, but in a broader perspective:
The \emph{ordinary cohomology groups}
$X \mapsto H^\bullet(X,\mathbb{Z})$
famously satisfy a list of nice properties, called the
\emph{Eilenberg-Steenrod axioms}.
Dropping just one of these axioms
(the \emph{dimension axiom})
yields a larger class of possible abelian group assignments
$X \mapsto E^\bullet(X)$, often called
\emph{generalized cohomology theories}. One example are
the \emph{complex topological K-theory groups}
$X \mapsto \mathrm{KU}^\bullet(X)$.

\medskip
By the \emph{Brown representability theorem}, every
generalized cohomology theory in this sense has a
\emph{classifying space} $E_n$ for each degree, such that
the $n$-th cohomology group is equivalently the set
of homotopy classes of maps into this space:
\footnote{Here and in the following, a dashed arrow indicates a map
representing a cocycle that can be freely choosen, as opposed to
solid arrows indicating fixed structure maps.}

\vspace{-.6cm}

\begin{equation}
  \label{StableCohomologyViaClassifyingSpaces}
  \;\;\;\;\;\;\;\;\;\;\;\;\;\;\;\;\;\;\;\;\;\;\;
  \mbox{
    \tiny
    \color{blue}
    \begin{tabular}{c}
      Generalized abelian
      \\
      cohomology theory
    \end{tabular}
  }
  \;\;\;
  E^n(X)
  \;\;
  \overset{
    \mathclap{
    \mbox{
      \tiny
      \color{blue}
      \begin{tabular}{c}
        Brown's
        \\
        representability
        \\
        theorem
        \\
        $\phantom{a}$
      \end{tabular}
    }
    }
  }{
    \simeq
  }
  \;\;
  \Big\{
    \xymatrix{
      X
      \ar@{-->}[rrr]^-{
        \mbox{
          \tiny
          \color{blue}
          continuous function
        }
      }_-{
        \mbox{
          \tiny
          \color{blue}
          = cocycle in $E$-theory
        }
      }
      &&&
      E_n
    }
  \Big\}_{\!\!\!\big/\sim_{{}_{\mathrm{homotopy}}}}
    \mbox{
  }\!.
 \end{equation}
For example, ordinary cohomology theory has as classifying spaces
the \emph{Eilenberg-MacLane spaces} $K(\mathbb{Z},n)$,
while complex topological K-theory in degree 1 is classified by the
space underlying the stable unitary group.

\medskip
For generalized cohomology theories in this sense of
Eilenberg-Steenrod, Brown's representability theorem
translates the \emph{suspension axiom}
into the statement that
the classifying spaces $E_n$ in \eqref{StableCohomologyViaClassifyingSpaces}
are loop spaces of each other,
$E_n \simeq \Omega E_{n + 1}$,
and thus organize into a sequence of classifying spaces
$(E_n)_{n \in \mathbb{N}}$ called a \emph{spectrum}.
The fact that each space in a spectrum is thereby an
infinite loop space makes it behave like a
homotopical \emph{abelian} group
(since higher-dimensional loops may
be homotoped and hence commuted around each other,
by the Eckmann-Hilton argument).

\medskip

\noindent {\bf Generalized non-abelian cohomology.}
We highlight the fact that not all cohomology theories are abelian.
The classical example, for $G$ any non-abelian Lie group,
is the \emph{first non-abelian cohomology}
$X \mapsto H^1\big(X, G\big)$, defined on
any manifold $X$ as the first  {\v C}ech cohomology of $X$ with coefficients
in the sheaf of $G$-valued functions.
Nevertheless, this non-abelian cohomology theory also
has a classifying space, called $B G$, and in terms of this
it is given exactly as the abelian generalized cohomology theories
in \eqref{StableCohomologyViaClassifyingSpaces}:
\begin{equation}
  \label{NonabelianCohomologyViaClassifyingSpaces}
  \mbox{
    \tiny
    \color{blue}
    \begin{tabular}{c}
      Degree-1 non-abelian
      \\
      cohomology theory
    \end{tabular}
  }
  \;\;\;
  H^1(X,G)
  \;\;
  \overset{
    \mathclap{
    \mbox{
      \tiny
      \color{blue}
      \begin{tabular}{c}
        principal bundle
        \\
        theory
      \end{tabular}
    }
    }
  }{\;\;
    \simeq
  }
  \;\;\;\;
  \Big\{
    \xymatrix{
      X
      \ar@{-->}[rrr]^-{
        \mbox{
          \tiny
          \color{blue}
          continuous function
        }
      }_-{
        \mbox{
          \tiny
          \color{blue}
          = cocycle
        }
      }
      &&&
      B G
    }
  \Big\}_{\!\!\!\big/\sim_{{}_{\mathrm{homotopy}}}}
  \,.
\end{equation}
Hence the joint generalization of generalized abelian cohomology theory
\eqref{StableCohomologyViaClassifyingSpaces} and non-abelian 1-cohomology
theories \eqref{NonabelianCohomologyViaClassifyingSpaces} are assignments
of homotopy classes of maps into \emph{any} coefficient space $A$
\begin{equation}
  \label{NonAbelianGeneralizedCohomology}
  \mbox{
    \tiny
    \color{blue}
    \begin{tabular}{c}
      Non-abelian generalized
      \\
      cohomology theory
    \end{tabular}
  }
  \;\;\;
  H(X,A)
  \;\;:=\;\;
  \Big\{
    \xymatrix{
      X
      \ar@{-->}[rrr]^-{
        \mbox{
          \tiny
          \color{blue}
          continuous function
        }
      }_-{
        \mbox{
          \tiny
          \color{blue}
          = cocycle
        }
      }
      &&&
      A
    }
  \Big\}_{\!\!\!\big/\sim_{{}_{\mathrm{homotopy}}}}
  \,.
\end{equation}
All this may naturally be further generalized from topological spaces
to higher stacks.  In the literature of this
broader context the perspective of
non-abelian generalized cohomology is more familiar.
But it applies to the topological situation as the easiest
special case, and this is the case with which we are concerned for the
present purpose.

\medskip

\noindent {\bf Higher principal bundles.}
This way, the classical statement \eqref{NonabelianCohomologyViaClassifyingSpaces}
of principal bundle theory finds the following elegant homotopy-theoretic generalization.
For every \emph{connected} space $A$, its based loop space $G := \Omega A$ is a higher
homotopical group under concatenation of loops (an ``$\infty$-group''). Moreover,
$A$ itself is equivalently the classifying space for that higher group:
\vspace{-2mm}
\begin{equation}
  \label{SpaceIsClassifyingSpaceOfItsLoopGroup}
  \overset{
    \mathclap{
    \mbox{
      \tiny
      \color{blue}
      \begin{tabular}{c}
        \emph{Every}
        \\
        connected
        \\
        space...
        \\
        $\phantom{a}$
      \end{tabular}
    }
    }
  }{
    A
  }
  \;\;\;\simeq\;\;\;
  \overset{
    \mathclap{
    \mbox{
      \tiny
      \color{blue}
      \begin{tabular}{c}
        ...is the
        \\
        classifying
        \\
        space...
        \\
        $\phantom{a}$
      \end{tabular}
    }
    }
  }{
    B
  }
  \;
  \underset{
    \mathclap{
    \mbox{
      \tiny
      \color{blue}
      \begin{tabular}{c}
        ...of its
        \\
        loop group.
      \end{tabular}
    }
    }
  }{
    \;\;
    \overset{G}{
      \overbrace{\Omega A}
    }
    \;
  }
\end{equation}

\vspace{-1mm}
\noindent in that non-abelian $G$-cohomology in degree 1 classifies
higher homotopical $G$-principal bundles:
\vspace{-1mm}
\begin{equation}
  \label{ClassificationOfGBundles}
  \begin{array}{ccc}
    \mathllap{
      H(X, BG)
      \;=\;
    }
    \overset{
      \mbox{
        \tiny
        \color{blue}
        \begin{tabular}{c}
          non-abelian
          \\
          $G$-cohomology
        \end{tabular}
      }
    }{
      H^1(X, G)
    }
    &
    \xrightarrow{\;\;\;\;\;\simeq\;\;\;\;\;}
    &
    \overset{
      \mathclap{
      \mbox{
        \tiny
        \color{blue}
        \begin{tabular}{c}
          higher homotopical
          \\
          $G$-principal bundles
        \end{tabular}
      }
      }
    }{
      G \mathrm{Bundles}(X)_{\!\!\!\big/ \sim}
    }
    \\
        \big[
      X
        \overset{
          \overset{
            \mbox{
              \tiny
              \color{blue}
              cocycle
            }
          }{c}
        }{\longrightarrow}
      B G
    \big]
    &\longmapsto&
    \left[
    \raisebox{19pt}{
    \;\;\;\;\;
    \xymatrix@C=3.5em@R=1.4em{
      \overset{
        \mathclap{
        \mbox{
          \tiny
          \color{blue}
          \begin{tabular}{c}
            $G$-principal bundle
            \\
            classified
            by $c$
          \end{tabular}
        }
        }
      }{
        P
      }
      \ar[rr]
      \ar@{}[drr]|-{
        \mbox{
          \tiny
          \color{blue}
          \begin{tabular}{c}
            homotopy
            \\
            pullback
          \end{tabular}
        }
      }
      \ar[d]|-{
        \mathllap{
          c^\ast(p_{{}_{B G}})
        }
      }
      &&
      \overset{
        \mathclap{
        \mbox{
          \tiny
          \color{blue}
          \begin{tabular}{c}
            universal
          \\
            $G$-principal
                        bundle
          \end{tabular}
        }
        }
      }{
        G \!\sslash\! G
      }
      \ar[d]^{ p_{{}_{B G}} }
      \\
      X
      \ar[rr]_-{
        \underset{
          \mathclap{
          \mbox{
           \tiny
           \color{blue}
           \begin{tabular}{c}
             classifying map
                         for $P$
           \end{tabular}
          }
          }
        }{c}
      }
      &&
      B G
    }
    }
    \;\;
    \right]
  \end{array}
\end{equation}

\noindent {\bf Cohomotopy cohomology theory.}
The primordial example of a non-abelian generalized
cohomology theory \eqref{NonAbelianGeneralizedCohomology} is
\emph{Cohomotopy cohomology theory},
denoted $\pi^\bullet$. By definition, its classifying spaces
are simply the $n$-spheres $S^n$:
\begin{equation}
  \label{IntroCohomotopy}
  \mbox{
    \tiny
    \color{blue}
    \begin{tabular}{c}
      Cohomotopy
      \\
      cohomology theory
    \end{tabular}
  }
  \;\;\;
  \pi^n(X)
  \;\;:=\;\;
  \Big\{
    \xymatrix{
      X
      \ar@{-->}[rrr]^-{
        \mbox{
          \tiny
          \color{blue}
          continuous function
        }
      }_-{
        \mbox{
          \tiny
          \color{blue}
          = cocycle
        }
      }
      &&&
      S^n
    }
  \Big\}_{\!\!\!\big/\sim_{{}_{\mathrm{homotopy}}}}
  \,.
\end{equation}
Since the $(n \geq 1)$-spheres are connected,
the equivalence \eqref{SpaceIsClassifyingSpaceOfItsLoopGroup}
applies and says that Cohomotopy theory is equivalently
non-abelian 1-cohomology for the loop groups
of spheres $G := \Omega S^n$:
$$
  \underset{
    \mathclap{
    \mbox{
      \tiny
      \color{blue}
      \begin{tabular}{c}
        Cohomotopy
        \\
        theory
      \end{tabular}
    }
    }
  }{
    \pi^n(X)
  }
  \;\;\simeq\;\;
  \underset{
    \mathclap{
    \mbox{
      \tiny
      \color{blue}
      \begin{tabular}{c}
        non-abelian 1-cohomology
        \\
        for sphere loop group
      \end{tabular}
    }
    }
  }{
    H^1(X, \Omega S^n)
  }
  \,.
$$
Evaluated on spaces which are themselves spheres, Cohomotopy cohomology
theory gives the (unstable!)
\emph{homotopy groups of spheres}, the ``vanishing point'' of algebraic topology:
$$
  \mbox{
    \tiny
    \color{blue}
    \begin{tabular}{c}
      $n$-cohomotopy groups
      \\
      of $k$-sphere
    \end{tabular}
  }
  \;\;\;
  \pi^n\big( S^k \big)
  \;\;\;\simeq\;\;\;
  \big\{
    \xymatrix{
      S^k
      \ar@{-->}[r]
      &
      S^n
    }
  \big\}_{\mathrlap{\!\!\!\big/\sim}}
  \;\;\;\simeq\;\;\;
  \pi_k\big( S^n \big)
  \;\;\;
  \mbox{
    \tiny
    \color{blue}
    \begin{tabular}{c}
      $k$-homotopy groups
      \\
      of $n$-sphere
    \end{tabular}
  }
  $$
A whole range of classical theorems in differential topology
all revolve around characterizations of Cohomotopy sets,
even if this is not often fully brought out in the terminology.

\medskip
\noindent {\bf The quaternionic Hopf fibration.}
A notable example, for the following purpose,
of a class in the Cohomotopy group of spheres,
is given by the
\emph{quaternionic Hopf fibration}
\begin{equation}
  \xymatrix@C=3.5em{
    S^7
    \ar@/^1.6pc/[rrrrr]^-{
      \overset{
       \mbox{
         \tiny
         \color{blue}
         \begin{tabular}{c}
           quaternionic Hopf fibration
         \end{tabular}
       }
      }{
        h_{\mathbb{H}}
      }
    }
    \ar@{}[r]|-{\simeq}
    &
    \underset{
      \mathclap{
      \mbox{
        \tiny
        \color{blue}
        \begin{tabular}{c}
          unit sphere
          \\
          in quaternionic
          \\
          2-space
        \end{tabular}
      }
      }
    }{
      S(\mathbb{H}^2)
    }
    \ar[rrr]_-{
      (q_1,q_2)
     \, \mapsto \,
      [q_1 : q_2]
    }
    &&&
    \underset{
      \mathclap{
      \mbox{
        \tiny
        \color{blue}
        \begin{tabular}{c}
          quaternionic
          \\
          projective
          \\
          1-space
        \end{tabular}
      }
      }
    }{
      \mathbb{H}P^1
    }
    \ar@{}[r]|-{\simeq}
    &
    S^4
  }
  \,,
\end{equation}
which represents a
generator of the non-torsion subgroup in the 4-Cohomotopy of
the 7-sphere, as shown on the left here:
\begin{equation}
  \label{QuaternionicHopfFibration}
  \xymatrix@R=7pt{
    \mbox{
      \tiny
      \color{blue}
      \begin{tabular}{c}
        $\phantom{a}$
        \\
        quaternionic
        \\
        Hopf fibration
      \end{tabular}
    }
    &
    \mathclap{
      [S^7 \overset{h_{\mathbb{H}}}{\to} S^4]
    }
    \ar@{}[d]|-{
      \mbox{
        \begin{rotate}{-90}
          $\!\!\!\!\!\mapsto$
        \end{rotate}
      }
    }
    \ar@{}[r]|>>>>{\in}
    &
    \overset
    {
      \mathclap{
      \mbox{
        \tiny
        \color{blue}
        \begin{tabular}{c}
          non-abelian/unstable
          \\
          Cohomotopy group
          \\
          $\phantom{a}$
        \end{tabular}
      }
      }
    }
    {
      \pi^4(S^7)
    }
    \ar@{}[d]|-{
      \mbox{
        \begin{rotate}{-90}
          $\!\!\!\simeq$
        \end{rotate}
      }
    }
    \ar[rrr]^-{
      \overset{
        \mbox{
          \tiny
          \color{blue}
          \begin{tabular}{c}
            stabilization
            \\
            $\phantom{a}$
          \end{tabular}
        }
      }{
        \Sigma^\infty
      }
    }
    &&&
    \overset{
      \mathclap{
      \mbox{
        \tiny
        \color{blue}
        \begin{tabular}{c}
          abelian/stable
          \\
          Cohomotopy group
          \\
          $\phantom{a}$
        \end{tabular}
      }
      }
    }{
      \;\mathbb{S}^4(S^7)\;
    }
    \ar@{}[d]|-{
      \mbox{
        \begin{rotate}{-90}
          $\!\!\!\simeq$
        \end{rotate}
      }
    }
    \ar@{}[r]|<{\ni}
    &
    \mathclap{
      \;\;
      \Sigma^\infty[S^7 \overset{h_{\mathbb{H}}}{\to} S^4]
    }
    \ar@{}[d]|-{
      \!\!\!\!\!
      \mbox{
        \begin{rotate}{-90}
          $\!\!\!\!\!\mapsto$
        \end{rotate}
      }
    }
    &
    \mbox{
      \tiny
      \color{blue}
      \begin{tabular}{c}
        stabilized
        \\
        quaternionic
        \\
        Hopf fibration
      \end{tabular}
    }
    \\
    \mbox{
      \tiny
      \color{blue}
      \begin{tabular}{c}
        non-torsion
        \\
        generator
      \end{tabular}
    }
    &
    (1,0)
    \ar@{}[r]|>>>{ \in }
     &
    \mathbb{Z} \times \mathbb{Z}_{12}
    \ar[rrr]_-{ (n,a) \, \mapsto \,  (n \,\,\mathrm{mod}\,\, {\color{magenta}24}) }
    &&&
    \mathbb{Z}_{\color{magenta}24}
    \ar@{}[r]|<<<<{\ni}
    &
    1
    &
    \mbox{
      \tiny
      \color{blue}
      \begin{tabular}{c}
        torsion
        \\
        generator
      \end{tabular}
    }
  }
\end{equation}
Shown on the right is the abelian approximation to non-abelian Cohomotopy
cohomology theory, called \emph{stable Cohomotopy theory}
and represented, via \eqref{StableCohomologyViaClassifyingSpaces},
by the \emph{sphere spectrum} $\mathbb{S}$,  whose component spaces
are the infinite-loop space completions of the $n$-spheres:
$\mathbb{S}_n \simeq \Omega^\infty \Sigma^\infty S^n$.
Crucially, in this approximation, the quaternionic Hopf fibration becomes
a torsion generator: non-abelian 4-Cohomotopy witnesses integer cohomology groups
not only in degree 4, but also in degree 7; but when seen in the abelian/stable
approximation this ``extra degree'' fades away and leaves only
a torsion shadow behind.
From the perspective, composition with the quaternionic Hopf fibration
can be viewed as a transformation that translates classes in degree-7 Cohomotopy
to classes in degree-4 Cohomotopy:
\begin{equation}
  \label{hHSendsDegree7ToDegree4}
  \raisebox{20pt}{
  \xymatrix{
    &&
    S^7
    \ar[d]^-{ h_{\mathbb{H}} }
    &
    \mbox{
      \tiny
      \color{blue}
      7-Cohomotopy
    }
    \ar@{}[d]|-{
      \mbox{
        \tiny
        \color{blue}
        reflects into
      }
    }
    &
    \pi^7(X)
    \ar[d]^-{ (h_{\mathbb{H}})_\ast   }
    \\
    X
    \ar@{-->}[urr]^-{c}
    \ar@{-->}[rr]_-{  (h_{\mathbb{H}})_\ast(c) }
    &&
    S^4
    &
    \mbox{
      \tiny
      \color{blue}
      4-Cohomotopy
    }
    &
    \pi^4(X)
  }
  }
\end{equation}

\medskip

\noindent {\bf Twisted non-abelian generalized cohomology.}
Regarding generalized cohomology theory as homotopy theory of
classifying spaces \eqref{NonAbelianGeneralizedCohomology}
makes  transparent the concept of \emph{twistings} in cohomology theory:
Instead of mapping into a fixed classifying spaces, a
\emph{twisted cocycle} maps into a varying classifying space that may twist and
turn as one moves in the domain space.
In other words, a \emph{twisting} $\tau$ of $A$-cohomology theory
on some $X$ is a bundle over $X$ with typical fiber $A$,
and a $\tau$-twisted cocycle is a \emph{section} of that bundle:
\begin{equation}
  \label{TwistedNonAbelianGeneralizedCohomology}
\hspace{-1cm}
  \begin{aligned}
  \mbox{
    \tiny
    \color{blue}
    \begin{tabular}{c}
      $\tau$-twisted
      \\
      non-abelian generalized
      \\
      $A$-cohomology theory
    \end{tabular}
  }
  \;
  A^\tau(X)
  & \;:=\;
  \left\{\!\!\!\!\!\!\!
    \raisebox{22pt}{
    \xymatrix@C=3em@R=1.5em{
      &
      \overset{
        \mathclap{
        \mbox{
          \tiny
          \color{blue}
          \begin{tabular}{c}
            $\phantom{a}$
            \\
            $A$-fiber bundle
          \end{tabular}
        }
        }
      }{
        P
      }
      \ar[d]^-p
      \ar[rr]
      &&
      \overset{
        \mathclap{
        \mbox{
          \tiny
          \color{blue}
          \begin{tabular}{c}
            universal
            \\
            $A$-fiber bundle
          \end{tabular}
        }
        }
      }{
      A \!\!\sslash\! \mathrm{Aut}(A)
      }
      \ar[d]
      \\
      X
      \ar@/^1.04pc/@{-->}[ur]^{
        \mbox{
          \tiny
          \color{blue}
          \begin{tabular}{c}
            continuous
            section
            \\
            =
            twisted cocycle
          \end{tabular}
        }
      }
      \ar@{=}[r]
      &
      X
      \ar[rr]^-{\tau}_-{
               \mathclap{
          \mbox{
            \tiny
            \color{blue}
            \begin{tabular}{c}
              classifying map
                            for $P$
            \end{tabular}
          }
          }
        }
      &&
      B \mathrm{Aut}(A)
    }
    }
  \right\}_{\!\!\!\!
       \Big/
    \sim_{
      {}_{
        \frac{
          \mathrm{homotopy}
        }
        {
          B \mathrm{Aut}(A)
        }
      }
    }
  }
  \\
  & \;\simeq\;
  \left\{
    \raisebox{22pt}{
    \xymatrix@C=4em{
      X
      \ar[dr]_-{
        \mathllap{
          \mbox{
            \tiny
            \color{blue}
            twist
          }
        }
        \;
        \tau
      }^>>>>{\ }="t"
      \ar@{-->}[rr]^-{
        \mbox{
          \tiny
          \color{blue}
          continuous function
        }
      }_-{\ }="s"
      &&
      A \!\!\sslash\! \mathrm{Aut}(A)
      \ar[dl]
      \\
      &
      B \mathrm{Aut}(A)
      \ar@{=>}^{
        \mbox{
          \begin{rotate}{48}
            \!\!\!\!\!\!\!\!\!
            \mbox{
              \tiny
              \color{blue}
              homotopy
            }
          \end{rotate}
        }
      } "s"; "t"
    }
    }
  \right\}_{
    \!\!\!\!
    \Big/
    \sim_{
      {}_{
        \frac{
          \mathrm{homotopy}
        }
        {
          B \mathrm{Aut}(A)
        }
      }
    }
  }
  \end{aligned}
  \end{equation}
Here the equivalent formulation shown in the second line
follows because $A$-fiber bundles are themselves classified by
nonabelian $\mathrm{Aut}(A)$-cohomology,
as shown on the right of the first line (due to \eqref{ClassificationOfGBundles}).

\medskip

\noindent{\bf Twisted Cohomotopy theory.}
For the example \eqref{IntroCohomotopy} of Cohomotopy cohomology theory
in degree $d-1$ there is a canonical twisting on
Riemannian $d$-manifolds, given by the unit sphere bundle in the orthogonal tangent bundle:

\begin{equation}
  \label{JTwistedCohomotopy}
  \hspace{-1cm}
  \begin{aligned}
  \mbox{
    \tiny
    \color{blue}
    \begin{tabular}{c}
      J-twisted
      \\
      Cohomotopy theory
    \end{tabular}
  }
  \;
  \pi^{{}^{T X^d}}(X^d)
  & \;:=\;
  \left\{\!\!\!\!\!\!\!\!
    \raisebox{20pt}{
    \xymatrix@C=3em@R=1.3em{
      &
      \overset{
        \mathclap{
        \mbox{
          \tiny
          \color{blue}
          \begin{tabular}{c}
            tangent
            \\
            unit sphere bundle
          \end{tabular}
        }
        }
      }{
        S(T X^d)
      }
      \ar[d]^-p
      \ar[rr]
      &&
      \overset{
        \mathclap{
        \mbox{
          \tiny
          \color{blue}
          \begin{tabular}{c}
            universal tangent
            \\
            unit sphere bundle
          \end{tabular}
        }
        }
      }{
      S^{d-1}
        \!\!\sslash\!
      \mathrm{O}(d)
      }
      \ar[d]
      \\
      X
      \ar@/^1.04pc/@{-->}[ur]^-{
        \mbox{
          \tiny
          \color{blue}
          \begin{tabular}{c}
            continuous
            section
            \\
            =
            twisted cocycle
          \end{tabular}
        }
      }
      \ar@{=}[r]
      &
      X
      \ar[rr]^-{T X^d}_-{
                \mathclap{
          \mbox{
            \tiny
            \color{blue}
            \begin{tabular}{c}
              classifying map of
              \\
              tangent/frame
              bundle
            \end{tabular}
          }
          }
        }
      &&
      B \mathrm{O}(d)
    }
    }
  \right\}_{
    \!\!\!\!
    \Big/
    \sim_{
      {}_{
        \frac{
          \mathrm{homotopy}
        }
        {
          B \mathrm{O}(d)
        }
      }
    }
  }
  \\
 & \;\simeq\;
 \left\{
    \;\;
    \raisebox{20pt}{
    \xymatrix@C=4em{
      X
      \ar[dr]_-{
        \underset{
          \mbox{
            \tiny
            \color{blue}
            \begin{tabular}{c}
              twist
            \end{tabular}
          }
        }{
          T X^{\mathrlap{d}}
        }
      }^>>>>{\ }="t"
      \ar@{-->}[rr]^-{
        \mbox{
          \tiny
          \color{blue}
          continuous function
        }
      }_-{\ }="s"
      &&
      S^{d-1}
        \!\!\sslash\!
      \mathrm{O}(d)
      \ar[dl]
      \\
      &
      B \mathrm{O}(d)
      \ar@{=>}^{
        \mbox{
          \begin{rotate}{48}
            \!\!\!\!\!\!\!\!\!
            \mbox{
              \tiny
              \color{blue}
              homotopy
            }
          \end{rotate}
        }
      } "s"; "t"
    }
    }
  \right\}_{
    \!\!\!
    \Big/
    \sim_{
      {}_{
        \frac{
          \mathrm{homotopy}
        }
        {
          B \mathrm{O}(d)
        }
      }
    }
  }
  \end{aligned}
\end{equation}
Since the canonical morphism $\mathrm{O}(d) \longrightarrow \mathrm{Aut}(S^{d-1})$
is known as the \emph{J-homomorphism}, we may call this
\emph{J-twisted Cohomotopy theory}, for short.

\medskip

\noindent {\bf Compatibly J-twisted Cohomotopy in degrees 4 \& 7.}
In view of \eqref{hHSendsDegree7ToDegree4}
it is natural to ask for the maximal subgroup $G \subset \mathrm{O}(8)$
for which the quaternionic Hopf fibration is equivariant,
so that its homotopy quotient $h_{\mathbb{H}} \sslash G$
exists and serves as a map of $G$-\emph{twisted} Cohomotopy theories
\eqref{JTwistedCohomotopy} from degree 7 and 4.
This subgroup turns out to be the central product of the quaternion unitary
groups $\mathrm{Sp}(n)$ for $n = 1,2$:

\vspace{-.5cm}

\begin{equation}
  \label{SP1Sp2TwistingOfhH}
  \hspace{-3cm}
  \underset{
    \mathclap{
    \mbox{
      \tiny
      \color{blue}
      \begin{tabular}{c}
        central product of
        \\
        quaternion-unitary groups
      \end{tabular}
    }
    }
  }{
    \mathrm{Sp}(2)
    \cdot
    \mathrm{Sp}(1)
  }
  \;\subset\;
  \mathrm{O}(8)
  \;\;\;
  \mbox{is maximal subgroup s.t. }
  \;\;
  \raisebox{44pt}{
  \xymatrix@C=4pt{
    &
    S^7 \!\!\sslash\! \mathrm{Sp}(2) \cdot \mathrm{Sp}(1)
    \ar[dl]
    \ar[dd]|-{
      \;\;\;\;\;\;\;\;
      {
        \color{brown}
        h_{\mathbb{H}} \sslash \mathrm{Sp}(2)\cdot\mathrm{Sp}(1)
      }
      \mathrlap{
      \!\!\!
      \mbox{
        \tiny
        \color{blue}
        \begin{tabular}{c}
          universally twisted
          \\
          quaternionic Hopf fibration
        \end{tabular}
      }
      }
    }
    \\
    B
    \big(
      \mathrm{Sp}(2) \cdot \mathrm{Sp}(1)
    \big)
    \\
    &
    S^4 \!\!\sslash\! \mathrm{Sp}(2)\cdot \mathrm{Sp}(1)
    \ar[ul]
  }
  }
 \end{equation}
In other words, J-twisted Cohomotopy \eqref{JTwistedCohomotopy}
exists compatibly in degrees 4 \& 7 precisely on those 8-manifolds which carry
topological $\mathrm{Sp}(2)\cdot \mathrm{Sp}(1)$-structure, i.e., whose structure
group of the tangent bundle is equipped with a reduction along
$\mathrm{Sp}(2)\cdot \mathrm{Sp}(1)\hookrightarrow \mathrm{O}(8)$.
This reduction is equivalent to a factorization of the classifying map
as shown on the left below, with some cohomological consequences
shown on the right:
\begin{equation}
  \label{SP1Sp2Structure}
  \hspace{-1cm}
  \raisebox{36pt}{
  \xymatrix@R=14pt@C=40pt{
    \overset{
      \mathclap{
      \mbox{
        \tiny
        \color{blue}
        \begin{tabular}{c}
        \end{tabular}
        \;\;\;\;\;\;\;\;\;\;\;\;\;\;\;
        }
      }
    }{
      X^8
    }
    \ar@{-->}[dr]^-{\tau}_<<<<{\ }="s2"
    \ar[d]_{
      \mathllap{
        \mbox{
          \tiny
          \color{blue}
          \begin{tabular}{c}
            tangent
            \\
            bundle
          \end{tabular}
        }
      }
      T X^8
    }^>>>>{\ }="t2"
    \\
    \underset{
      \mathclap{
      \mbox{
        \tiny
        \color{blue}
        \begin{tabular}{c}
          $\phantom{a}$
          \\
          classifying space of
          \\
          orthogonal structure
        \end{tabular}
      }
      }
      \;\;\;\;
    }{
      B \mathrm{O}(8)
    }
    &
    \underset{
      \mathclap{
      \;\;\;
      \mbox{
        \tiny
        \color{blue}
        \begin{tabular}{c}
          $\phantom{a}$
          \\
          classifying space of
          \\
          $\mathrm{Sp}(2)\cdot\mathrm{Sp}(1)$-twists
        \end{tabular}
      }
      }
    }{
      B \big( \mathrm{Sp}(2)\cdot \mathrm{Sp}(1) \big)
    }
    \ar[l]
    \ar@{=>}^-{
      \begin{rotate}{-30}
      \!\!\!\!
      \mbox{
        \tiny
        \color{brown}
        \begin{tabular}{c}
          $\mathrm{Sp}(2)\cdot\mathrm{Sp}(1)$
          \\
          -structure
        \end{tabular}
      }
      \end{rotate}
    }
      "s2"; "t2"
  }
  }
  \;\;\;
    \Longrightarrow
  \;\;\;
  \left\{
  \begin{array}{l}
    \tfrac{1}{\color{magenta}24}
    \overset{
      \mathclap{
      \mbox{
        \tiny
        \color{blue}
        \begin{tabular}{c}
          Euler
          \\
          class
          \\
          $\phantom{a}$
        \end{tabular}
      }
      }
    }{
      \rchi_8
    }
      \;=\;
      I_8
      \;:=\;
      \tfrac{1}{48}
      \overset{
        \!\!
        \mbox{
          \tiny
          \color{blue}
          \begin{tabular}{c}
            Pontrjagin
            classes
            \\
            $\phantom{a}$
s          \end{tabular}
        }
      }{
      \big(
        p_2
        -
        \tfrac{1}{4} (p_1)^2
      \big)
      }
      \\
      $\phantom{a}$
      \\
      \big(
        H^2(X^8, \mathbb{Z}_2) = 0
      \big)
      \;\Rightarrow\;
     \;\; (
        \underset{
          \mathclap{
          \mbox{
            \tiny
            \color{blue}
            \begin{tabular}{cc}
              $\phantom{a}$ &
              \\
              $\phantom{a}$ &
            \\
 &                Stiefel-Whitney
              \\
   &           class
            \end{tabular}
          }
          }
        }{
          w_6
        }
        = 0
      )
      \;\Rightarrow\;
      (
        \underset{
          \mathclap{
          \mbox{
            \tiny
            \color{blue}
            \begin{tabular}{cc}
              $\phantom{a}$ &
              \\
      &        integral
              \\
        &      Stiefel-Whitney
              \\
        &      class
            \end{tabular}
          }
          }
        }{
          W_7
        }
        = 0
      )
  \end{array}
  \right.
\end{equation}

\medskip

\noindent {\bf J-Twisted Cohomotopy and Topological $G$-Structure.}
For every topological coset space realization $G/H$ of an $n$-sphere,
there is a canonical homotopy equivalence between
the classifying spaces for $G$-twisted Cohomotopy
and for topological $H$-structure
(i.e., reduction of the structure group to $H$), as follows:
$$
  \overset{
    \mbox{
      \tiny
      \color{blue}
      \begin{tabular}{c}
        coset space structure
        \\
        on topological $n$-sphere
        \\
        $\phantom{a}$
      \end{tabular}
    }
  }{
    S^n
    \;
      \underset{
        \mbox{\tiny homeo}
      }{
        \simeq
      }
    \;
    G/H
  }
  \;\;\;\;
  \Rightarrow
  \;\;\;\;
  \overset{
    \mbox{
      \tiny
      \color{blue}
      \begin{tabular}{c}
        $G$-twisted Cohomotopy /
        \\
        topological $H$-structure
        \\
        $\phantom{a}$
      \end{tabular}
    }
  }{
    S^n \!\!\sslash\! G
    \;
      \underset{
        \mbox{\tiny htpy}
      }{
        \simeq
      }
    \;
    B H
  }
 \!\!\!\! .
$$
(One may think of this as ``moving $G$ from numerator on the right to denominator on the left''.)

In particular, on Spin 8-manifolds we have the following equivalences
between J-twisted Cohomotopy cocycles \eqref{JTwistedCohomotopy}
and topological $G$-structures:
\begin{equation}
  \label{CohomotopyAndSpin7}
  \hspace{-5mm}
  \begin{array}{r}
    S^7 \!\!\sslash \!\mathrm{Spin}(8)
    \\
    \simeq
    \;
    B \, \mathrm{Spin}(7)
  \end{array}
  \;\;\;\;\;\;\;
  \Longrightarrow
  \;\;\;\;\;\;\;
  \left\{
    \raisebox{31pt}{
  \xymatrix@C=5em{
    &
    \overset{
      \mathclap{
      \mbox{
       \tiny
       \color{blue}
       \begin{tabular}{c}
         classifying space
         \\
         for J-twisted
         \\
         Cohomotopy theory
         \\
         $\phantom{a}$
       \end{tabular}
      }
      }
    }{
      S^7 \!\!\sslash\! \mathrm{Spin}(8)
    }
    \ar[d]^-{}
    \\
    X^8
    \ar@/^1.2pc/@{-->}[ur]^{\!\!\!\!\!\!\!
      \overset{
        \mathllap{
        \mbox{
          \tiny
          \color{brown}
          \begin{tabular}{c}
            cocycle in
            \\
            J-twisted Cohomotopy
          \end{tabular}
        }
        \;\;\;
        }
      }{
        c
      }
    }_>>{\ }="s"
   \ar[r]_-{
     \underset{
       \mathclap{
       \mbox{
         \tiny
         \color{blue}
         \begin{tabular}{c}
           $\phantom{a}$
           \\
           tangent
           \\
           spin structure
         \end{tabular}
       }
       }
     }{
       T X^8
     }
   }^>>>{\ }="t"
    &
    B \mathrm{Spin}(8)
    \ar@{=>}|-{
      \mbox{
        \tiny
        \color{blue}
        homotopy
      }
    } "s"; "t"
  }
  }
  \right\}
  \;\simeq\;
  \left\{
  \;\;
  \raisebox{31pt}{
  \xymatrix@C=5em{
    &
    \overset{
      \mathclap{
      \mbox{
       \tiny
       \color{blue}
       \begin{tabular}{c}
         classifying space
         \\
         for topological
         \\
         $\mathrm{Spin}(7)$-structure
         \\
         $\phantom{a}$
       \end{tabular}
      }
      }
    }{
      B \mathrm{Spin}(7)
    }
    \ar[d]^-{B i}
    \\
    X^8
    \ar@/^1.2pc/@{-->}[ur]^{
      \overset{
        \mathllap{
        \mbox{
          \tiny
          \color{brown}
          \begin{tabular}{c}
            topological
            \\
            $\mathrm{Spin}(7)$-structure
          \end{tabular}
        }
        \!\!
        }
      }{g}
    }_>>{\ }="s"
   \ar[r]_-{
     \underset{
       \mathclap{
       \mbox{
         \tiny
         \color{blue}
         \begin{tabular}{c}
           $\phantom{a}$
           \\
           tangent
           \\
           spin structure
         \end{tabular}
       }
       }
     }{
       T X^8
     }
   }^>>{\ }="t"
    &
    B \mathrm{Spin}(8)
    \ar@{=>}|-{
      \mbox{
        \tiny
        \color{blue}
        homotopy
      }
    } "s"; "t"
  }
  }
  \right\}
\end{equation}
and
\begin{equation}
  \label{CohomotopyAndSp2Sp1}
   \begin{array}{r}
    S^7 \!\!\sslash \!\mathrm{Sp}(2)\cdot \mathrm{Sp}(1)
    \\
    \simeq
    \;
    B \, \mathrm{Sp}(1)\cdot \mathrm{Sp}(1)
  \end{array}
  \;\;\;\;\;
  \Longrightarrow
  \;\;\;\;\;
  \left\{
  \raisebox{31pt}{
  \xymatrix@C=4em{
    &
    \overset{
      \mathclap{
      \mbox{
       \tiny
       \color{blue}
       \begin{tabular}{c}
         classifying space
         \\
         for $\mathrm{Sp}(2)\cdot\mathrm{Sp}(1)$-twisted
         \\
         Cohomotopy theory
         \\
         $\phantom{a}$
       \end{tabular}
      }
      }
    }{
      S^7 \!\!\sslash\! \mathrm{Sp}(2)\cdot \mathrm{Sp}(1)
    }
    \ar[d]^-{}
    \\
    X^8
    \ar@/^1.2pc/@{-->}[ur]^{
      \overset{
        \mathllap{
        \mbox{
          \tiny
          \color{brown}
          \begin{tabular}{c}
            cocycle in
            \\
            $\mathrm{Sp}(2)\cdot \mathrm{Sp}(1)$-twisted
            \\
            Cohomotopy theory
          \end{tabular}
        }
        \;\;\;\;\;\;\;\;
        }
      }{
        c
      }
    }_>>{\ }="s"
   \ar[r]_-{
     \underset{
       \mathclap{
       \mbox{
         \tiny
         \color{blue}
         \begin{tabular}{c}
           $\phantom{a}$
           \\
           tangent
           \\
           spin structure
         \end{tabular}
       }
       }
     }{
       T X^8
     }
   }^>>>>{\ }="t"
    &
    B \mathrm{Spin}(8)
    \ar@{=>}|-{
      \mbox{
        \tiny
        \color{blue}
        homotopy
      }
    } "s"; "t"
  }
  }
  \right\}
  \;\simeq\;
  \left\{
  \raisebox{31pt}{
  \xymatrix@C=4em{
    &
    \overset{
      \mathclap{
      \mbox{
       \tiny
       \color{blue}
       \begin{tabular}{c}
         classifying space
         \\
         for topological
         \\
         $\mathrm{Sp}(1)\cdot \mathrm{Sp}(1)$-structure
         \\
         $\phantom{a}$
       \end{tabular}
      }
      }
    }{
      B \mathrm{Sp}(1)\cdot \mathrm{Sp}(1)
    }
    \ar[d]^-{B i}
    \\
    X^8
    \ar@/^1.2pc/@{-->}[ur]^{
      \overset{
        \mathllap{
        \mbox{
          \tiny
          \color{brown}
          \begin{tabular}{c}
            topological
            \\
            $\mathrm{Sp}(1)\cdot \mathrm{Sp}(1)$-structure
          \end{tabular}
        }
        \!\!\;\;\;\;\;\;
        }
      }{g}
    }_>>{\ }="s"
   \ar[r]_-{
     \underset{
       \mathclap{
       \mbox{
         \tiny
         \color{blue}
         \begin{tabular}{c}
           $\phantom{a}$
           \\
           tangent
           \\
           spin structure
         \end{tabular}
       }
       }
     }{
       T X^8
     }
   }^>>>>{\ }="t"
    &
    B \mathrm{Spin}(8)
    \ar@{=>}|-{
      \mbox{
        \tiny
        \color{blue}
        homotopy
      }
    } "s"; "t"
  }
  }
  \right\}
\end{equation}

As the existence of a $G$-structure is a non-trivial topological condition,
so is hence the existence of $J$-twisted Cohomotopy cocycles.
Notice that this is
a special effect of twisted non-abelian generalized Cohomology:
A non-twisted generalized cohomology theory (abelian or non-abelian) always
admits at least one cocycle, namely the trivial or zero-cocycle.
But here for non-abelian J-twisted Cohomotopy theory
on 8-manifolds,
the existence of \emph{any} cocycle is a
non-trivial topological condition.

\medskip

\noindent {\bf Compatibly $\mathrm{Sp}(2)$-Twisted Cohomotopy in degree 4 \& 7.}
For focus of the discussion,
we will now restrict attention to $G$-structure for the further
quaternion-unitary subgroup
$$
  \mathrm{Sp}(2)
  \;\longhookrightarrow
  \;
  \mathrm{Sp}(1)\cdot \mathrm{Sp}(2)
$$
in diagram \eqref{SP1Sp2TwistingOfhH}.
In summary then, due to the $\mathrm{Sp}(2)$-equivariance of the quaternionic Hopf fibration
\eqref{SP1Sp2TwistingOfhH}, the map \eqref{hHSendsDegree7ToDegree4}
from degree-7 to degree-4 Cohomotopy passes to
$\mathrm{Sp}(2)$-twisted Cohomotopy:
$$
\hspace{-1cm}
  \xymatrix@C=6em@R=2em{
    &&
    S^7 \!\!\sslash\! \mathrm{Sp}(2)
    \ar[ddl]|-{ {\phantom{{AA} \atop {AA}}} }
    \ar[d]^-{
         {
        \color{brown}
        h_{\mathbb{H}}\sslash \mathrm{Sp}(2)
      }
      \;
      \mathrlap{
      \mbox{
        \tiny
        \color{blue}
        \begin{tabular}{c}
        $\mathrm{Sp}(2)$-twisted
        \\
        quaternionic Hopf fibration
        \end{tabular}
      }
      }
    }
    \\
    X
    \ar@{-->}[rr]|<<<<<<<<<<<<<<<<<<<<<<<<<<<<<{
    \,  (h_{\mathbb{H}}\sslash \mathrm{Sp}(2))_\ast (c) \,
    }
    \ar@{-->}[urr]^-{
      \overset{
        \mathclap{
        \mbox{
          \tiny
          \color{blue}
          \begin{tabular}{c}
            cocycle in
            \\
            twisted
            \\
            7-Cohomotopy
          \end{tabular}
        }
        }
      }{
        c
      }
    }
    \ar[dr]_-{
      \mathllap{
        \mbox{
          \tiny
          \color{blue}
          \begin{tabular}{c}
            twist, uniformly
            \\
            in degrees 4 \& 7
          \end{tabular}
        }
      }
      \tau
    }
    &&
    S^4 \!\!\sslash\! \mathrm{Sp}(2)
    \ar[dl]
    \mathrlap{
      \;\;
      \mbox{
        \tiny
        \color{blue}
        \begin{tabular}{c}
          induced cocycle
          \\
          in twisted
          \\
          4-Cohomotopy
        \end{tabular}
      }
    }
    \\
    &
    B \mathrm{Sp}(2)
  }
$$
and hence \eqref{hHSendsDegree7ToDegree4} becomes:
\begin{equation}
  \label{ReflectTwisted7Cohomotopy}
  \begin{aligned}
  \mbox{
    \tiny
    \color{blue}
    \begin{tabular}{c}
      $\mathrm{Sp}(2)$-twisted
      \\
      7-Cohomotopy
    \end{tabular}
  }
  \;
  \xymatrix{
    \pi^{i_7\circ\tau}(X)
    \ar[ddd]|-{
      \mathllap{
        \mbox{
          \tiny
          \color{blue}
          \begin{tabular}{c}
            reflects into
          \end{tabular}
        }
        \;
      }
    \;   (h_{\mathbb{H}} \sslash \mathrm{Sp}(2))_\ast \;
    }
    \\
    \\
    \\
    \\
  }
  & \;:=\;
  \left\{
    \raisebox{20pt}{
    \xymatrix@C=4em{
      X
      \ar[dr]_-{
        \mathllap{
          \mbox{
            \tiny
            \color{blue}
            twist
          }
        }
        \;
        \tau
      }^>>>>{\ }="t"
      \ar@{-->}[rr]^-{
        \mbox{
          \tiny
          \color{blue}
          continuous function
        }
      }_-{\ }="s"
      &&
      S^7 \!\!\sslash\! \mathrm{Sp}(2)
      \ar[dl]
      \\
      &
      B \mathrm{Sp}(2)
      \ar@{=>}^-{
        \mbox{
          \begin{rotate}{49}
            \!\!\!\!\!\!\!\!\!\!\!
            \mbox{
              \tiny
              \color{blue}{
              homotopy
            }}
          \end{rotate}
        }
      } "s"; "t"
    }
    }
  \right\}_{
    \!\!\!
    \Big/
    \sim_{
      {}_{
        \frac{
          \mathrm{homotopy}
        }
        {
          B \mathrm{Sp}(2)
        }
      }
    }
  }
  \end{aligned}
 \end{equation}

\vspace{-35pt}

\begin{equation}
  \label{Sp24Cohomotopy}
  \begin{aligned}
  \mbox{
    \tiny
    \color{blue}
    \begin{tabular}{c}
      $\mathrm{Sp}(2)$-twisted
      \\
      4-Cohomotopy theory
    \end{tabular}
  }
  \;
  \pi^{i_4 \circ \tau}(X)
  & \;:=\;
  \left\{
    \raisebox{20pt}{
    \xymatrix@C=4em{
      X
      \ar[dr]_-{
        \mathllap{
          \mbox{
            \tiny
            \color{blue}
            twist
          }
        }
        \;
        \tau
      }^>>>>{\ }="t"
      \ar@{-->}[rr]^-{
        \mbox{
          \tiny
          \color{blue}
          continuous function
        }
      }_-{\ }="s"
      &&
      S^4 \!\!\sslash\! \mathrm{Sp}(2)
      \ar[dl]
      \\
      &
      B \mathrm{Sp}(2)
      \ar@{=>}^-{
        \mbox{
          \begin{rotate}{49}
            \!\!\!\!\!\!\!\!\!\!
            \mbox{
              \tiny
              \color{blue}
              homotopy
            }
          \end{rotate}
        }
      } "s"; "t"
    }
    }
  \right\}_{
    \!\!\!
    \Big/
    \sim_{
      {}_{
        \frac{
          \mathrm{homotopy}
        }
        {
          B \mathrm{Sp}(2)
        }
      }
    }
  }
  \end{aligned}
\end{equation}

\medskip

\noindent
{\bf Triality between $\mathrm{Sp}(2)$-structure and $\mathrm{Spin}(5)$-structure.}
While the group
$\mathrm{Sp}(2)\cdot \mathrm{Sp}(1)$ \eqref{SP1Sp2TwistingOfhH}
is abstractly isomorphic to  the central product of Spin-groups
$\mathrm{Spin}(5)\cdot \mathrm{Spin}(3)$,
the two are \emph{distinct} as
subgroups of $\mathrm{Spin}(8)$, and not conjugate to each other.
But as subgroups they are turned into each other by the
ambient action of {triality}:
$$
  \xymatrix@R=4pt{
    &
    \mathrm{Sp}(2)
    \; \ar@{^{(}->}[r]
    &
    \overset{
      \mathclap{
      \mbox{
        \tiny
        \color{blue}
        \begin{tabular}{c}
          central product of
          \\
          quaternion-unitary groups
          \\
          $\phantom{a}$
        \end{tabular}
      }
      }
    }{
      \mathrm{Sp}(2)
      \cdot
      \mathrm{Sp}(1)
    }
    \ar@{^{(}->}[dd]
    \ar@{<->}[rr]^-{\simeq}
    &&
    \overset{
      \mbox{
        \tiny
        \color{blue}
        \begin{tabular}{c}
          central product of
          \\
          Spin-groups
          \\
          $\phantom{a}$
        \end{tabular}
      }
    }{
      \mathrm{Spin}(5)
      \cdot
      \mathrm{Spin}(3)
    }
    \ar@{^{(}->}[dd]
    \ar@{<-^{)}}[r]
    &
    \; \mathrm{Spin}(5)
    \\
    &&&&&&
    \\
    &
    &
    \mathrm{Spin}(8)
    \ar@{<->}[rr]^-{\simeq}_-{
      \underset{
        \mathclap{
        \mbox{
          \tiny
          \color{blue}
          \begin{tabular}{c}
            triality automorphism
          \end{tabular}
        }
        }
      }{
        \mathrm{tri}
      }
    }
    &&
    \mathrm{Spin}(8)
  }
$$
While $\mathrm{Spin}(5)$ on the right is the structure group
of normal bundles to M5-branes, acting on fibers of 4-spherical
fibrations around 5-branes through its vector representation,
 $\mathrm{Sp}(2)$ on the left is the structure group
of normal bundles to M2-branes, acting on the 7-spherical
fibrations around 2-branes via its defining left action
on quaternionic 2-space
$\mathbb{H}^2 \simeq_{{}_{\mathbb{R}}} \mathbb{R}^8$
(\cite{MFGM09}\cite{MF10}):

\vspace{-.4cm}

$$
  \!\!\!\!\!\!\!\!\!
  \!\!\!\!\!\!\!\!\!
  S
  \big(
  \!\!\!
  \xymatrix{\mathbb{H}^2\ar@(ul,ur)^-{
    \mathllap{
      \mbox{
        \tiny
        \color{blue}
        \begin{tabular}{c}
          left quaternion
          \\
          multiplication
        \end{tabular}
      }
      \;
    }
    \mathrm{Sp}(2)}
  }
  \!\!\!
  \big)
  \;=\;
  S^7
  \phantom{AAAAAAAAAAAAA}
  S^4
  \;=\;
  S\big(
  \!\!\!\!\!\!
  \xymatrix{
    \mathbb{R}^5\ar@(ur,ul)_-{
      \mathrm{Spin}(5)
      \mathrlap{
        \;\;
        \mbox{
          \tiny
          \color{blue}
          \begin{tabular}{c}
            vector
            \\
            representation
          \end{tabular}
        }
      }
    }
  }
  \!\!\!\!\!\!
  \big)
$$
In this article we consider only the M2-case. But all formulas
we derive translate to the M5 case via triality.

\medskip

\noindent {\bf Generalized Chern characters.}
Since generalized cohomology theory is rich, one needs tools
to break it down. The first and foremost of these is the
\emph{generalized Chern character} map.
This extracts differential form data underlying a
cocycle in nonabelian generalized cohomology.
The Chern character is familiar in twisted K-theory (see \cite{GS3}\cite{GS-RR}),
as shown in the top half of the following:
\begin{equation}
  \label{GeneralizedChernCharacter}
  \xymatrix@R=1pt{
    &
    \fbox{
      \tiny
      \color{oxford}
      \begin{tabular}{c}
        Torsionful generalized
        \\
        cohomology theory
      \end{tabular}
    }
    \ar[rr]^{
      \mbox{
        \tiny
        \color{brown}
        \begin{tabular}{c}
          approximation by
          \\
          generalized Chern character
        \end{tabular}
      }
    }
    &&
    \fbox{
      \tiny
      \color{oxford}
      \begin{tabular}{c}
        $L_\infty$-valued de Rham
        \\
        cohomology theory
      \end{tabular}
    }
    \\
    \ar@{-}[rrrr]
    &&&&
    \\
    \fbox{
      \tiny
      \color{brown}
      \begin{tabular}{c}
        Chern character on
        \\
        ordinary
        \\
        integral cohomology
      \end{tabular}
    }
    &
    \overset{
      \mathclap{
      \mbox{
        \tiny
        \color{blue}
        \begin{tabular}{c}
          ordinary
          \\
          integral cohomology
        \end{tabular}
      }
      }
    }{
      H^3(X,\mathbb{Z})
    }
    \ar[rr]^{
      \mbox{
        \tiny
        \color{blue}
        \begin{tabular}{c}
          extension of scalars
          \\
          \& de Rham theorem
        \end{tabular}
      }
    }
    &&
    \overset{
      \mathclap{
      \mbox{
        \tiny
        \color{blue}
        \begin{tabular}{c}
          de Rham
          \\
          cohomology
        \end{tabular}
      }
      }
    }{
      H^3_{{}_{\mathrm{dR}}}(X)
    }
    \\
    &
    \raisebox{22pt}{
      $
        \underset{
          \mathclap{
          \mbox{
            \tiny
            \color{blue}
            \begin{tabular}{c}
              bundle gerbe
            \end{tabular}
          }
          }
        }{
          \tau
        }
      $
    }
    &
    \raisebox{22pt}{
      $\longmapsto$
    }
    &
    \raisebox{22pt}{
      $
        \underset{
          \mathclap{
          \mbox{
            \tiny
            \color{blue}
            \begin{tabular}{c}
              closed 3-form
            \end{tabular}
          }
          }
        }{
          [H_3]
        }
      $
    }
    \\
        \fbox{
      \tiny
      \color{brown}
      \begin{tabular}{c}
        Chern character on
        \\
        B-field-twisted
        \\
        K-theory
      \end{tabular}
    }
    &
    \overset{
      \mathclap{
      \mbox{
        \tiny
        \color{blue}
        \begin{tabular}{c}
          $\tau$-twisted
          \\
          complex K-theory
        \end{tabular}
      }
      }
    }{
      \mathrm{KU}^\tau(X)
    }
    \ar[rr]^-{
      \overset{
        \mbox{
          \tiny
          \color{blue}
          \begin{tabular}{c}
            $\tau$-twisted
            \\
            Chern character
          \end{tabular}
        }
      }{
        \mathrm{ch}^\tau
      }
    }
    &&
    \overset{
      \mathclap{
      \mbox{
        \tiny
        \color{blue}
        \begin{tabular}{c}
          $H_3$-twisted
          \\
          de Rham cohomology
        \end{tabular}
      }
      }
    }{
      H^{{}^{[H_3]}}_{{}_{\mathrm{dR}}}(X)
    }
    \\
    \ar@{-}@<-8pt>[rrrr]
    &
    \raisebox{22pt}{
      $
      \underset{
        \mathclap{
        \mbox{
          \tiny
          \color{blue}
          \begin{tabular}{c}
            virtual twisted
            \\
            vector bundle
          \end{tabular}
        }
        }
      }{
        V
      }
      $
    }
    &
    \raisebox{22pt}{
      $\longmapsto$
    }
    &
    \raisebox{22pt}{
      $
      \underset{
        \mathclap{
        \mbox{
          \tiny
          \color{blue}
          \begin{tabular}{c}
            exponentiated
            \\
            curvature form
          \end{tabular}
        }
        }
      }{
        \mathclap{
          \big[
            \mathrm{tr}\big( \exp(\mathrm{F}) \big)
          \big]
        }
      }
      $
    }
    &
    \\
    \fbox{
      \tiny
      \color{brown}
      \begin{tabular}{c}
        Chern character on
        \\
        non-abelian
        \\
        $\mathrm{O}(n)$-cohomology
      \end{tabular}
    }
    &
    \overset{
      \mbox{
        \tiny
        \color{blue}
        \begin{tabular}{c}
          non-abelian
          \\
          $O(n)$-cohomology
        \end{tabular}
      }
    }{
      H^1
      \big(
        X,
        \mathrm{O}(n)
      \big)
    }
    \ar[rr]^-{
      \mbox{
        \tiny
        \color{blue}
        characteristic forms
      }
    }
    &&
    \overset{
      \mathclap{
      \mbox{
      \tiny
      \color{blue}
      \begin{tabular}{c}
        de Rham cohomology tensor
        \\
        invariant polynomials on $\mathfrak{o}(n)$
      \end{tabular}
      }
      }
    }{
      H_{{}_{\mathrm{dR}}}
      \big(
        X
      \big)
      \otimes
      \mathrm{inv}(\mathfrak{o}(n))
    }
    \\
    &
    \raisebox{22pt}{
      $
      \underset{
        \mathclap{
        \mbox{
          \tiny
          \color{blue}
          \begin{tabular}{c}
            vector bundle
          \end{tabular}
        }
        }
      }{
        \tau
      }
      $
    }
    &
    \raisebox{22pt}{
      $
      \longmapsto
      $
    }
    &
    \mathrlap{
    \!\!\!\!\!\!\!\!\!\!\!\!\!\!\!\!\!\!\!\!\!\!
    \!\!\!\!\!\!\!\!\!\!\!\!\!\!\!\!\!\!
    \raisebox{22pt}{
      $
      \tau_{\mathbb{R}}
      \in
      \mathbb{R}
      \big[
        \underset{
          \mathclap{
          \mbox{
            \tiny
            \color{blue}
            \begin{tabular}{c}
              Stiefel-Whitney
              \\
              forms
            \end{tabular}
          }
          }
        }{
          [W_i(\nabla_\tau)]
        }
        \,,\,
        \underset{
          \mathclap{
          \mbox{
            \tiny
            \color{blue}
            \begin{tabular}{c}
              Pontrjagin
              \\
              forms
            \end{tabular}
          }
          }
        }{
          [p_k(\nabla_\tau)]
        }
      \big]_{i,k}
      $
    }
    }
    \\
    \fbox{
      \tiny
      \color{brown}
      \begin{tabular}{c}
        Chern character
        \\
        on J-twisted
        \\
        $n$-Cohomotopy
      \end{tabular}
    }
    &
    \overset{
      \mathclap{
      \mbox{
        \tiny
        \color{blue}
        \begin{tabular}{c}
          $\tau$-twisted
          \\
          Cohomotopy theory
        \end{tabular}
      }
      }
    }{
      \pi^\tau
      \big(
        X
      \big)
    }
    \ar[rr]^{
      \mbox{
        \tiny
        \color{blue}
        \begin{tabular}{c}
          cohomotopical
          \\
          Chern character
        \end{tabular}
      }
    }
    &&
    \overset{
      \mathclap{
      \mbox{
        \tiny
        \color{blue}
        \begin{tabular}{c}
          $\tau_{\mathbb{R}}$-twisted
          \\
          rational Cohomotopy theory
        \end{tabular}
      }
      }
    }{
      \pi^{\tau_{\mathbb{R}}}
      \big(
        X
      \big)_{\mathbb{R}}
    }
  }
\end{equation}
\noindent In order to see what the \emph{cohomotopical Chern character}
in the last line is, we need some general theory of
generalized Chern characters. This is
\emph{rational homotopy theory}:

\medskip

\noindent {\bf Rational homotopy theory.}
In the language of homotopy theory,
generalized Chern character maps are examples of
\emph{rationalization}, whereby the homotopy type of a
topological space
(here: the classifying space of a generalized cohomology theory)
is approximated by tensoring all its homotopy groups with the
rational numbers (equivalently: the real numbers),
thereby disregarding all torsion subgroups in homotopy groups
and in cohomology groups.

$$
  
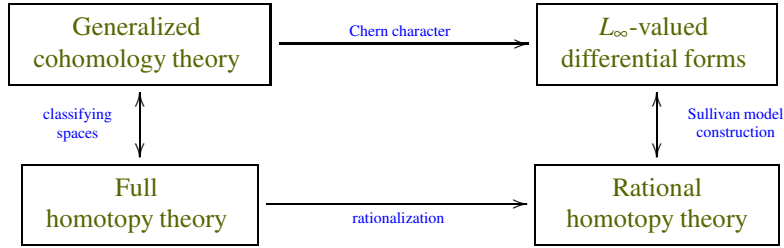
\xymatrix@C=4em{
    \fbox{\small
      \color{brown}
      \begin{tabular}{c}
        Generalized
        \\
        cohomology theory
      \end{tabular}
    }
    \ar@{<->}[d]_-{
      \mbox{
        \tiny
        \color{blue}
        \begin{tabular}{c}
          classifying
          \\
          spaces
        \end{tabular}
      }
    }
    \ar[rr]^{
      \mbox{
        \tiny
        \color{blue}
        Chern character
      }
    }
    &&
    \fbox{\small
      \color{brown}
      \begin{tabular}{c}
        $L_\infty$-valued
        \\
        differential forms
      \end{tabular}
    }
    \ar@{<->}[d]^{
      \mbox{
        \tiny
        \color{blue}
        \begin{tabular}{c}
          Sullivan model
          \\
          construction
        \end{tabular}
      }
    }
    \\
    \fbox{\small
      \color{brown}
      \begin{tabular}{c}
        Full
        \\
        homotopy theory
      \end{tabular}
    }
    \ar[rr]_-{
      \mbox{
        \tiny
        \color{blue}
        rationalization
      }
    }
    &&
    \fbox{ \small
      \color{brown}
      \begin{tabular}{c}
        Rational
        \\
        homotopy theory
      \end{tabular}
    }
  }
$$
What makes rational homotopy theory amenable to
computations is the existence of \emph{Sullivan models}.
These are differential graded-commutative algebras (dgc-algebras)
on a finite number of generating elements
(spanning the rational homotopy groups) subject to differential relations
(enforcing the intended rational cohomology groups).
In the supergravity literature Sullivan models are also known as ``FDA''s.
Here are some basic examples
(see \cite{FSS16b}\cite{FSS18}\cite{Higher-T}\cite{FSS19d},
review in \cite{FSS19a}):

\medskip
\medskip
{\small
\hspace{-.9cm}
\def\arraystretch{1.2}
\begin{tabular}{|c||c|c|rcl|}
  \hline
  &
  \!\!\!\!\!\! \begin{tabular}{c}
    {\bf Rational }
    \\
    {\bf super space }
  \end{tabular}
  \!\!\!\!\!\!
   &
 \!\!\!\!\!\!  \begin{tabular}{c}
    {\bf Loop }
    \\
    {\bf super $L_\infty$-algebra }
  \end{tabular}
  \!\!\!\!\!\!\!\!
  &
  \multicolumn{3}{c|}{
  \!\!\!\! \begin{tabular}{c}
    {\bf Chevalley-Eilenberg }
    \\
    {\bf super dgc-algebras}
    \\
    (``Sullivan models'', ``FDA''s)
  \end{tabular}
  \!\!\!\!
  }
    \\
  \hline
  \hline
  $\mathclap{ \phantom{ {\vert \atop \vert} \atop \vert } }$
 \!\!\!\!\!\! \mbox{General} \!\!\!\!\!\!
    &
  $\color{blue}X$
    &
  $\color{blue}\mathfrak{l}X$
    &
  \multicolumn{3}{c|}{
    \color{blue}
    $\mathrm{CE}\big( \mathfrak{l}X \big)$
  }
  \\
  \hline
  $\mathclap{ \phantom{ {{\vert \atop \vert} \atop {\vert \atop \vert} } \atop { {\vert \atop \vert } } } }$
  \begin{tabular}{c}
    \color{brown}
    Super
    \\
    \color{brown}
    spacetime
  \end{tabular}
    &
  $\mathbb{T}^{d,1\vert \mathbf{N}}$
    &
  $\mathbb{R}^{d,1\vert \mathbf{N}}$
    &
  $
    \mathbb{R}\big[ \{\psi^\alpha\}_{\alpha = 1}^N, \{e^a\}_{a = 0}^d \big]
  $
  &
  $
    \!\!\!\!\!\!\!\!\!
    \Big/
    \!\!\!\!\!\!\!\!\!
  $
  &
  $
    \left(
      \!\!\!
      \begin{array}{lcl}
        d\,\psi^\alpha & \!\!\!\!\!\!\ = \!\!\!\!\!\!\ & 0
        \\
        d\,e^a & \!\!\!\!\!\!\ = \!\!\!\!\!\!\ &
          \overline{\psi}\,\Gamma^a \psi
      \end{array}
      \!\!\!
    \right)
  $
  \\
  \hline
  $\mathclap{ \phantom{ {{\vert \atop \vert} \atop {\vert \atop \vert} } \atop { {\vert \atop \vert } } } }$
  \begin{tabular}{c}
    \color{brown}
    Eilenberg-MacLane
    \\
    \color{brown}
    space
  \end{tabular}
  &
  $
  \begin{aligned}
    & K(\mathbb{R},p+2)
    \\
   & \simeq_{{}_{\mathbb{R}}} B^{p+1} S^1
  \end{aligned}
  $
  &
  $\mathbb{R}[p+1]$
  &
  $
    \mathbb{R}[ c_{p+2}]
  $
  &
  $
    \!\!\!\!\!\!\!\!\!
    \big/
    \!\!\!\!\!\!\!\!\!
  $
  &
  $
    \big(
      \!\!\!
      \begin{array}{lcl}
        d\,c_{p+2} & \!\!\!\!\!\!\ = \!\!\!\!\!\!\ & 0
      \end{array}
      \!\!\!
    \big)
  $
  \\
  \hline
  $\mathclap{ \phantom{ {{\vert \atop \vert} \atop {\vert \atop \vert} } \atop { {\vert \atop \vert } } } }$
  \begin{tabular}{c}
    \color{brown}
    Odd-dimensional
    \\
    \color{brown}
    sphere
  \end{tabular}
  &
  $S^{2k+1}_{\mathbb{R}}$
  &
  $\mathfrak{l}(S^{2k+1})$
  &
  $
  \mathbb{R}[\omega_{2k+1}]
  $
  &
  $
    \!\!\!\!\!\!\!\!\!
    \big/
    \!\!\!\!\!\!\!\!\!
  $
  &
  $
  \big(
    \!\!\!
    \begin{array}{lcl}
      d\,\omega_{2k+1} & \!\!\!\!\!\! = \!\!\!\!\!\! & 0
    \end{array}
    \!\!\!
  \big)
  $
  \\
  \hline
  $\mathclap{ \phantom{ {{\vert \atop \vert} \atop {\vert \atop \vert} } \atop { {\vert \atop \vert } } } }$
  \begin{tabular}{c}
    \color{brown}
    Even-dimensional
    \\
    \color{brown}
    sphere
  \end{tabular}
  &
  $S^{2k}_{\mathbb{R}}$
  &
  $\mathfrak{l}(S^{2k})$
  &
  $
    \mathbb{R}\big[\omega_{2k}, \omega_{4k-1}\big]
  $
  &
  $
    \!\!\!\!\!\!\!\!\!
    \Big/
    \!\!\!\!\!\!\!\!\!
  $
  &
  $
  \left(
    \!\!\!
    \begin{array}{lcl}
      d\,\omega_{2k} & \!\!\!\!\!\! = \!\!\!\!\!\! & 0
      \\
      d\,\omega_{4k-1} & \!\!\!\!\!\! = \!\!\!\!\!\! &
       -\omega_{2k} \wedge \omega_{2k}
    \end{array}
    \!\!\!
  \right)
  $
  \\
  \hline
  $\mathclap{ \phantom{ {{\vert \atop \vert} \atop {\vert \atop \vert} } \atop { {\vert \atop \vert } } } }$
  \begin{tabular}{c}
    \color{brown}
    M2-extended
    \\
    \color{brown}
    super spacetime
  \end{tabular}
  &
  $\widehat{ \mathbb{T}^{10,1\vert \mathbf{32}} }$
  &
  $\mathfrak{m}2\mathfrak{brane}$
  &
  $
    \mathbb{R}
    \big[
      \{\psi^\alpha\}_{\alpha = 1}^{32},
      \{e^a\}_{a = 0}^{10},
      h_3
    \big]
  $
  &
  $
    \!\!\!\!\!\!\!\!\!
    \Big/
    \!\!\!\!\!\!\!\!\!
  $
  &
  $
    \left(
      \!\!\!
      \begin{array}{lcl}
        d\,\psi^\alpha & \!\!\!\!\!\!\!\!\! = \!\!\!\!\!\!\ & 0
        \\
        d\,e^a & \!\!\!\!\!\!\!\!\! = \!\!\!\!\!\!\ &
          \overline{\psi}\,\Gamma^a \psi
        \\
        d\,h_3 & \!\!\!\!\!\!\!\!\!= \!\!\!\!\!\!\ &
          \tfrac{i}{2}
         (\overline{\psi} \Gamma_{a b} \psi)
         \wedge
         e^a \wedge e^b
      \end{array}
      \!\!\!\!\!
    \right)
  $
  \\
  \hline
\end{tabular}
}

\medskip

\noindent Under \emph{Sullivan's theorem} the rational homotopy type
of well-behaved spaces
are equivalently encoded in their Sullivan model dgc-algebras.
For spaces and algebras which are  nilpotent and of finite type one has:
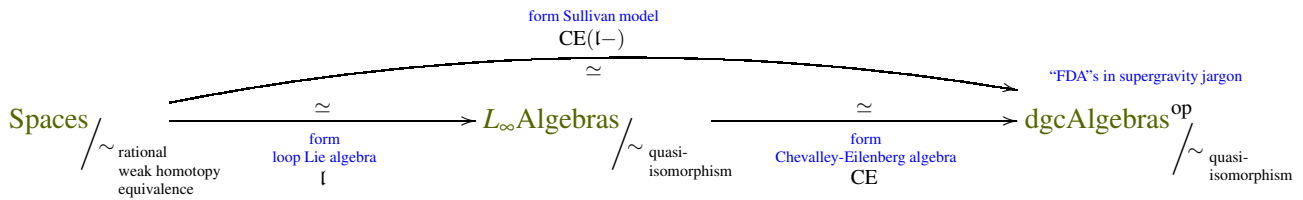
\vspace{-.5cm}

 $$
   \xymatrix@C=5em{
   {
     \color{brown}
     \mathrm{Spaces}
   }_{
     \!\!\!\!
     \Big/
     \!\!
     \sim_{
     \!\!\!\!\!\!\!\!\!\!\!
     \mathrlap{
     \mbox{
       \tiny
       \begin{tabular}{l}
         rational
         \\
         weak homotopy
         \\
         equivalence
       \end{tabular}
     }
     }
     \;\;\;\;\;\;\;\;\;
     }
   }
  \;\;\;\;\;
  \ar[rr]_-{
     \overset{
       \mathclap{
       \mbox{
         \tiny
         \color{blue}
         \begin{tabular}{c}
           form
           \\
           loop Lie algebra
         \end{tabular}
       }
       }
     }{
       \mathfrak{l}
     }
   }^-{\simeq}
   \ar@/^2pc/[rrrr]^-{
     \overset{
       \mbox{
         \tiny
         \color{blue}
         \begin{tabular}{c}
           form
           Sullivan model
         \end{tabular}
       }
     }{
       \mathrm{CE}(\mathfrak{l} - )
     }
   }_-{\simeq}
   &&
   {\color{brown}
     L_\infty \mathrm{Algebras}
   }_{
     \!\!\!\!
     \Big/
     \!\!\!
     \sim_{
     \!\!\!\!\!\!\!\!\!\!
     \mathrlap{
     \mbox{
       \tiny
       \begin{tabular}{l}
         quasi-
         \\
         isomorphism
       \end{tabular}
     }
     }
     \;\;\;
   }
   }
   \;\;\;\;\;\;\;\;\;\;
   \ar[rr]_-{
     \overset{
       \mathclap{
       \mbox{
         \tiny
         \color{blue}
         \begin{tabular}{c}
           form
           \\
           Chevalley-Eilenberg algebra
         \end{tabular}
       }
       }
     }{
       \mathrm{CE}
     }
   }^-{\simeq}
   &&
     {\color{brown}
       \overset{
         \mathrlap{
         \;\;\;\;\;\;\;\;\;\;
         \mbox{
           \tiny
           \color{blue}
           \begin{tabular}{c}
           ``FDA''s in supergravity jargon
           \\
           $\phantom{a}$
           \end{tabular}
         }
         }
       }{
         \mathrm{dgcAlgebras}
       }
     }^{\mathrm{op}}_{
       \Big/\!\!\!
       \sim_{
         \!\!\!\!\!\!\!\!\!\!
         \mbox{
           \tiny
           \begin{tabular}{l}
             quasi-
             \\
             isomorphism
           \end{tabular}
         }
       }
     }
  }
 $$

 \vspace{-7mm}
When applying the rational approximation to twisted generalized
cohomology theory, the order matters: There are in general more
\emph{rational twists}
$X \overset{\tau}{\longrightarrow} B \mathrm{Aut}(A_{\mathbb{R}})$
for \emph{twisted rational cohomology}
than there are rationalizations $\tau_{\mathbb{R}}$ of full twists
$X \overset{\tau}{\longrightarrow} B \mathrm{Aut}(A)$
for \emph{rational twisted cohomology}.
\footnote{This is in contrast with twisting vs. differential refinement where the
order does not matter -- see \cite{GS3}\cite{GS4}.}
We consider first the general
rational twists:

\medskip

\noindent {\bf Rationally twisted rational Cohomotopy.}
We find that the \emph{rationally twisted rational Cohomotopy}
sets in degrees 4 and 7 are equivalently characterized by
cohomotopical Chern character forms as follows:

\begin{equation}
\label{RationalTwists}
\mbox{
\begin{tabular}{c||c|ccc}
    &
    {
      \tiny
      \color{blue}
      \begin{tabular}{c}
        rational
        twist
      \end{tabular}
    }
    &
    {
      \tiny
      \color{blue}
      \begin{tabular}{c}
        rational
        twisted
        \\
        Cohomotopy
      \end{tabular}
    }
    &
    &
    {
      \tiny
      \color{blue}
      \begin{tabular}{c}
        $\phantom{a}$
        \\
        cohomotopical
        \\
        Chern characters
        \\
        $\phantom{a}$
      \end{tabular}
    }
    \\
    \hline
    &&
    \\
    $
    \;\;\;\;\;\;\;
    \mathllap{
    \mbox{
      \tiny
      \color{blue}
      7-Cohomotopy
    }
    }
    $
    &
    $
    X
      \overset{
        \tau^7
      }{\longrightarrow}
    B
    \mathrm{Aut}
    \big(
      S^7_{\mathbb{R}}
    \big)
    $
    &
    $
    \pi^{(\tau^7)}(X)
    $
    &
    $
    \simeq
    $
    &
    $
    \overset{
      \;\;\;\;\;\;\;\;\;\;\;\;\;
      \;\;\;\;\;\;\;\;\;\;\;\;\;
      \;\;\;\;\;\;\;\;\;\;\;\;\;
      \;\;\;\;\;\;\;\;\;\;\;\;\;
      \;\;\;\;\;\;\;
          \mathclap{
          \mbox{
            \tiny
            \color{blue}
            \begin{tabular}{c}
              characteristic form
              \\
              of twist $\tau^7$
            \end{tabular}
          }
          }
        }
    {
    \Big\{
      \phantom{(G_4,}
      \overset{
        \mbox{
          \tiny
          \color{brown}
          7-form
        }
      }{
        \widetilde G_7
      }
      \phantom{)\;}
      \;
    \left|
      \;\;
      \begin{aligned}
        d\,\widetilde G_7
          &=
        \phantom{-\tfrac{1}{2} G_4 \wedge G_4 +}
        K_8
      \end{aligned}
      \!\!
      \;\;\;\;
    \right.
    \Big\}_{\big/\sim}
    }
    $
    \\
    &&
    \\
    $
    \;\;\;\;\;\;\;
    \mathllap{
    \mbox{
      \tiny
      \color{blue}
      4-Cohomotopy
    }
    }
    $
    &
    $
    X
      \overset{
        \tau^4
      }{\longrightarrow}
    B
    \mathrm{Aut}
    \big(
      S^4_{\mathbb{R}}
    \big)
    $
    &
    $
    \pi^{(\tau^4)}(X)
    $
    &
    $\simeq$
    &
    $
      \underset{
        \;\;\;\;\;\;\;\;\;\;\;\;\;
        \;\;\;\;\;\;\;\;\;\;\;\;\;
        \;\;\;\;\;\;\;\;\;\;\;\;\;
        \;\;\;\;\;\;\;\;\;\;\;\;\;
        \;\;\;\;\;\;\;\;\;\;\;\;\;
        \mathclap{
        \mbox{
          \tiny
          \color{blue}
          \begin{tabular}{c}
            characteristic form
            \\
            of twist $\tau^4$
          \end{tabular}
        }
        }
      }
    {
    \left\{
      \overset{
        \mbox{
          \tiny
          \color{brown}
          \begin{tabular}{c}
            $\;\,\,$ 4-form
            \\
            \& 7-form
          \end{tabular}
        }
      }{
        (G_4, G_7)
      }
      \;
    \left|
      \;\;
      \begin{aligned}
      d\,G_4 & = \phantom{-}0
      \\
      d\,G_7 & = - \tfrac{1}{2} G_4 \wedge G_4
      +
      L_8
      \end{aligned}
      \;\;\;\;
    \right.
    \right\}_{\big/\sim}
    }
    $
\end{tabular}
}
\end{equation}
Here \emph{all} real 8-classes
$
  [K_8]
  ,\,
  [L_8]
  \;\in\;
  H^8(X, \mathbb{R})
$
may appear, for \emph{some} rational twists $\tau^{4, 7}$.
Constraints on these characteristic forms
appear when we
consider more than rational twisted structure:

\medskip

\noindent {\bf Compatibly rationally twisted rational Cohomotopy.}
We may ask that the rational twists $\tau^{4,7}$
in \eqref{RationalTwists}
are related analogously to how the twisted parametrized
Hopf fibration \eqref{SP1Sp2TwistingOfhH} relates the
full (non-rational) twists, through \eqref{ReflectTwisted7Cohomotopy}.
We find that this happens precisely when the difference of the
characteristic 8-classes in \eqref{RationalTwists}
is a complete square
$$
  L_8
  \;=\;
  K_8
  +
  \big(
    \tfrac{1}{4} P_4
  \big)
    \wedge
  \big(
    \tfrac{1}{4}P_4
  \big)
$$
and in that case the situation of \eqref{RationalTwists}
becomes the following:

\begin{equation}
\label{CompatibleRationalTwists}
\mbox{
\begin{tabular}{c||c|ccc}
    &
    {
      \tiny
      \color{blue}
      \begin{tabular}{c}
        {\color{brown} Compatible}
        \\
        rational
        twists
      \end{tabular}
    }
    &
    {
      \tiny
      \color{blue}
      \begin{tabular}{c}
        Rational
        \\
        {\color{brown} compatibly} twisted
        \\
        Cohomotopy
        \\
        $\phantom{a}$
      \end{tabular}
    }
    &
    &
    {
      \tiny
      \color{blue}
      \begin{tabular}{c}
        $\phantom{a}$
        \\
        Cohomotopical
        \\
        Chern characters
        \\
        $\phantom{a}$
      \end{tabular}
    }
    \\
    \hline
    &&
    \\
    $
    \;\;\;\;\;\;\;
    \mathllap{
    \mbox{
      \tiny
      \color{blue}
      7-Cohomotopy
    }
    }
    $
    &
    $
    X
      \overset{
        \tau^7
      }{\longrightarrow}
    B
    \mathrm{Aut}
    \big(
      S^7_{\mathbb{R}}
    \big)
    $
    &
    $
    \pi^{(\tau^7)}(X)
    $
    &
    $
    \simeq
    $
    &
    $
    \overset{
      \;\;\;\;\;\;\;\;\;\;\;\;\;
      \;\;\;\;\;\;\;\;\;\;\;\;\;
      \;\;\;\;\;\;\;\;\;\;\;\;\;
      \;\;\;\;\;\;\;\;\;\;\;\;\;
      \;\;\;\;\;\;\;
          \mathclap{
          \mbox{
            \tiny
            \color{blue}
            \begin{tabular}{c}
              characteristic form
              \\
              of twist $\tau^7$
            \end{tabular}
          }
          }
        }
    {
    \left\{
      \phantom{(G_4,}
      \widetilde G_7
      \phantom{)\;}
      \;
      \;
    \left|
      \;\;
      \begin{aligned}
        d\,\widetilde G_7
          &=
        \phantom{-\tfrac{1}{2} G_4 \wedge G_4 +}
        K_8
      \end{aligned}
      \!\!
      \;\;\;\;
    \right.
    \right\}_{\big/\sim}
    }
    $
    \\
    &
    \multicolumn{2}{c}
    {
      \fbox{
        $
        \begin{aligned}
          \overset{
            \mathrlap{
            \mbox{
              \tiny
              \color{brown}
              \begin{tabular}{c}
                shifted
                4-form
              \end{tabular}
            }
            }
          }{
            \widetilde G_4
          }
            & := G_4 + \tfrac{1}{4}P_4
          \\
          \underset{
            \mathrlap{
            \mbox{
              \tiny
              \color{brown}
              \begin{tabular}{c}
                shifted
                7-form
              \end{tabular}
            }
            }
          }{
            \widetilde G_7
          }
            & : = G_7 + \tfrac{1}{2} H_3 \wedge \widetilde G_4
        \end{aligned}
        $
      }
    }
    &
    $
    \simeq
    $
    &
    $
    {
    \left\{
      \left(
        \begin{aligned}
           H_3, &
           \\
           \widetilde G_4, & \, G_7
         \end{aligned}
      \right)
      \;
    \left|
      \;\;
      \begin{aligned}
        d\,H_3
          & =
          \phantom{-\tfrac{1}{2}} \widetilde G_4 - \tfrac{1}{2} P_4
        \\
        d\, \widetilde G_4 & = 0
        \\
        d\,G_7
          &=
        -\tfrac{1}{2}
         d H_3
          \wedge
        \widetilde G_4
        +
        K_8
      \end{aligned}
      \!\!
      \;\;\;\;
    \right.
    \right\}_{\big/\sim}
    }
    $
    \\
    &&
    \\
    &&
    \\
    $
    \;\;\;\;\;\;\;
    \mathllap{
    \mbox{
      \tiny
      \color{blue}
      4-Cohomotopy
    }
    }
    $
    &
    $
    X
      \overset{
        \tau^4
      }{\longrightarrow}
    B
    \mathrm{Aut}
    \big(
      S^4_{\mathbb{R}}
    \big)
    $
    &
    $
    \pi^{(\tau^4)}(X)
    $
    &
    $\simeq$
    &
    $
    {
    \left\{
      (\widetilde G_4, G_7)
      \;
    \left|
      \;\;\;
      \begin{aligned}
      d\,\widetilde G_4 & = \phantom{-}0
      \\
      d\,G_7
        & =
        - \tfrac{1}{2}
        (\widetilde G_4 -\tfrac{1}{2}P_4)
          \wedge
        \widetilde G_4
      +
      K_8
      \end{aligned}
      \;\;\;\;
    \right.
    \right\}_{\big/\sim}
    }
    $
\end{tabular}
}
\end{equation}
Here still \emph{all} real 8-classes and 4-classes
$
  [K_8] \;\in\; H^8(X,\mathbb{R})
  \,,
  \;\;
  [P_4] \;\in\; H^4(X,\mathbb{R})
$
may appear, for \emph{some} pair of compatible rational twists.

\medskip

Next we find that these real classes are fixed
as we consider full (not just rational) $\mathrm{Sp}(2)$-twists,
compatible by the full (not just rational)
$\mathrm{Sp}(2)$-twisted
quaternionic Hopf fibration  \eqref{SP1Sp2TwistingOfhH}.

\medskip

\noindent{\bf J-Twisted 4-Cohomotopy of $\mathrm{Sp}(2)$-manifolds.}
Consider a simply-connected Riemannian Spin manifold $\mathbb{R}^{2,1} \times X^8$
with affine connection $\nabla$ and equipped with:
\begin{enumerate}[{\bf (i)}]
\vspace{-3mm}
\item an $\mathrm{Sp}(2)$-structure $\tau$ \eqref{SP1Sp2Structure};
\vspace{-3mm}
\item a cocycle $c$ in $\tau$-twisted 4-Cohomotopy \eqref{Sp24Cohomotopy};
\end{enumerate}
\vspace{-3mm}
\noindent hence equipped with a homotopy-commutative diagram of
continuous maps as follows:
$$
  \left[
  \;\;\;\;\;\;\;\;
  \raisebox{29pt}
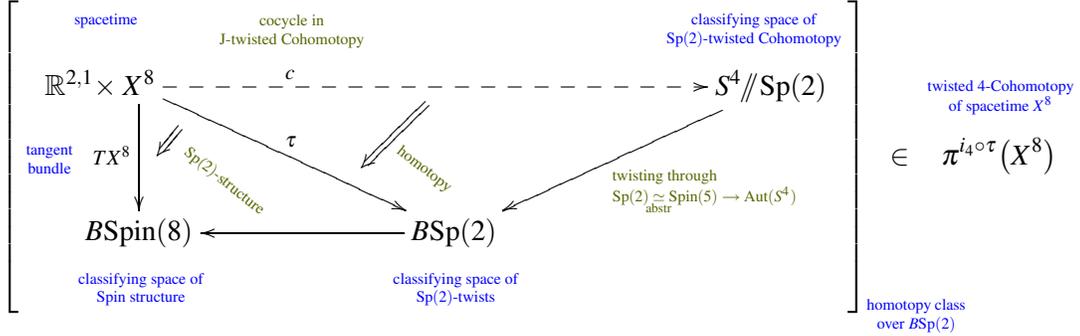
{
  \xymatrix@R=14pt@C=7em{
    \mathllap{
      \mathbb{R}^{2,1}
      \times
      \;
    }
    \overset{
      \mathclap{
      \mbox{
        \tiny
        \color{blue}
        \begin{tabular}{c}
          spacetime
          \\
          $\phantom{a}$
          \\
          $\phantom{a}$
        \end{tabular}
        \;\;\;\;\;\;\;\;\;\;\;\;\;\;\;
        }
      }
    }{
      X^8
    }
    \ar[dr]^-{\tau}^>>>>>>>>>{\ }="t"_<<<<<{\ }="s2"
    \ar[d]_{
      \mathllap{
        \mbox{
          \tiny
          \color{blue}
          \begin{tabular}{c}
            tangent
            \\
            bundle
          \end{tabular}
        }
      }
      T X^8
    }^>>>>>>>{\ }="t2"
    \ar@{-->}[rr]^<<<<<<<<<<<<<<<<<{
      \overset
      {
        \mathclap{
        \mbox{
          \tiny
          \color{brown}
          \begin{tabular}{c}
            cocycle in
            \\
            J-twisted Cohomotopy
            \\
            $\phantom{a}$
          \end{tabular}
        }
        }
      }
      {
        c
      }
    }_-{\ }="s"
    &&
    \overset{
      \mathclap{
      \mbox{
        \hspace{-.6cm}
        \tiny
        \color{blue}
        \begin{tabular}{c}
          classifying space of
          \\
          $\mathrm{Sp}(2)$-twisted Cohomotopy
          \\
          $\phantom{a}$
        \end{tabular}
      }
      }
    }{
      S^4 \!\!\sslash\! \mathrm{Sp}(2)
    }
    \ar[dl]^-{\;\;\;\;\;\;\;\;\;\;
      \mbox{
        \tiny
        \color{brown}
        \begin{tabular}{l}
          twisting through
          \\
          $\mathrm{Sp}(2)
          \underset{\mathclap{\mathrm{abstr}}}{\simeq}
          \mathrm{Spin}(5) \to \mathrm{Aut}(S^4)$
        \end{tabular}
      }
    }
    \\
    \underset{
      \mathclap{
      \mbox{
        \tiny
       \color{blue}
        \begin{tabular}{c}
          $\phantom{a}$
          \\
           classifying space of
          \\
          Spin structure
        \end{tabular}
      }
      }
    }{
      B \mathrm{Spin}(8)
    }
    &
    \underset{
      \mathclap{
      \mbox{
        \tiny
        \color{blue}
        \begin{tabular}{c}
          $\phantom{a}$
          \\
          classifying space of
          \\
          $\mathrm{Sp}(2)$-twists
        \end{tabular}
      }
      }
    }{
      B \mathrm{Sp}(2)
    }
    \ar[l]
    \ar@{=>}^-{
      \raisebox{25pt}{
      \hspace{-12pt}
      \begin{rotate}{-38}
        \mbox{
          {
          \tiny
          \color{brown}
          homotopy
          }
        }
      \end{rotate}
      }
    }
      "s"; "t"
    \ar@{=>}^-{
      \begin{rotate}{-39.6}
      \,
      \mbox{
        \tiny
        \color{brown}
        $\mathrm{Sp}(2)$-structure
      }
      \end{rotate}
    }
      "s2"; "t2"
  }
  }
  \;\;
  \right]_{
    \!\!\!\!\!\!\!\!\!\!
    \mathrlap{
    \mbox{
      \color{blue}
      \tiny
      \begin{tabular}{c}
        homotopy class
        \\
        over $B \mathrm{Sp}(2)$
      \end{tabular}
    }
    }
  }
  \;\;\;\;\;\;\in\;\;\;\;
  \overset{
    \mathclap{
    \mbox{
      \tiny
      \color{blue}
      \begin{tabular}{c}
        twisted 4-Cohomotopy
        \\
        of spacetime $X^8$
        \\
        $\phantom{a}$
      \end{tabular}
    }
    }
  }{
    \pi^{i_4 \circ \tau}\big( X^8 \big)
  }
$$

\noindent Then the {\bf 4-Cohomotopical Chern character}
\eqref{GeneralizedChernCharacter} of $[c]$,
hence the differential flux forms $(G_4,G_7)$ underlying
$[c]$ by \eqref{RationalTwists},
as indicated on the left in the following diagram

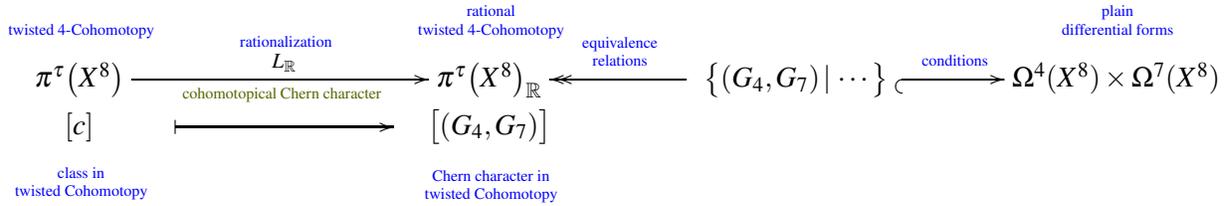
\vspace{-.7cm}

$$
  \xymatrix@R=-25pt@C=3.5em{
    \overset{
      \mathclap{
      \mbox{
        \tiny
        \color{blue}
        \begin{tabular}{c}
          $\phantom{a}$
          \\
          twisted 4-Cohomotopy
          \\
          $\phantom{a}$
        \end{tabular}
      }
      }
    }{
      \pi^\tau\big( X^8 \big)
    }
    \ar[rr]^{
      \overset{
        \mathclap{
        \mbox{
          \tiny
          \color{blue}
          \begin{tabular}{c}
            rationalization
            \\
          \end{tabular}
        }
        }
      }{
        L_{\mathbb{R}}
      }
    }_-{
      \mathclap{
      \mbox{
        \tiny
        \color{brown}
        cohomotopical
        Chern character
      }
      }
    }
    &&
    \overset{
      \mathclap{
      \mbox{
        \tiny
        \color{blue}
        \begin{tabular}{c}
          rational
          \\
          twisted 4-Cohomotopy
          \\
          $\phantom{a}$
        \end{tabular}
      }
      }
    }{
      \pi^\tau\big( X^8 \big)_{\mathbb{R}}
    }
    \ar@{<<-}[r]^-{
      \mbox{
        \tiny
        \color{blue}
        \begin{tabular}{c}
          equivalence
          \\
          relations
        \end{tabular}
      }
    }
    &
  \;\,
  \big\{
    (G_4, G_7)
    \,\vert\,
    \cdots
  \big\}
    \;\ar@{_{(}->}[r]^-{
      \mbox{
        \tiny
        \color{blue}
        \begin{tabular}{c}
          conditions
        \end{tabular}
      }
    }
    &
    \overset{
      \mathclap{
      \mbox{
        \tiny
        \color{blue}
        \begin{tabular}{c}
          plain
          \\
          differential forms
          \\
          $\phantom{a}$
        \end{tabular}
      }
      }
    }{
      \Omega^4(X^8) \times \Omega^7(X^8)
    }
    \\
    \underset{
      \mbox{
        \tiny
        \color{blue}
        \begin{tabular}{c}
          $\phantom{a}$
          \\
          class in
          \\
          twisted Cohomotopy
        \end{tabular}
      }
    }{
      [c]
    }
    \ar@{|->}[rr]
    &&
    \underset{
      \mbox{
        \tiny
        \color{blue}
        \begin{tabular}{c}
          $\phantom{a}$
          \\
          Chern character in
          \\
          twisted Cohomotopy
        \end{tabular}
      }
    }{
    \big[
      (
      G_4, G_7
      )
    \big]
    }
  }
$$
satisfy, first of all, this condition:

\noindent The {\bf shifted 4-flux} form

\vspace{-.9cm}

\begin{equation}
  \label{ShiftByBackgroundCharge}
  \widetilde G_4
  \;\coloneqq\;
  \underset{
    \mathclap{
    \mbox{
      \tiny
      \color{blue}
      \begin{tabular}{c}
        $\phantom{a}$
        \\
        naive
        \\
        4-flux
      \end{tabular}
    }
    }
  }{
    G_4
  }
  \;+\;
  \underset{
    \mathclap{
    \mbox{
      \tiny
      \color{blue}
      \begin{tabular}{c}
        $\phantom{a}$
        \\
        shift by first
        \\
        fractional
        \\
        Pontrjagin form
      \end{tabular}
    }
    }
  }{
    \tfrac{1}{4}p_1(\nabla)
  }
  \;\in\;
  \underset{
    \mathclap{
    \mbox{
      \tiny
      \color{blue}
      \begin{tabular}{c}
        $\phantom{a}$
        \\
        differential
        \\
        4-forms
      \end{tabular}
    }
    }
  }{
    \Omega^4
    \big(
      X^8
    \big)
  }
\end{equation}
represents an {\bf  integral} cohomology class
\begin{equation}
  \label{ShiftedFluxQuantizationCondition}
  \xymatrix{
  \underset{
    \mathclap{
    \mbox{
    \tiny
    \color{brown}
    \begin{tabular}{c}
      $\phantom{a}$
      \\
      shifted
      \\
      4-flux
    \end{tabular}
    }
    }
  }{
    [\widetilde G_4]
  }
  \;\in\;
  \underset{
    \mathclap{
    \mbox{
      \tiny
      \color{blue}
      \begin{tabular}{c}
        $\phantom{a}$
        \\
        integral
        cohomology
      \end{tabular}
      }
    }
  }{
    H^4\big( X^8, \mathbb{Z}\big)
  }
  \ar[rr]^-{
    \mathclap{
    \mbox{
      \tiny
      \color{blue}
      \begin{tabular}{c}
        extension
        of scalars
      \end{tabular}
    }
    }
  }
  &&
  \underset{
    \mathclap{
    \mbox{
      \tiny
      \color{blue}
      \begin{tabular}{c}
        $\phantom{a}$
        \\
        real cohomology
      \end{tabular}
    }
    }
  }{
    H^4\big( X^8, \mathbb{R}\big)
  }
  \;\simeq\;
  \underset{
    \mathclap{
    \mbox{
      \tiny
      \color{blue}
      \begin{tabular}{c}
        $\phantom{a}$
        \\
        de Rham cohomology
      \end{tabular}
    }
    }
  }{
    H_{{}_{\mathrm{dR}}}(X^8)
  }
  }
\end{equation}
on which the action of the {\bf Steenrod square vanishes}:
\begin{equation}
  \label{SteenrodSquareVanishes}
  \overset{
    \mathclap{
    \mbox{
      \tiny
      \color{blue}
      \begin{tabular}{c}
        Steenrod square of
        \\
        mod-2 reduction of
        \\
        integral shifted 4-flux
        \\
        $\phantom{a}$
      \end{tabular}
    }
    }
  }{
  \mathrm{Sq}^2
  \big(
    [\widetilde G_2]
  \big)
  }
  = 0
  \;\;\;\;\;\;\;\;
  \mbox{hence also}
  \;\;\;\;\;\;\;\;
  \overset{
    \mathclap{
    \mbox{
      \tiny
      \color{blue}
      \begin{tabular}{c}
        Steenrod cube of
        \\
        mod-2 reduction of
        \\
        integral shifted 4-flux
        \\
        $\phantom{a}$
      \end{tabular}
    }
    }
  }{
    \mathrm{Sq}^3
    \big(
      [\widetilde G_2]
    \big)
  }
  = 0
  \,,
\end{equation}
and its {\bf background charge} in the case of factorization through
$h_{\mathbb{H}} \!\sslash \mathrm{Sp}(2) $ is
\begin{equation}
  \label{BackgroundChargeValue}
  \overset{
    \mathclap{
    \mbox{
      \tiny
      \color{blue}
      \begin{tabular}{c}
        residual flux of cocucle
        \\
        factoring through $h_{\mathbb{H}} \!\sslash\! \mathrm{Sp}(2)$
        \\
        $\phantom{a}$
      \end{tabular}
    }
    \;\;\;\;\;\;\;\;\;\;
    }
  }{
    (G_4)_0
  }
    \;\;=\;\;
  \overset{
    \mathclap{
    \;\;\;\;
    \mbox{
      \tiny
      \color{blue}
      \begin{tabular}{c}
        background charge
        \\
        $\phantom{a}$
      \end{tabular}
    }
    }
  }{
    \tfrac{1}{4}p_1(\nabla)
  }
  \,.
\end{equation}

To see the next condition
satisfied by the pair $(G_4,G_7)$,
consider the homotopy pullback of the 4-Cohomotopy cocycle
$c$ along the $\mathrm{Sp}(2)$-twisted quaternionic Hopf fibration
$h_{\mathbb{H}}$ to a cocycle in twisted 7-Cohomotopy
on the induced 3-spherical fibration $\widehat H^8$ over spacetime:

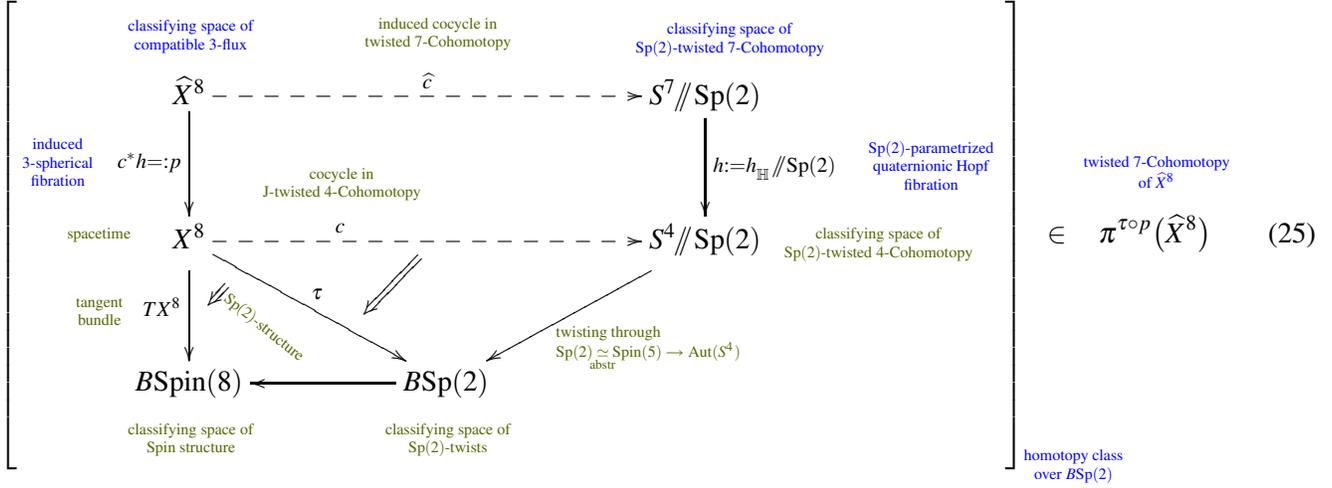
\begin{equation}
  \label{ExtendedSpacetimeIntro}
  \hspace{-8mm}
  \left[
  \;\;\;\;\;\;\;\;\;\;\;\;\;\;\;\;
  \raisebox{55pt}{
  \xymatrix@R=23pt@C=5em{
    \overset{
      \mathclap{
      \mbox{
        \tiny
        \color{blue}
        \begin{tabular}{c}
          classifying space of
          \\
          compatible 3-flux
          \\
          $\phantom{a}$
        \end{tabular}
        }
      }
    }{
      \widehat X^8
    }
    \ar@{-->}[rr]^-{
      \overset{
        \mathclap{
        \;\;\;\;\;
        \mbox{
          \tiny
          \color{brown}
          \begin{tabular}{c}
            induced cocycle in
            \\
            twisted 7-Cohomotopy
            \\
            $\phantom{a}$
          \end{tabular}
        }
        }
      }
      {
        \widehat c
      }
    }
    \ar[d]_-{
      \mathllap{
      \mbox{
        \tiny
        \color{blue}
        \begin{tabular}{c}
          induced
          \\
          3-spherical
          \\
          fibration
        \end{tabular}
      }
      \;
      c^\ast h
      =:
      p
      }
    }
    &&
    \overset{
      \mathrlap{
      \;\;\;\;\;\;\;\;\;
      \mbox{
        \hspace{-.6cm}
        \tiny
        \color{blue}
        \begin{tabular}{c}
          classifying space of
          \\
          $\mathrm{Sp}(2)$-twisted 7-Cohomotopy
          \\
          $\phantom{a}$
        \end{tabular}
      }
      }
    }
    {
      S^7 \!\!\sslash\! \mathrm{Sp}(2)
    }
    \ar[d]^-{
      h :=
      h_{{}_{\mathbb{H}}}
      \sslash \mathrm{Sp}(2)
      \;\;
      \mbox{
        \tiny
        \color{blue}
        \begin{tabular}{c}
          $\mathrm{Sp}(2)$-parametrized
          \\
          quaternionic Hopf
          \\
          fibration
        \end{tabular}
      }
    }
    \\
    \mathllap{
      \mbox{
        \tiny
        \color{greeni}
        \begin{tabular}{c}
          spacetime
          \\
          $\phantom{a}$
        \end{tabular}
      }
      \;\;\;\,
    }
    {
      X^8
    }
    \ar[dr]^-{\tau}^>>>>>>>>>{\ }="t"_<<<<<{\ }="s2"
    \ar[d]_{
      \mathllap{
        \mbox{
          \tiny
          \color{greeni}
          \begin{tabular}{c}
            tangent
            \\
            bundle
          \end{tabular}
        }
      }
      T X^8
    }^>>>>>>>{\ }="t2"
    \ar@{-->}[rr]^<<<<<<<<<<<<<<<<<{
      \overset
      {
        \mathclap{
        \mbox{
          \tiny
          \color{brown}
          \begin{tabular}{c}
            cocycle in
            \\
            J-twisted 4-Cohomotopy
            \\
            $\phantom{a}$
          \end{tabular}
        }
        }
      }
      {
        c
      }
    }_-{\ }="s"
    &&
    {
      S^4 \!\!\sslash\! \mathrm{Sp}(2)
    }
    \mathrlap{
      \mbox{
        \tiny
        \color{greeni}
        \begin{tabular}{c}
          classifying space of
          \\
          $\mathrm{Sp}(2)$-twisted 4-Cohomotopy
        \end{tabular}
      }
    }
    \ar[dl]^-{\;\;\;
      \mbox{
        \tiny
        \color{greeni}
        \begin{tabular}{l}
          twisting through
          \\
          $\mathrm{Sp}(2)
          \underset{\mathclap{\mathrm{abstr}}}{\simeq}
          \mathrm{Spin}(5) \to \mathrm{Aut}(S^4)$
        \end{tabular}
      }
    }
    \\
    \underset{
      \mathclap{
      \mbox{
        \tiny
        \color{greeni}
        \begin{tabular}{c}
          $\phantom{a}$
          \\
          classifying space of
          \\
          Spin structure
        \end{tabular}
      }
      }
    }{
      B \mathrm{Spin}(8)
    }
    &
    \underset{
      \mathclap{
      \mbox{
        \tiny
        \color{greeni}
        \begin{tabular}{c}
          $\phantom{a}$
          \\
          classifying space of
          \\
          $\mathrm{Sp}(2)$-twists
        \end{tabular}
      }
      }
    }{
      B \mathrm{Sp}(2)
    }
    \ar[l]
    \ar@{=>} "s"; "t"
    \ar@{=>}^-{
      \begin{rotate}{-39.6}
      \mbox{
        \tiny
        \color{greeni}
        $\mathrm{Sp}(2)$-structure
      }
      \end{rotate}
    }
      "s2"; "t2"
  }
  }
  \!\!\!\!\! \right]_{
    \!\!\!\!\!\!\!\!\!\!
    \mathrlap{
    \mbox{
      \color{blue}
      \tiny
      \begin{tabular}{c}
        homotopy class
        \\
        over $B \mathrm{Sp}(2)$
      \end{tabular}
    }
    }
  }
  \;\;\;\;\;\;\in\;\;\;\;
  \overset{
    \mathclap{
    \mbox{
      \tiny
      \color{blue}
      \begin{tabular}{c}
        twisted 7-Cohomotopy
        \\
        of $\widehat X^8$
        \\
        $\phantom{a}$
      \end{tabular}
    }
    }
  }{
    \pi^{\tau \circ p}\big( \widehat X^8 \big)
  }
\end{equation}

\noindent Then:

\noindent The {\bf pullback 3-spherical fibration} over spacetime
$$
  \widehat X^8
   \;\coloneqq\;
  c^\ast
  \big(
    S^7 \!\!\sslash\! \mathrm{Sp}(2)
  \big)
$$
{\bf carries a universal 3-flux} $H_3^{\mathrm{univ}}$
which trivializes the 4-flux relative to its background value
\begin{equation}
  \label{3Flux}
  d\, H_3^{\mathrm{univ}}
  \;=\;
  p^\ast \widetilde G_4 - \tfrac{1}{4}p_1(\nabla)
  \,.
\end{equation}
\noindent Moreover, the {\bf 7-Cohomotopical Chern character} of
$[\hat c]$,
hence the flux forms underlying
$[\hat  c]$ by \eqref{CompatibleRationalTwists},
as indicated on the left in the following diagram
$$
  
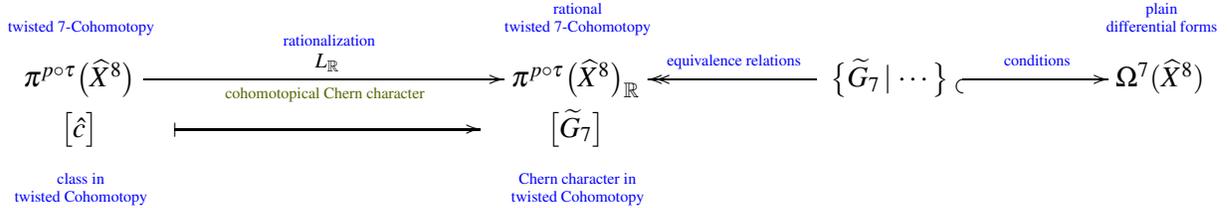
\xymatrix@R=-25pt@C=5em{
    \overset{
      \mathclap{
      \mbox{
        \tiny
        \color{blue}
        \begin{tabular}{c}
          $\phantom{a}$
          \\
          twisted 7-Cohomotopy
          \\
          $\phantom{a}$
        \end{tabular}
      }
      }
    }{
      \pi^{p \circ \tau}\big( \widehat X^8 \big)
    }
    \ar[rr]^{
      \overset{
        \mathclap{
        \mbox{
          \tiny
          \color{blue}
          \begin{tabular}{c}
            rationalization
            \\
          \end{tabular}
        }
        }
      }{
        L_{\mathbb{R}}
      }
    }_-{
      \mathclap{
      \mbox{
        \tiny
        \color{brown}
        cohomotopical
        Chern character
      }
      }
    }
    &&
    \overset{
      \mathclap{
      \mbox{
        \tiny
        \color{blue}
        \begin{tabular}{c}
          rational
          \\
          twisted 7-Cohomotopy
          \\
          $\phantom{a}$
        \end{tabular}
      }
      }
    }{
      \pi^{p \circ \tau}\big( \widehat X^8 \big)_{\mathbb{R}}
    }
    \ar@{<<-}[r]^-{
      \mbox{
        \tiny
        \color{blue}
        \begin{tabular}{c}
          equivalence relations
        \end{tabular}
      }
    }
    &
    \;
    \big\{
      \widetilde G_7
        \,\vert\,
      \cdots
    \big\}
    \;
    \ar@{_{(}->}[r]^-{
      \mbox{
        \tiny
        \color{blue}
        \begin{tabular}{c}
          conditions
        \end{tabular}
      }
    }
    &
    \overset{
      \mathclap{
      \mbox{
        \tiny
        \color{blue}
        \begin{tabular}{c}
          plain
          \\
          differential forms
          \\
          $\phantom{a}$
        \end{tabular}
      }
      }
    }{
      \Omega^7(\widehat X^8)
    }
    \\
    \underset{
      \mbox{
        \tiny
        \color{blue}
        \begin{tabular}{c}
          $\phantom{a}$
          \\
          class in
          \\
          twisted Cohomotopy
        \end{tabular}
      }
    }{
      \big[ \hat c \big]
    }
    \ar@{|->}[rr]
    &&
    \underset{
      \mbox{
        \tiny
        \color{blue}
        \begin{tabular}{c}
          $\phantom{a}$
          \\
          Chern character in
          \\
          twisted Cohomotopy
        \end{tabular}
      }
    }{
    \big[
      \widetilde G_7
    \big]
    }
  }
$$
satisfy this condition:
\\
\noindent The  {\bf shifted 7-flux} form

\vspace{-1.2cm}

\begin{equation}
  \label{PageCharge7Form}
  \widetilde G_7
  \;=\;
  \underset{
    \mathclap{
    \mbox{
      \tiny
      \color{blue}
      \begin{tabular}{c}
        $\phantom{a}$
        \\
        naive 7-flux
      \end{tabular}
    }
    }
  }{
    p^\ast G_7
  }
  \;
  +
  \;
  \underset{
    \mathclap{
    \mbox{
    \tiny
    \color{blue}
    \begin{tabular}{c}
      shift by
      \\
      Hopf-Whitehead term
    \end{tabular}
    }
    }
  }{
  \underbrace{
  \tfrac{1}{2}
  \overset{
    \mathclap{
    \!\!
    \mbox{
      \tiny
      \color{blue}
      \begin{tabular}{c}
        3-flux
      \end{tabular}
    }
    }
  }{
    H^{\mathrm{univ}}_3
  }
  \wedge
  p^\ast
  \overset{
    \mathclap{
    \;\;\;
    \mbox{
      \tiny
      \color{blue}
      \begin{tabular}{c}
        shifted
        4-flux
      \end{tabular}
    }
    }
  }{
    \widetilde G_4
  }
  }
  }
\end{equation}
{\bf is closed up to the Euler 8-form}

\vspace{-.7cm}

\begin{equation}
  \label{PageChargeFormIsClosed}
  d\, \widetilde G_7
    \;=\;
  - \tfrac{1}{2}p^\ast \rchi_8(\nabla)
\end{equation}
and
{\bf half-integral} on every 7-sphere $S^7 \overset{i}{\to} \widehat X^8$:

\vspace{-.7cm}

\begin{equation}
  \label{HalfIntegral7Flux}
  2
  \int_{{}_{S^7}} i^\ast \tilde G_7
  \;\in\;
  \mathbb{Z}
  \,.
\end{equation}

\noindent Finally, consider the case when:

\begin{enumerate}[{\bf(i)}]
\vspace{-2mm}
\item Our manifold is the complement in a closed 8-manifold
of a finite set of disjoint open balls,
i.e. of a tubular neighbourhood $\mathcal{N}$
around a finite set $\{x_1, x_2, \cdots\}$ of points:
\begin{equation}
  \label{ComplementOfM2s}
  X^8
  \;=\;
  \overset{
    \mathclap{
    \mbox{
      \tiny
      \color{blue}
      \begin{tabular}{c}
        closed
        \\
        manifold
        \\
        $\phantom{a}$
      \end{tabular}
    }
    \;\;\;
    }
  }{
    X^8_{\mathrm{clsd}}
  }
   \setminus
  \overset{
    \mathclap{
    \mbox{
      \tiny
      \color{blue}
      \begin{tabular}{c}
        tubular
        \\
        neighbourhood
        \\
        $\phantom{a}$
      \end{tabular}
    }
    }
  }{
    \mathcal{N}_{
      \underset{
        \mathclap{
        \mbox{
          \tiny
          \color{blue}
          \begin{tabular}{c}
            $\phantom{a}$
            \\
            $\mathllap{\mbox{\tiny \color{blue}around}\; }$
            points
            in $X^{8}_{\mathrm{clsd}}$
          \end{tabular}
        }
        }
      }{
        \{x_1, x_2, \cdots \}
      }
    }
  }
  \;\;\;\;\;\;
  \Rightarrow
  \;\;\;\;\;\;
  \overset{
    \mathclap{
    \mbox{
      \tiny
      \color{blue}
      \begin{tabular}{c}
        boundary
        \\
        of $X^8$
        \\
        $\phantom{a}$
      \end{tabular}
    }
    }
  }{
    \partial X^8
  }
  \;\simeq\;
  \underset{
    \{x_1, x_2, \cdots\}
  }{\sqcup}
  \overset{
    \mathclap{
    \mbox{
      \tiny
      \color{blue}
      \begin{tabular}{c}
        sphere
        \\
        around $x_i$
        \\
        $\phantom{a}$
      \end{tabular}
    }
    }
  }{
    S^7
  }
\end{equation}
This implies that $X^8$ is a manifold
with boundary a disjoint union of 7-spheres.

\vspace{-2mm}
\item Such that
the corresponding extended spacetime fibration $\widehat X^8\to X^8$
\eqref{ExtendedSpacetimeIntro} admits a global section;
hence,
equivalently, such that the given cocycle in twisted 4-Cohomotopy
lifts through the quaternionic Hopf fibration to a cocycle in
twisted 7-Cohomotopy:
\begin{equation}
  \label{LiftsThrowughhHSp2}
  \raisebox{20pt}{
  \xymatrix@C=3em@R=3em{
    &
    \overset{
      \mathclap{
      \mbox{
        \tiny
        \color{blue}
        \begin{tabular}{c}
          classifying space of
          \\
          compatible 3-flux
          \\
          $\phantom{a}$
        \end{tabular}
      }
      }
    }{
      \widehat X^8
    }
    \ar[d]^-{
      p := c^\ast(h)
      \;\;\;
      \mbox{
        \tiny
        \color{blue}
        \begin{tabular}{c}
          induced
          \\
          3-spherical
          \\
          fibration
        \end{tabular}
      }
    }
    \\
    X^8
    \ar@/^1.2pc/@{-->}[ur]^-{
      \overset{
        \mathllap{
        \mbox{
          \tiny
          \color{brown}
          \begin{tabular}{c}
            global section of
            \\
            3-spherical fibration
            \\
            $\phantom{a}$
          \end{tabular}
        }
        \!\!
        }
      }{
        i
      }
    }
    \ar@{=}[r]
    &
    X^8
  }
  }
  \;\;\;\;\;\;
  \Leftrightarrow
  \;\;\;\;\;\;\;\;\;
  \raisebox{20pt}{
  \xymatrix@C=3em@R=4em{
    & &
    S^7 \!\!\sslash\! \mathrm{Sp}(2)
    \ar[d]^-{
      h := h_{\mathbb{H}}\sslash \mathrm{Sp}(2)
      \;\;\;
      \mbox{
        \tiny
        \color{blue}
        \begin{tabular}{c}
          $\mathrm{Sp}(2)$-parametrized
          \\
          quaternionic Hopf
          \\
          fibration
        \end{tabular}
      }
    }
    \\
    X^8
    \ar@{-->}@/^2pc/[urr]^-{
      \overset{
        \mathllap{
        \mbox{
          \tiny
          \color{brown}
          \begin{tabular}{c}
            lift to
            cocycle in
            \\
            J-twisted 7-Cohomotopy
            \\
            $\phantom{a}$
          \end{tabular}
        }
        \!\!\!
        }
      }{
        \hat c
      }
    }_>>>>>>{\ }="s"
    \ar[rr]_-{
      \underset{
        \mbox{
          \tiny
          \color{blue}
          \begin{tabular}{c}
            $\phantom{a}$
            \\
            cocycle in
            \\
            J-twisted 4-Cohomotopy
          \end{tabular}
        }
      }{
        c
      }
    }^-{\ }="t"
    &&
    S^4 \!\!\sslash\! \mathrm{Sp}(2)
    \ar@{=>}|-{
      \mbox{
        \tiny
        \color{blue}
        homotopy
      }
    } "s"; "t"
  }
  }
\end{equation}
Here the choice of points in \eqref{ComplementOfM2s}
matters only in so far as a sufficient number of points
has to be removed for a lifted cocycle $\hat c$ \eqref{LiftsThrowughhHSp2}
to exist at all.

\end{enumerate}

\noindent By \eqref{3Flux} this lift exhibits a {\bf 4-fluxless}
background at least at the level of integral cohomology.
In order to refine this to 4-fluxlessness at the finer level
of (stable) Cohomotopy, we observe the following:
\begin{enumerate}[{\bf (i)}]
\vspace{-2mm}
\item  Since the 7-sphere is parallelizable,
upon
restriction of $\hat c$ \eqref{LiftsThrowughhHSp2} to
the boundary $\partial X^8 \overset{i}{\longrightarrow} X^8$
\eqref{ComplementOfM2s} the twist vanishes, and we are
left with a pair of compatible cocycles in plain Cohomotopy theory
as in \eqref{hHSendsDegree7ToDegree4}:
$$
  \xymatrix@C=5em{
    &&
    S^7
    \ar[d]^-{
      h_{\mathbb{H}}
      \;\;
      \mbox{
        \tiny
        \color{blue}
        \begin{tabular}{c}
          plain
          \\
          quaternionic
          \\
          Hopf fibration
        \end{tabular}
      }
    }
    \\
    \mathllap{
      \underset{
        \{x_1, x_2, \cdots\}
      }{\sqcup}
      \overset{
        \mathclap{
        \mbox{
          \tiny
          \color{blue}
          \begin{tabular}{c}
            boundary 7-spheres
            \\
            $\phantom{a}$
          \end{tabular}
        }
        }
      }{
        S^7
      }
      \simeq
      \;
    }
    \partial X^8
    \ar[rr]_-{
      \underset{
        \mathclap{
        \mbox{
          \tiny
          \color{brown}
          \begin{tabular}{c}
            underlying boundary
            \\
            4-Cohomotopy cocycle
          \end{tabular}
        }
        }
      }{
        (h_\mathbb{H})_\ast
        \hat c_{\vert \partial X^8}
      }
    }^{\ }="t"
    \ar@{-->}[urr]^{
      \overset{
        \mathllap{
        \mbox{
          \tiny
          \color{brown}
          \begin{tabular}{c}
            boundary restriction of
            \\
            twisted 7-Cohomotopy cocycle
            \\
            $\phantom{a}$
          \end{tabular}
        }
        \!\!
        }
      }{
        \hat c_{\vert \partial X^8}
      }
    }_{\ }="s"
    &&
    S^4
    \ar@{=>} "s"; "t"
  }
$$
\vspace{-5mm}
\item By \eqref{QuaternionicHopfFibration},
cocycles in
stable 7-Cohomotopy have no side-effect in stable 4-Cohomotopy,
hence remain stably {\bf cohomotopically 4-fluxless}
precisely if they are multiples of {\color{magenta}24}:
$$
  \mbox{For}
  \;
  [c_1], [c_2]
  \;\in\;
  \overset{
    \mathclap{
    \mbox{
      \tiny
      \color{blue}
      \begin{tabular}{c}
        7-Cohomotopy
        \\
        $\phantom{a}$
      \end{tabular}
    }
    }
  }{
    \pi^7
    \big(
      \partial X^8
    \big)
  }
  \;\;\;\;\mbox{we have}\;\;\;\;
  \left\{
  \begin{array}{crlrl}
    &
    (h_{\mathbb{H}})_\ast [c_1]
      &=&
    (h_{\mathbb{H}})_\ast [c_2]
    &
    \in\;
    \overset{
      \mathclap{
      \mbox{
        \tiny
        \color{blue}
        \begin{tabular}{c}
          stable
          4-Cohomotopy
        \end{tabular}
      }
      }
    }{
      \mathbb{S}^4
      \big(
        \partial X^8
      \big)
    }
    \\
    \Leftrightarrow
    \\
    &
    [c_1]
      &=_{\mathrm{mod}\;{\color{magenta}24}}&
    [c_2]
    &
    \in\;
    \underset{
      \mathclap{
      \mbox{
        \tiny
        \color{blue}
        \begin{tabular}{c}
          stable
          7-Cohomotopy
        \end{tabular}
      }
      }
    }{
      \mathbb{S}^7
      \big(
        \partial X^8
      \big)
    }
  \end{array}
  \right.
$$
This means that the {\bf unit charge} of
a lift $\hat c$ in \eqref{LiftsThrowughhHSp2},
as seen by stable Cohomotopy, is {\color{magenta}24}.
In view of \eqref{HalfIntegral7Flux}
this says that
the
{\bf cohomotopically normalized 7-flux} of $X^8$ is
\begin{equation}
  \label{CohomotopicalNormalization}
  N_{\mathrm{M2}}
  \;\coloneqq\;
  \tfrac{-1}{12} \int_{{}_{X^8}} i^\ast d \widetilde{G}_{7}
  \;=\;
  \tfrac{-1}{12} \int_{{}_{\partial X^8}} i^\ast \widetilde{G}_{7}
  \,.
\end{equation}
\end{enumerate}
Our final result is that:
\\
this
{\bf equals the $I_8$-number} \eqref{SP1Sp2Structure} of the manifold:
\begin{equation}
  \label{NIsI}
  N_{\mathrm{M2}}
  \;=\;
  I_8[X^8]
  \,.
\end{equation}

\newpage

\section{J-Twisted Cohomotopy theory}
\label{Cohomotopy}

We now introduce our twisted generalized cohomology theory, J-twisted Cohomotopy theory,
and discuss some general properties.

\subsection{Twisted Cohomotopy}
\label{CohomotopyTwisted}

The non-abelian cohomology theory
(see \cite{NSS12}, following \cite{SatiSchreiberStasheff12})
represented by the $n$-spheres is called
\emph{Cohomotopy}, going  back to \cite{Borsuk36}\cite{Spanier49}.
Hence for $X$ a topological space, its \emph{Cohomotopy set} in degree $n$
is
\begin{equation}
  \label{CohomotopyUntwisted}
  \pi^n(X)
  \;=\;
  \pi_0
  \mathrm{Maps}
  \big(
    X, S^n
  \big)
  \;=\;
  \Big\{
    \xymatrix{
      X
      \ar@{-->}[rr]^{
        \mbox{
          \tiny
          \color{blue}
          \begin{tabular}{c}
            cocycle in
            \\
            Cohomotopy
          \end{tabular}
        }
      }
      &&
      S^n
    }
  \Big\}_{\big/\sim}.
\end{equation}
A basic class of examples is Cohomotopy
of a manifold $X$ in the same degree as the dimension $\mathrm{dim}(X)$
of that manifold. The classical \emph{Hopf degree theorem}
(see, e.g., \cite[IX (5.8)]{Kosinski93}, \cite[7.5]{Kobin16})
says that for $X$ connected, orientable and closed,
this is canonically identified with the integral cohomology
of $X$, and hence with the integers

\begin{equation}
  \label{HopfDegreeTheorem}
  \xymatrix{
    \pi^{n}(X)
    \ar[rrr]_-{\simeq}^-{ S^{n} \to K(\mathbb{Z},n) }
    &&&
    H^{n}(X; \mathbb{Z})
    \simeq
      \mathbb{Z}
    \,,
    \phantom{AA}
    \mbox{for $n = \mathrm{dim}(X)$.}
  }
\end{equation}
In its generalization to the
\emph{equivariant Hopf degree theorem} this becomes
a powerful statement about equivariant Cohomotopy theory
and thus, via \hyperlink{HypothesisH}{\it Hypothesis H},
about brane charges at orbifold singularities \cite{ADE}.
We discuss this in detail elsewhere \cite{SS19a}\cite{SS19b}.

\medskip

Here we generalize ordinary Cohomotopy
\eqref{CohomotopyUntwisted} to \emph{twisted Cohomotopy} (Def. \ref{TwistedCohomotopy} below),
following the general theory of non-abelian (unstable)
twisted cohomology theory \cite[Sec. 4]{NSS12}.
\footnote{All constructions here are homotopical,
in particular all group actions, principal bundles, etc.
are ``higher structures up to coherent homotopy'',
in a sense that has been made completely rigorous
via the notion of $\infty$-groups,
and their $\infty$-actions on $\infty$-principal bundles
\cite{NSS12}.
But the pleasant upshot of this theory is that
when homotopy coherence is systematically accounted for,
then higher structures behave in all general ways
as ordinary structures, for instance in that homotopy
pullbacks satisfy the same structural pasting laws as
ordinary pullbacks. Beware, this means in particular
that all our equivalences
are weak homotopy equivalences
(even when we denote them as equalities),
and
that all our commutative diagrams are commutative up to
specified homotopies (even when we do not display these).
}
Generally, Cohomotopy in degree $n$ may by twisted by $\mathrm{Aut}(S^n)$-principal $\infty$-bundles,
for $\mathrm{Aut}(S^n) \subset \mathrm{Maps}(S^n, S^n)$ the automorphism $\infty$-group
of $S^n$ inside the mapping space from $S^n$ to itself.

\medskip
A well-behaved subspace of twists comes from
$\mathrm{O}(n+1)$-principal bundles,
or their associated real vector bundles of rank $n+1$, under the inclusion
\begin{equation}
\label{Jn}
\xymatrix{
  \widehat J_n
  \;:\;
  \mathrm{O}(n+1)
  \;
  \ar@{^{(}->}[r]
  &
  \mathrm{Aut}(S^n)
  \;\ar@{^{(}->}[r]
  &
  \mathrm{Maps}(S^n, S^n)
  },
\end{equation}
which witnesses the canonical action of orthogonal transformations in Euclidean space
$\mathbb{R}^{n+1}$ on the unit  sphere $S^n = S(\mathbb{R}^{n+1})$. The restriction
of these to $O(n)$-actions
$$
\xymatrix{
  J_n
  \;:\;
  \mathrm{O}(n)
 \; \ar@{^{(}->}[r]&
  \mathrm{O}(n+1)
  \; \ar@{^{(}->}[r]^-{\widehat J_n} &
  \mathrm{Maps}(S^n, S^n)
  }
$$
are known as the unstable \emph{J-homomorphisms} \cite{Whitehead42}
(see \cite{Kosinski93}\cite{Mathew12} for expositions).
By general principles \cite{NSS12}, the homotopy quotient $S^n \sslash O(n+1)$ of
$S^n$ by the action via $\widehat{J}_n$  is canonically equipped with a map $\tilde J_n$ to the
classifying space $B O(n+1)$, such that the homotopy fiber is $S^n$:
$$
  \xymatrix@R=1.3em{
    S^n
    \ar[r]
    &
    S^n \sslash {\rm O}(n+1)
    \ar[d]
    \\
    &
    B {\rm O}(n+1)\;.
  }
$$
One may think of this as the universal spherical fibration which is the $S^n$-fiber $\infty$-bundle
associated to the universal $O(n+1)$-principal bundle via the homotopy action $\hat{J}_n$.

\begin{defn}[Twisted Cohomotopy]
  \label{TwistedCohomotopy}
Given a map $\tau : X \to B \mathrm{O}(n+1)$, we define the
\emph{$\tau$-twisted cohomotopy set of $X$} to be
\begin{equation}
  \begin{aligned}
  \pi^{\tau}(X)
  & \;: =\;
  \left\{
  \raisebox{20pt}{
  \xymatrix{
    &&&
    S^n \!\sslash\! \mathrm{O}(n+1)
    \ar[d]^-{ \widetilde J_{\mathrm{O}(n+1)} }
    \\
    X
    \ar@{-->}[urrr]^{
      \mbox{
        \tiny
        \color{blue}
        \begin{tabular}{c}
          cocycle in
          \\
          twisted
          \\
          Cohomotopy
        \end{tabular}
      }
    }_>>>>>>>{\ }="s"
    \ar[rrr]^-{\tau}_-{
      \mbox{
        \tiny
        \color{blue}
        twist
      }
    }^{\ }="t"
    &&&
    B \mathrm{O}(n+1)
    \ar@{==>} "s"; "t"
  }
  }
  \right\}_{\raisemath{16pt}{\Bigg/\sim}}
  \\
  & =
  \left\{
  \raisebox{20pt}{
  \xymatrix@C=3.5em{
    &&
    E
    \ar[d]
    \ar[r]
    \ar@{}[dr]|-{
      \mbox{
        \tiny
        (pb)
      }
    }
    &
    S^n \!\sslash\! \mathrm{O}(n+1)
    \ar[d]^-{ \widetilde J_{\mathrm{O}(n+1)} }_<{\ }="s"
    \\
    X
    \ar@{=}[rr]
    \ar@{-->}[urr]^{
      \mbox{
        \tiny
        \color{blue}
        \begin{tabular}{c}
          cocycle in
          \\
          twisted
          \\
          Cohomotopy
        \end{tabular}
      }
    }
    &&
    X
    \ar[r]^-{\tau}_-{
      \mbox{
        \tiny
        \color{blue}
        twist
      }
    }^{\ }="t"
    &
    B \mathrm{O}(n+1)
    \ar@{==>} "s"; "t"
  }
  }
  \right\}_{\raisemath{16pt}{\Bigg/\sim}}
  \end{aligned}
\end{equation}
Here in the second line, $E \to X$ denotes the $n$-spherical fibration
classified by $\tau$ and the universal property of the
homotopy pullback shows that cocycles in $\tau$-twisted equivariant
Cohomotopy are equivalently sections of this $n$-spherical fibration.
\end{defn}
\begin{remark}[Notation]
Here the notation $\pi^{\tau}\big(X\big)$  is motivated, as usual in twisted cohomology,
from thinking of  the map $\tau$ as encoding, in particular, also the degree $n \in \mathbb{N}$.
\end{remark}

\begin{remark}[Cohomotopy twist by Spin structure]
  \label{CohomotopyTwistBySpinStructure}
  In applications, the twisting map $\tau$ is often equipped with a lift through some stage of the
  Whitehead tower of $B {\rm O}(n+1)$, notably with a lift through   $B \mathrm{SO}(n+1)$ or
  further to $B \mathrm{Spin}(n+1)$
  $$
    \xymatrix{
      X
      \ar@/_1.5pc/[rr]_-{ \tau }
      \ar[r]^-{ \widehat{\tau} }
      &
      B \mathrm{Spin}(n+1)
      \ar[r]
      &
      B \mathrm{O}(n+1)
    }\!.
  $$
  In this case, due to the homotopy pullback diagram
  $$
    \xymatrix@C=6em{
      S^n \!\sslash\! \mathrm{Spin}(n+1)
      \ar[r]
      \ar[d]
      \ar@{}[dr]|-{
        \mbox{
          \tiny
          (pb)
        }
      }
      &
      S^n \!\sslash\! \mathrm{O(n+1)}
      \ar[d]^-{\tilde J_{\mathrm{O}(n+1)}}
      \\
      B \mathrm{Spin}(n+1)
      \ar[r]
      &
      B \mathrm{O}(n+1)
    }
  $$
  the twisted cohomotopy set from Def. \ref{TwistedCohomotopy}
  is equivalently given by
\begin{equation}
  \label{SpinTwistedCohomotopy}
  \pi^{\tau}(X)
  \;\simeq\;
  \left\{
  \raisebox{20pt}{
  \xymatrix{
    &&&
    S^n \!\sslash\! \mathrm{Spin}(n+1)
    \ar[d]^-{ {\widetilde J}_{\mathrm{Spin}(n+1)} }
    \\
    X
    \ar@{-->}[urrr]^{
      \mbox{
        \tiny
        \color{blue}
        \begin{tabular}{c}
          cocycle in
          \\
          twisted
          \\
          Cohomotopy
        \end{tabular}
      }
    }_>>>>>>>{\ }="s"
    \ar[rrr]^-{\widehat \tau}_-{
      \mbox{
        \tiny
        \color{blue}
        twist
      }
    }^{\ }="t"
    &&&
    B \mathrm{Spin}(n+1)
    \ar@{==>} "s"; "t"
  }
  }
  \right\}_{\raisemath{16pt}{\Bigg/\sim}}
\end{equation}
Most of the examples in  \cref{TwistedCohomotopyInDegrees}
and \cref{CancellationFromCohomotopy} arise in this form.
\end{remark}

In order to extract differential form data (``flux densities'') from cocycles in twisted Cohomotopy,
in Prop. \ref{SectionsOfRationalSphericalFibrations} below, we consider rational twisted
Cohomotopy (Def. \ref{RationalTwistedCohomotopy}) below. A standard reference on the rational
homotopy theory involved is \cite{FHT00}.
Reviews streamlined to our context can be found
in \cite[Appendix A]{FSS16a}\cite{GaugeEnhancement}.

\begin{defn}[Chern character on twisted Cohomotopy]
  \label{RationalTwistedCohomotopy}
  We write
  $
    \xymatrix@C=1em{
      \pi^{\tau}( X )
      \ar[rr]^-{ L_{\mathbb{R}} }
      &&
      \pi^\tau( X )_{\mathbb{R}}
     }
  $
  for the rationalization of twisted Cohomotopy
  to \emph{rational twisted Cohomotopy},
  given by applying rationalization to all spaces
  and maps involved in a twisted Cohomotopy cocycle.
\end{defn}

\medskip

We now characterize cocycles in rational twisted Cohomotopy in terms of
differential form data
(which will be the corresponding ``flux density'' in \cref{CancellationFromCohomotopy}).

\begin{prop}[Differential form data underlying twisted Cohomotopy]
  \label{SectionsOfRationalSphericalFibrations}
  Let $X$ be a smooth manifold which is simply connected
  (see Remark \ref{SimplyConnected} below)
  and $\tau : X \to B \mathrm{SO}(n+1)$
  a twisting for Cohomotopy in degree $n$, according to   Def. \ref{TwistedCohomotopy}.
  Let $\nabla_\tau$ be any connection on the real vector bundle $V$ classified by $\tau$
with Euler form $\rchi_{2k+2}( \nabla_\tau)$  (see \cite[below (7.3)]{MathaiQuillen86}\cite[2.2]{Wu06}).

  \item {\bf (i)} {\bf If $n = 2k+1$ is odd}, $n\geq 3$,
  a cocycle defining a class in the
  rational $\tau$-twisted Cohomotopy   of $X$ (Def. \ref{RationalTwistedCohomotopy}) is
  equivalently given by a differential $2k+1$-form   $G_{2k+1} \in \Omega^{2k+1}(X)$ on
  $X$ which trivializes the negative of the Euler form
  \begin{equation}
    \label{FormDataForOddSphereCase}
    \pi^\tau
    (
     X
    )_{\mathbb{R}}
    \;\simeq\;
    \big\{
      G_{2k+1}
      \;\vert\;
      d \, G_{2k + 1}
      \;=\;
      -
      \rchi_{2k+2}( \nabla_\tau )
    \big\}_{\big/\sim}  \,.
  \end{equation}

 \item {\bf (ii) If $n = 2k$ is even}, $n \geq 2$, a cocycle defining a class in  the
  rational $\tau$-twisted Cohomotopy  of $X$ (Def. \ref{RationalTwistedCohomotopy}) is
  given by a pair of  differential forms $G_{2k}\in \Omega^{2k}(X)$ and $G_{4k-1}\in \Omega^{4k-1}(X)$
  such that
  \begin{align}
    d G_{2k}
    &=0;
    \qquad
    \pi^*G_{2k}
    = \tfrac{1}{2}\rchi_{2k}( \nabla_{\hat{\tau}})
  \\
  \label{FHTRelationII}
  d
    2G_{4k-1}
   &=
   - G_{2k} \wedge G_{2k}
   +\tfrac{1}{4}p_{{}_{k}}( \nabla_{\tau}),
  \end{align}
 where $p_{{}_{k}}( \nabla_{\tau})$ is the $k$-th Pontrjagin form of $\nabla_\tau$,
 $\pi\colon E\to X$ is the unit sphere bundle
 over $X$ associated with $\tau$,
 $\hat{\tau}\colon E\to B \mathrm{SO}(n)$ classifies
 the vector bundle $\widehat{V}$
 on $E$ defined by the splitting
 $\pi^*V=\mathbb{R}_E\oplus \widehat{V}$ associated
 with the tautological section of $\pi^*V$ over $E$,
 and $\nabla_{\hat{\tau}}$
 is the induced connection on $\widehat{V}$.
 That is,
 \begin{equation}
   \label{RationallyTwistedCohomtopyInEvenDegree}
    \pi^\tau
    (
      X
    )_{\mathbb{R}}
    \;\simeq\;
    \left\{
      \big( G_{2k}, 2G_{4k-1} \big)
      \;\Big\vert\;
      \begin{aligned}
        d \, G_{2k} & = 0\,,
          \phantom{AA}
          \pi^*G_{2k}
          =
          \tfrac{1}{2}\rchi_{2k}( \nabla_{\hat{\tau}})
        \\
        d \, 2G_{4k - 1} & =
          - G_{2k} \wedge G_{2k}
          + \tfrac{1}{4}p_{{}_{k}}( \nabla_{\tau})
      \end{aligned}
    \right\}_{\raisemath{16pt}{\Big/\sim}}.
  \end{equation}
\end{prop}
\begin{proof}
  By the assumption that the smooth manifold $X$   is simply connected, it has a Sullivan model
  dgc-algebra $\mathrm{CE}\big( \mathfrak{l}X \big)$ and we may compute the rational twisted
  Cohomotopy by choosing a Sullivan model $\mathfrak{l} E$  for the spherical fibration
  classified by $\tau$.  By definition of rational twisted Cohomotopy, we are interested in the set of
  homotopy equivalence classes of dgca morphisms
  $ \mathrm{CE}(\mathfrak{l} E)\to  \mathrm{CE}(\mathfrak{l} X)$ that are sections
  of the morphism $ \mathrm{CE}(\mathfrak{l} X)\to  \mathrm{CE}(\mathfrak{l} E)$
  corresponding to the projection $E\to X$.
  The Sullivan model model for $E$ is well known. We recall from
  {\cite[Sec. 15, Example 4]{FHT00}}:

\begin{enumerate}[{\bf (I).}]
\vspace{-3mm}
\item
The Sullivan model
  for the total space of a $2k+1$-spherical fibration $E \to X$ is of the form
  \begin{equation}
    \label{FibS2k1Model}
    \mathrm{CE}(\mathfrak{l} E )
    \;=\;
    \mathrm{CE}(\mathfrak{l}X)
    \otimes
    \mathbb{R}[ \omega_{2k+1}]
    /
    \left(
      d \, \omega_{2k+1}  = -  c_{2k+2}
   \right)
    ,
  \end{equation}
  where

  \begin{enumerate}
  \item
  $c_{2k+2} \in \mathrm{CE}\big( \mathfrak{l}X \big)$ is some element in the base algebra, which by
     \eqref{FibS2k1Model} is closed and so it represents a rational cohomology class
     $$
       [
         c_{2k+2}
       ]
       \;=\;
       H^{2k+2}
       (
         X; \mathbb{R}
       ).
     $$
  This class classifies the spherical fibration, rationally. Moreover, if the spherical fibration $E \to X$ happens to
  be the unit sphere bundle $E = S(V)$ of a real vector bundle $V \to X$, then
the class of $c_{2k+2}$ is the rationalized Euler class $\rchi_{2k+2}(V)$ of $V$:
    \begin{equation}
      \label{EulerC2k2}
      [
        c_{2k+2}
      ]
      \;=\;
      \rchi_{2k+2}(V)
      \;\in\;
      H^{2k+2}
      (
        X; \mathbb{R}
      )\;.
    \end{equation}
  \item
    and in this case,
    under the quasi-isomorphism
    $\mathrm{CE}(\mathfrak{l} E)\to \Omega^\bullet_{\mathrm{dR}}(E)$
    the new generator $\omega_{2k+1}$ corresponds to
    a differential form that evaluates to
    the unit volume on
    each $(2k+1)$-sphere fiber:
    \begin{equation}
      \label{OddDegreeGeneratorNormalization}
      \big\langle
        \omega_{2k+1},
        [S^{2k+1}]
      \big\rangle
      \;=\;
      1\;.
    \end{equation}
    (This is not stated in {\cite[Sec. 15, Example 4]{FHT00}},
    but follows with \cite{Chern44},
    see \cite[Ch. 6.6, Thm. 6.1]{Walschap04}.)
 \end{enumerate}
The morphism $ \mathrm{CE}(\mathfrak{l} X)\to  \mathrm{CE}(\mathfrak{l} E )$ is the obvious inclusion,
so a section is completely defined by the image of $\omega_{2k+1}$ in $ \mathrm{CE}(\mathfrak{l} X)$. This image
will be an element $g_{2k+1}\in  \mathrm{CE}(\mathfrak{l} X)$ such that $dg_{2k+1}=c_{2k+2}$, and every such
element defines a section $ \mathrm{CE}(\mathfrak{l} E)\to  \mathrm{CE}(\mathfrak{l} X)$ and so a
cocycle in rational twisted cohomotopy. Under the quasi-isomorphism
$\mathrm{CE}(\mathfrak{l} X)\to \Omega^\bullet_{\mathrm{dR}}(X)$ defining $\mathrm{CE}(\mathfrak{l} X )$
as a Sullivan model of $X$, the element $c_{2k+2}$ is mapped to a closed differential form $\rchi_{2k+2}( \nabla_\tau)$
representing the Euler class $\rchi_{2k+2}(V)$ of $V$, and so $g_{2k+1}$ corresponds to a differential form $G_{2k+1}$ on
$X$ with $dG_{2k+1}=\rchi_{2k+2}(\nabla_\tau)$.

\vspace{-2mm}
 \item
  The Sullivan model for  the total space of $2k$-spherical fibration $E \to X$ is of the form\footnote{
    There is an evident sign typo in the statement
    (but not in the proof) of
    {\cite[Sec. 15, Example 4]{FHT00}}
    with respect to equation \eqref{FibS2k1Model}:
    The standard fact that the Euler class squares to the
    top Pontrjagin class means that there must be the relative
    minus sign in \eqref{FibS2k1Model}, which is exactly
    what the proof of {\cite[Sec. 15, Example 4]{FHT00}}
    actually concludes.
  }
  \begin{equation}
    \label{FibS2kModel}
    \mathrm{CE}(\mathfrak{l} E )
    \;=\;
    \mathrm{CE}(\mathfrak{l}X)
    \otimes
    \mathbb{R}\big[ \omega_{2k}, \omega_{4k-1} \big]
    \big/
    \left(\!\! \small
      \begin{array}{ll}
        d \, \omega_{2k} & = 0
        \\
        d 2\omega_{4k-1}
         &
         =
         - \omega_{2k} \wedge \omega_{2k}
         +
         c_{4k}
      \end{array}
   \!\! \right)
    ,
  \end{equation}
  where

  \vspace{-3mm}
\begin{enumerate}
 \item
    $c_{4k} \in \mathrm{CE}( \mathfrak{l}X )$ is some element in the base algebra, which by
     \eqref{FibS2kModel} is closed and represents the rational cohomology class of the cup square
     of the class of $\omega_{4k}$:
     $$
       [
         c_{4k}
       ]
       \;=\;
       [
         \omega_{2k}
       ]^2
       \;\in\; \
       H^{4k}
       (
         X; \mathbb{R}
       ).
     $$
    This class classifies the spherical fibration, rationally.
 \item
    under the quasi-isomorphism $\mathrm{CE}(\mathfrak{l} E)\to \Omega^\bullet_{\mathrm{dR}}(E)$
    the new generator $\omega_{2k}$ corresponds to a closed differential form that restricts to the volume form on the
    $2k$-sphere fibers $S^{2k} \simeq E_x \hookrightarrow E$ over each point $x \in X$:
    \begin{equation}
      \label{EvenDegreeGeneratorNormalization}
      \big\langle \omega_{2k}, [S^{2k}]\big\rangle
      \;=\;
      1
      \;.
    \end{equation}

\end{enumerate}
Note that the element $[
         \omega_{2k}
       ]^2$ is a priori an element in $H^{4k}
       (
         E, \mathbb{R}
       )$. By writing $[
         c_{4k}
       ]
       =
       [
         \omega_{2k}
       ]^2
       \;\in\; \
       H^{4k}
       (
         X; \mathbb{R}
       )$ we mean that $[
         \omega_{2k}
       ]^2$ is actually the pullback of the element $[
         c_{4k}
       ]$ via the projection $\pi\colon E\to X$.

Moreover, if the spherical fibration $\pi\colon E \to X$ happens to be the unit sphere bundle $E = S(V)$
of a real vector bundle $V \to X$, then the tautological section of $\pi^*V$ defines a splitting
$\pi^*V=\mathbb{R}_E\oplus \widehat{V}$ and
\begin{enumerate}
\vspace{-2mm}
  \item
        the class of $\omega_{2k}$ is half the rationalized Euler class $\rchi_{2k}(\widehat V)$ of $\widehat V$:
    \begin{equation}
      \label{EulerC2k2}
      [
        \omega_{2k}
      ]
      \;=\;
      \tfrac{1}{2}\rchi(\widehat{V})
      \;\in\;
      H^{2k}
      (
        E; \mathbb{R}
      )\;.
    \end{equation}

\vspace{-2mm}
  \item
    the class of $c_{4k}$ is one fourth the rationalized $k$-th Pontrjagin class $p_k(V)$ of $V$:
    \begin{equation}
      \label{c4pk}
      [
        c_{4k}
      ]
      \;=\;
      \tfrac{1}{4} p_k(V)
      \;\in\;
      H^{4k}
      (
        X;  \mathbb{R}
      )\;.
    \end{equation}
\end{enumerate}
\end{enumerate}
\vspace{-2mm}
The second equation is actually a consequence of the first one and of the naturality and  multiplicativity
of the total rational Pontrjagin class:
\[
\pi^*p_k(V)=p_k(\mathbb{R}_E\oplus \widehat{V})=
p_k(\widehat{V})=\rchi_{2k}(\widehat{V})^2.
\]
Reasoning as in the odd sphere bundles case, a section of $ \mathrm{CE}(\mathfrak{l} X )\to  \mathrm{CE}(\mathfrak{l} E )$,
and so a cocycle in rational twisted cohomotopy, is the datum of elements $g_{2k}, g_{4k-1}\in  \mathrm{CE}(\mathfrak{l} X)$
such that $dg_{2k}=0$ and $d 2g_{4k-1}=- g_{2k} \wedge g_{2k}+c_{4k}$.
Under the quasi-isomorphism $\mathrm{CE}(\mathfrak{l} E)\to \Omega^\bullet_{\mathrm{dR}}(E)$, the element $g_{2k}$,
seen as an element in $\mathrm{CE}(\mathfrak{l} E)$, is mapped to a closed differential form $\frac{1}{2}\rchi_{2k}
( \nabla_{\hat{\tau}})$ representing 1/2 the Euler class $\rchi_{2k}(\widehat{V})$ of $\widehat{V}$, while under the
quasi-isomorphism $\mathrm{CE}(\mathfrak{l} X )\to \Omega^\bullet_{\mathrm{dR}}(X)$ the element $c_{4k}$ is mapped
to a closed differential form $\frac{1}{4}p_{{}_{k}}\big( \nabla_{\hat{\tau}}\big)$ representing 1/4 the $k$-th Pontrjagin class
$\frac{1}{4}p_{{}_{k}}(V)$ of $V$. Therefore, the quasi-isomorphism
$\mathrm{CE}(\mathfrak{l} X)\to \Omega^\bullet_{\mathrm{dR}}(X)$ turns the  elements $g_{2k}$ and $g_{4k-1}$
into differential forms $G_{2k}$ and $G_{4k-1}$ on $X$,
subject to the identities
$dG_{2k}=0$, $\pi^*G_{2k}=\frac{1}{2}\rchi_{2k}( \nabla_{\hat{\tau}})$,
and
$
  d 2G_{4k-1}
  =
  - G_{2k} \wedge G_{2k}
  +
  \tfrac{1}{4}p_{{}_{k}}( \nabla_{\hat{\tau}})$.
\end{proof}

\medskip

\begin{remark}[Simply-connectedness assumption]
  \label{SimplyConnected}
  The assumption in Prop. \ref{SectionsOfRationalSphericalFibrations}  that $X$ be simply connected is
  just to ensure  the existence of a Sullivan model for $X$, as used in the proof. (It would be sufficient to
  assume, for that purpose, that the fundamental group is nilpotent).  If $X$ is not simply connected and
  not even nilpotent, then a similar statement  about differential form data underlying twisted Cohomotopy
  cocycles on $X$ will still hold, but statement and proof will  be much more involved.  Hence we assume
  simply connected $X$ here only for convenience,  not for fundamental reasons.
  A direct consequence of this assumption, which will play a role in
  \cref{CancellationFromCohomotopy},  is that, by the Hurewicz theorem and the universal
  coefficient theorem, the degree 2 cohomology of  $X$ with coefficients in $\mathbb{Z}_2$ is given by:
  \begin{equation}
    \label{H2Dual}
    H^2(
      X; \mathbb{Z}_2
    )
    \;\simeq\;
    \mathrm{Hom}_{\mathrm{Ab}}
    \big(
            H_2(X,\mathbb{Z}),
      \;
      \mathbb{Z}_2
         \big)\,.
  \end{equation}
\end{remark}

\medskip

\subsection{Twisted Cohomotopy via topological $G$-structure}
\label{ViaReductionOfStructureGroups}

We discuss how cocycles in $J$-twisted Cohomotopy are equivalent to
choice of  certain topological $G$-structures (Prop. \ref{TwistedCohomotopyIsReductionOfStructureGroup} below).

\medskip

The following fact plays a crucial role throughout:
\begin{lemma}[Homotopy actions and reduction of structure group]
  \label{fibBiota}
  Let $G$ be a topological group and $V$ any topological
  space.

  \vspace{-1mm}
  \item {\bf (i)} Then for every homotopy-coherent action of $G$ on
  $V$, the corresponding homotopy quotient $V \!\sslash\! G$
  forms a homotopy fiber sequence of the form
  $$
    \xymatrix{
      V
      \ar[r]
      &
      V \!\sslash\! G
      \ar[r]
            & B G
    }
  $$
  and, in fact, this association establishes an
  equivalence between homotopy $V$-fibrations over
  $B G$ and homotopy coherent actions of $V$ on $G$.

\vspace{-1mm}
   \item {\bf (ii)} In particular, if $\iota : H \hookrightarrow G$
  is an inclusion of topological groups,
  then the homotopy fiber of the induced map
  $B \iota$ of classifying spaces is the coset
  space $G/H$:
  $$
    \xymatrix{
      G/H
      \ar[r]^-{\mathrm{fib}}
      &
      B H
      \ar[r]^-{ B \iota }
        & B G
    }
  $$
  thus exhibiting the weak homotopy equivalence
  $
    \big( G/H\big) \!\sslash\! G
    \;\simeq\;
    B H
     $.
\end{lemma}
\begin{proof}
  This equivalence goes back to \cite{DDK}.  A modern account which generalizes to
  geometric situations (relevant for refinement  of all constructions here to
  differential cohomology) is in \cite[Sec. 4]{NSS12}.
  When the given homotopy-coherent action of the topological group $G$ on $V$ happens to
  be given by an actual topological action we may use the Borel construction to represent the
  homotopy quotient. For the case of $H \hookrightarrow G$ a topological subgroup inclusion,
  we may compute as follows:
  $$
  \begin{aligned}
    B H
    & \simeq
    \ast \times_{H}
    E H
    \\
    & \simeq
    \ast
      \times_{H}
    E G
    \\
    & \simeq
    \ast
      \times_{H}
    \left(
      G
        \times_{G}
      E G
    \right)
    \\
    & \simeq
    \left(
      \ast
        \times_{H}
      G
    \right)
      \times_{G}
    E G
    \\
    & \simeq
    \left(
      G/H
    \right)
      \times_{G}
    E G
    \\
    & \simeq
    (G/H) \sslash G
    \,.
  \end{aligned}
$$
Here the first weak equivalence is the usual definition of the classifying space, while the second uses that one may
take a universal $H$-bundle $E H$, up to weak homotopy equivalence, any contractible space with free $H$-action, hence in
particular $E G$. The third line uses that
$G$ is the identity under
Cartesian product followed by the quotient by the diagonal $G$-action.
\end{proof}

\begin{prop}[Twisted cohomotopy cocycle is reduction of structure group]
 \label{TwistedCohomotopyIsReductionOfStructureGroup}
Cocycles in twisted Cohomotopy (Def. \ref{TwistedCohomotopy})
are equivalent to
choices of topological $G$-structure for
$G = O(n) \hookrightarrow O(n+1)$:
$$
  \pi^{\tau}(X)
  \;=\;
  \left\{
  \raisebox{23pt}{
  \xymatrix{
    &&&
    B \mathrm{O}(n)
    \ar[d]
    \\
    X
    \ar@{-->}[urrr]^{
      \mbox{
        \tiny
        \color{blue}
        \begin{tabular}{c}
          cocycle in
          \\
          twisted
          \\
          Cohomotopy
        \end{tabular}
      }
    }
    \ar[rrr]^-{\tau}_-{
      \mbox{
        \tiny
        \color{blue}
        twist
      }
    }
    &&&
    B \mathrm{O}(n+1)
  }
  }
  \right\}_{\raisemath{16pt}{\Bigg/\sim}}
$$
Moreover, if the twist is itself is factored through $B \mathrm{Spin}(n+1)$
as in Remark \ref{CohomotopyTwistBySpinStructure}, then
$\tau$-twisted Cohomotopy is equivalent to reduction along $\mathrm{Spin}(n) \hookrightarrow \mathrm{Spin}(n+1)$:
$$
  \pi^{\tau}(X)
  \;=\;
  \left\{
  \raisebox{23pt}{
  \xymatrix{
    &&&
    B \mathrm{Spin}(n)
    \ar[d]
    \\
    X
    \ar@{-->}[urrr]^{
      \mbox{
        \tiny
        \color{blue}
        \begin{tabular}{c}
          cocycle in
          \\
          twisted
          \\
          Cohomotopy
        \end{tabular}
      }
    }
    \ar[rrr]^-{\widehat \tau}_-{
      \mbox{
        \tiny
        \color{blue}
        twist
      }
    }
    &&&
    B \mathrm{Spin}(n+1)
  }
  }
  \right\}_{\raisemath{16pt}{\Bigg/\sim}}
$$
Generally, if there is a coset realization of an $n$-sphere
$S^n \simeq G/H$ and the twist is factored through $G$-structure,
then $\tau$-twisted Cohomotopy is further reduction to
topological $H$-structure:
$$
  \pi^{\tau}(X)
  \;=\;
  \left\{
  \raisebox{23pt}{
  \xymatrix{
    &&&
    B H
    \ar[d]
    \\
    X
    \ar@{-->}[urrr]^{
      \mbox{
        \tiny
        \color{blue}
        \begin{tabular}{c}
          cocycle in
          \\
          twisted
          \\
          Cohomotopy
        \end{tabular}
      }
    }
    \ar[rrr]^-{\widehat \tau}_-{
      \mbox{
        \tiny
        \color{blue}
        twist
      }
    }
    &&&
    B G
  }
  }
  \right\}_{\raisemath{16pt}{\Bigg/\sim}}
$$
\end{prop}
\begin{proof}
This follows by applying Lemma \ref{fibBiota} and using the fact that
$S^n \simeq O(n+1)/O(n)$.
\end{proof}

\begin{remark}[Cohomotopy twists from coset space structures on spheres]
\label{CohomotopyTwistsFromCosetSpaceStructuresOnSpheres}
\item {\bf (i)}   Prop. \ref{TwistedCohomotopyIsReductionOfStructureGroup}
  say that
  for each topological coset space structure on an $n$-sphere
  $
    S^n \;\simeq\; G/H
  $
  the corresponding $G$-twisted Cohomotopy
  (Def. \ref{TwistedCohomotopy})
  classifies reduction
  to topological $H$-structure.

\item {\bf(ii)} Coset space structures on $n$-spheres
come in three infinite series
and a few exceptional cases:

\hypertarget{TableS}{}
\begin{center}
\begin{tabular}{|l|l|}
  \hline
  {\bf Spherical coset spaces} & \cite{MontgomerySamelson43}, see \cite[p.2]{GrayGreen70}
  \\
  \hline
  \hline
  \rule{0pt}{2.5ex}  $S^{n-1} \simeq \mathrm{Spin}(n)/\mathrm{Spin}(n-1)$
  &
  \multirow{3}{*}{
    \begin{tabular}{l}
      standard,
      \\
      e.g. \cite[17.1]{BorelSerre53}
    \end{tabular}
  }
  \\
  \cline{1-1}
  \rule{0pt}{2.5ex}  $S^{2n-1} \simeq \mathrm{SU}(n)/\mathrm{SU}(n-1)$
  &
  \\
  \cline{1-1}
  \rule{0pt}{2.5ex}  $S^{4n-1} \simeq \mathrm{Sp}(n)/\mathrm{Sp}(n-1)$
  &
  \\
  \hline
  \hline
   \rule{0pt}{2.5ex}  $S^7 \simeq \mathrm{Spin}(7)/\mathrm{G}_2$
  &
  \begin{tabular}{l}
  \cite[Thm. 3]{Varadarajan01}
  \end{tabular}
  \\
  \hline
   \rule{0pt}{2.5ex}  $S^7 \simeq \mathrm{Spin}(6)/\mathrm{SU}(3)$
  &
  \begin{tabular}{l}
  by $\mathrm{Spin}(6) \simeq \mathrm{SU}(4)$
  \end{tabular}
  \\
  \hline
   \rule{0pt}{2.5ex}  $S^7 \simeq \mathrm{Spin}(5)/\mathrm{SU}(2)$
  &
  \begin{tabular}{l}
  by $\mathrm{Spin}(5) \simeq \mathrm{Sp}(2)$
  \\
  and $\mathrm{SU}(2) \simeq \mathrm{Sp}(1)$
  \\
  \cite{ADP}\cite{DNP1}
  \end{tabular}
  \\
  \hline
  \rule{0pt}{2.5ex}  $S^6 \simeq \mathrm{G}_2/\mathrm{SU}(3)$
  &
  \begin{tabular}{l}
  \cite{FuIsh55}
  \end{tabular}
  \\
  \hline
\end{tabular}

\vspace{.1cm}

{\bf Table S.}
{\small
Coset space structures on topological $n$-spheres.}
\end{center}

\item {\bf (iii)} Assembling these for the case of the 7-sphere, we interpret the
result in terms of special holonomy and $G$-structures
corresponding to consecutive reductions.

\end{remark}

Using this, the following construction is a rich source of twisted Cohomotopy
cocycles:

\begin{lemma}[Classifying maps to normal bundles via twisted Cohomotopy]
  \label{ClassifyingMapsToNormalBundles}
  Let $Y$ be a manifold, $n \in \mathbb{N}$
  a natural number and
  $
    \xymatrix{
      Y
      \ar[r]^-{ N }
      &
      B \mathrm{O}(n+1)
    }
  $
  a classifying map.

\vspace{-2mm}
  \item {\bf (i)} If $X \overset{\pi}{\to} Y$ denotes the $n$-spherical
  fibration classified by $N$, hence,
  by Remark \ref{CohomotopyTwistsFromCosetSpaceStructuresOnSpheres},
  the homotopy pullback in the following diagram:
  \begin{equation}
    \label{SphericalFibrationByPullback}
    \xymatrix@R=1.4em@C=4em{
      X
      \ar[rr]_-{ c }
      \ar@{}[ddrr]|-{
        \mbox{
          \tiny
         {\rm  (pb)}
        }
      }
      \ar@/^1pc/[rrr]^-{ T_{Y} X }
      \ar[dd]_-{\pi}
      &&
      S^n \!\sslash\! \mathrm{O}(n+1)
      \ar[r]_-{\simeq}
      \ar[dd]
      &
      B \mathrm{O}(n)
      \ar[ddl]
      \\
      \\
      Y
      \ar[rr]_-{N}
      &&
      B \mathrm{O}(n+1)
    }
  \end{equation}
  then the total top horizontal map is equivalently the classifying
  map for the vertical tangent bundle $T_Y X$, of $X$ over $Y$, as shown.

\vspace{-1mm}
  \item {\bf (ii)} In particular, if
  $
    Y
    :=
    \Sigma \times \mathbb{R}_{\gt 0}
  $
  is the Cartesian product of some manifold $\Sigma$
  with a real ray,  so that each fiber of $\pi$ over $\Sigma$ is identified with
  the Cartesian space $\mathbb{R}^{n+1}$ with the origin removed
  $$
    S^n \times \mathbb{R}_{\gt 0}
    \;\simeq\;
    \mathbb{R}^{n + 1} \setminus \{0\}
  $$
  then the pullback map $c$ in \eqref{SphericalFibrationByPullback}
  is a cocycle in $N\circ \pi$-twisted Cohomotopy on $X$, according to
  Def. \ref{TwistedCohomotopy}:
  \begin{equation}
    \label{TwistedCohomotopyCocycleFromN}
    \xymatrix@C=3em{
      &&
      S^n \!\sslash\! \mathrm{O}(n+1)
      \ar[d]
      \\
      X
      \ar[rr]|-{\, N \circ \pi  \,}
      \ar[rrd]_-{ T X }
      \ar@{-->}[urr]^-{ c }
      &&
      B \mathrm{O}(n+1)
      \ar[d]
      \\
      &&
      B \mathrm{O}( \mathrm{dim}(\Sigma) + n + 1 )
    }
  \end{equation}
\end{lemma}

\medskip

\subsection{Twisted Cohomotopy in degrees 4 and 7 combined}
\label{TwistedCohomotopyInDegrees}

We discuss here twisted Cohomotopy in degree 4 and 7 jointly,
related by the quaternionic Hopf fibration $h_{\mathbb{H}}$.
This requires first determining the space of twists that
are compatible with $h_{\mathbb{H}}$, which is the
content of Prop. \ref{EquihH} and Prop. \ref{hHsslash}
below.
This yields the scenario
of incremental $G$-structures
shown in \hyperlink{FigureT}{\it Figure T}.
The twists that appear are subgroups
of $\mathrm{Spin}(8)$ related by triality
(Prop. \ref{QuaternionicSubgroupTriality} below), and
in fact the classifying space for the C-field
implied by \hyperlink{HypothesisH}{\it Hypothesis H}
comes out to be the homotopy-fixed locus of triality.

\medskip
It will be useful to have the following notation for a basic but crucial operation on $\mathrm{Spin}$ groups:
\begin{defn}[Central product of groups]
\label{Def-dot}
Given a tuple of groups $G_1, G_2, \cdots, G_n$,
each equipped with
a central $\mathbb{Z}_2$-subgroup inclusion
$
  \mathbb{Z}_2 \simeq \{1, -1\} \subset Z(G_i) \subset G_i
$,
we write
\begin{equation}
  \label{QuotientDot}
  G_1
    \boldsymbol{\cdot}
  G_2
   \boldsymbol{\cdot}
  \cdots
   \boldsymbol{\cdot}
  G_{n-1}
   \boldsymbol{\cdot}
  G_n
  \;:=\;
  \big(
     G_1 \times G_2 \times \cdots \times G_n
  \big)/_{\mathrm{diag}} \mathbb{Z}_2
\end{equation}
for the quotient group of their direct product group
by the corresponding diagonal $\mathbb{Z}_2$-subgroup:
$$
\xymatrix{
  \{
    (1,1, \cdots, 1),
    \;
    (-1, -1, \cdots, -1)
  \}
  \; \ar@{^{(}->}[r] &
  G_1 \times G_2 \times \cdots \times G_n
 }.
$$
Just to save space we will sometimes suppress the dots and write
$
  G_1 G_2 :=
  G_1
    \boldsymbol{\cdot}
  G_2
$,
etc.
\end{defn}
\begin{example}[Central product of symplectic groups]
\label{CentralProductOfSymplecticGroups}
The notation in Def. \ref{Def-dot} originates in \cite{Alekseevskii}\cite{Gray} for the examples
\begin{equation}
  \label{SpnSp1}
  \mathrm{Sp}(n)\boldsymbol{\cdot}\mathrm{Sp}(1)
  \;:=\;
  \big(
    \mathrm{Sp}(n) \times \mathrm{Sp}(1)
  \big)/ \{(1,1), (-1,-1)\} .
\end{equation}
For $n \geq 2$ this is such that a
$\mathrm{Sp}(n) \!\boldsymbol{\cdot}\! \mathrm{Sp}(1)$-structure on a $4n$-dimensional
manifold is equivalently a quaternion-K{\"a}hler structure
\cite{Salamon82}.
Specifically, for $n = 2$ there is a canonical subgroup inclusion
\begin{equation}
  \label{ActionOfSp2Sp1}
  \raisebox{30pt}{
  \xymatrix@C=2em@R=.1pt{
    &&
    \mathrm{Spin}(8)
    \ar@{->>}[dddd]
    \\
    \\
    \\
    \\
    \mathrm{Sp}(2)
     \boldsymbol{\cdot}
    \mathrm{Sp}(1)
    \ar@{^{(}->}[uuuurr]|-{   }
   \; \ar@{^{(}->}[rr]
    &&
    \mathrm{SO}(8)
    \mathrlap{
      \;
      \simeq \mathrm{SO}( \mathbb{H}^2)
    }
    \\
\;\;\; \;\;\;   (
      A, q
    )
    \;\;\;\;\;\;\;\;\ar@{|->}[rr]
    &&
  \;\;\;\; \;\;\;\;\;\;\; ( x \mapsto A \cdot x \cdot \overline{q} )
  }
  }
\end{equation}
given by identifying elements of
$\mathrm{Sp}(2)$
as quaternion-unitary $2 \times 2$-matrices $A$, elements of $\mathrm{Sp}(1)$ as
multiples of the $2 \times 2$ identity matrix by unit quaternions $q$, and acting with
such pairs by quaternionic matrix conjugation on elements
$x \in \mathbb{H}^2 \simeq_{\mathbb{R}} \mathbb{R}^8$ as indicated.
This lifts to an inclusion into $\mathrm{Spin}(8)$ through the defining double-covering map
(see \cite[2.]{CV97}).
Notice that reversing the $\mathrm{Sp}$-factors gives an
isomorphic group, but a different subgroup inclusion
\begin{equation}
  \label{ActionOfSp1Sp2}
  \raisebox{30pt}{
  \xymatrix@C=3em@R=.1pt{
    &&
    \mathrm{Spin}(8)
    \ar@{->>}[dddd]
    \\
    \\
    \\
    \\
    \mathrm{Sp}(1)
     \boldsymbol{\cdot}
    \mathrm{Sp}(2)
    \ar@{^{(}->}[uuuurr]|-{   }
    \;\ar@{^{(}->}[rr]
    &&
    \mathrm{SO}(8)
    \mathrlap{
      \;
      \simeq \mathrm{SO}( \mathbb{H}^2)
    }
    \\
  \;\;\; \;\;\;   (
      q, A
    )
  \;\;\;\;\;\;\;\;\ar@{|->}[rr]
    &&
   \;\;\;\; \;\;\;\;\;\;\;  ( x \mapsto q \cdot x \cdot \overline{A})
  }
  }
\end{equation}
For more on this see Prop. \ref{QuaternionicSubgroupTriality} below.
\end{example}
\begin{example}[Central product of Spin groups]
\label{ProductSpinNotSubgroup}
For $n_1, n_2 \in \mathbb{N}$, we have the central product (Def. \ref{Def-dot})
of the corresponding Spin groups
\begin{equation}
  \label{Spinn1Spinn2}
  \mathrm{Spin}(n_1)  \boldsymbol{\cdot} \mathrm{Spin}(n_2)
  \;:=\;
  \big(
    \mathrm{Spin}(n_1) \times \mathrm{Spin}(n_2)
  \big) / \{ (1,1), (-1,-1)\}\;.
\end{equation}
(This notation is used for instance in \cite[p. 9]{McInnes99}  \cite[Prop. 17.13.1]{HN}.)
Here the canonical subgroup inclusions of Spin groups
$
    \mathrm{Spin}(n)
    \overset{ \iota_{n} }{\longhookrightarrow}
    \mathrm{Spin}(n + k)
$
induce a canonical subgroup inclusion of \eqref{Spinn1Spinn2}
into $\mathrm{Spin}(n_1 + n_2)$:
\begin{equation}
  \label{Spinn1Spinn2SubgroupInclusion}
  \xymatrix@C=.1em@R=-1pt{
    &
 \;\;\;\;  (\alpha, \beta)
    \;\;\;\ar@{|->}[rr]
    &&
   \iota_{n_1}(\alpha) \cdot \iota_{n_2}(\beta)
    \\
    \big( \mathbb{Z}_2\big)_{\mathrm{diag}}
   \; \ar@{^{(}->}[r]^-{\mathrm{ker}}
    &
    \mathrm{Spin}(n_1)
    \times
    \mathrm{Spin}(n_2)
    \ar[rr]
    \ar[dddd]_-{\mathrm{quot}}
    &&
    \mathrm{Spin}(n_1 + n_2)\;.
    \\
    \\
    \\
    \\
    &
 \;\;\;\;\;\;\;\;\;\;\;\;\;\;   \mathrm{Spin}(n_1)
    \boldsymbol{\cdot}
    \mathrm{Spin}(n_2)
  \;\;\;\;\;\;\;\;\;\;\;\;\;\;
    \ar@{^{(}->}[uuuurr]
  }
\end{equation}
Notice that these groups sit in short exact sequences as follows:
\begin{equation}
  \label{Spinn1Spinn2ExactSequences}
  \xymatrix{
    1
    \ar[r]
    &
    \mathrm{Spin}(n_1)
   \; \ar@{^{(}->}[rr]^-{ \iota_{n_1} }
    &&
    \mathrm{Spin}(n_1)
      \boldsymbol{\cdot}
    \mathrm{Spin}(n_2)
    \ar@{->>}[rr]^-{ \mathrm{pr}_{n_2} }
    &&
    \mathrm{SO}(n_2)
    \ar[r]
    &
    1
  }.
\end{equation}
\end{example}

For low values of $n_1, n_2$ there are exceptional isomorphisms between the groups
\eqref{SpnSp1} and \eqref{Spinn1Spinn2} as abstract groups, but as subgroups under
the inclusions \eqref{ActionOfSp2Sp1} and \eqref{Spinn1Spinn2SubgroupInclusion}
these are different. This is the content
of Prop. \ref{QuaternionicSubgroupTriality} below.
First we record the following, for later use:

\begin{defn}[Universal class of central products]
  \label{Epsilon}
  For $n_1, n_2 \in \mathbb{N}$,
  write
  $$
    \varpi
    \;\in\;
    H^2\big(
      B
      (
        \mathrm{Spin}(n_1)
        \boldsymbol{\cdot}
        \mathrm{Spin}(n_2)
      )
      ;\,
      \mathbb{Z}_2
    \big)
  $$
  for the \emph{universal characteristic class}   on the classifying space of the central product
  Spin group (Def. \ref{ProductSpinNotSubgroup})  which is the pullback of the second Stiefel-Whitney class
  $w_2 \in H^2\big(
    B \mathrm{SO}(n_2),
    \mathbb{Z}_2
  \big)$
  from the classifying space of the underlying
  $\mathrm{SO}(n_2)$-bundles, via the projection \eqref{Spinn1Spinn2ExactSequences}:
  \begin{equation}
    \label{EpsilonFormula}
    \varpi
    \;:=\;
    (B \mathrm{pr}_{n_2})^\ast( w_2)
    \,.
  \end{equation}
\end{defn}
See also \cite[Def. 2.1]{Salamon82}, following \cite{MarchiafavaRomani76}.

\begin{lemma}[Obstruction to direct product structure]
  \label{ObstructionToDirectProductStructure}
  For $n_1, n_2 \in \mathbb{N}$,
  let
  $
    X
      \overset{\tau}{\longrightarrow}
      B
      \big(
        \mathrm{Spin}(n_1)
          \boldsymbol{\cdot}
        \mathrm{Spin}(n_2)
      \big)
  $
  be a classifying map for a central product Spin structure  (Def. \ref{ProductSpinNotSubgroup}).
  Then the following are equivalent:
  \begin{enumerate}[{\bf (i)}]
  \vspace{-2mm}
    \item The class $\varpi$ from Def. \ref{Epsilon} vanishes:
    $$
      \varpi(\tau) \;=\; 0
      \;\in\;
      H^2(X; \mathbb{Z}_2)
      \,.
    $$
    \vspace{-6mm}
    \item The classifying map $\tau$ has a lift to the direct product Spin structure:
   \vspace{-2mm}
    $$
      \xymatrix@C=3em{
        &&
        B
        \big(
          \mathrm{Spin}(n_1)
            \times
          \mathrm{Spin}(n_2)
        \big)
        \ar[d]
        \\
        X
          \ar[rr]_-{\tau}
          \ar@{-->}[urr]^-{\widehat\tau}
          &&
        B
        \big(
          \mathrm{Spin}(n_1)
            \!\boldsymbol{\cdot}\!
          \mathrm{Spin}(n_2)
        \big)\;.
      }
    $$
    \vspace{-4mm}
    \item The underlying $\mathrm{SO}(n_2)$-bundle
    admits Spin structure:
    $$
      \xymatrix@C=5em{
        &&
        B
          \mathrm{Spin}(n_2)
        \ar[d]
        \\
        X
          \ar[rr]_-{B\mathrm{pr}_{n_2} \circ \tau}
          \ar@{-->}[urr]^-{\widehat{ B\mathrm{pr}_{n_2} \circ \tau }}
          &&
        B
        \mathrm{SO}(n_2)\;.
      }
    $$
  \end{enumerate}
\end{lemma}
\begin{proof}
By \eqref{Spinn1Spinn2} and \eqref{Spinn1Spinn2ExactSequences}
we have the following short exact sequence of short exact sequences
of groups:
$$
  \xymatrix@R=.8em{
    &
    1
    \ar@{^{(}->}[dd]
    \ar@{^{(}->}[rr]
    &&
    \mathrm{Spin}(5)
    \ar@{=}[rr]
    \ar@{^{(}->}[dd]
    &&
    \mathrm{Spin}(5)
    \ar@{^{(}->}[dd]
    \\
    \\
    &
    \mathbb{Z}_2
    \ar@{=}[dd]
    \ar@{^{(}->}[rr]
    &&
    \mathrm{Spin}(5)
    \times
    \mathrm{Spin}(3)
    \ar@{->>}[dd]
    \ar@{->>}[rr]
    &&
    \mathrm{Spin}(5)
    \!\boldsymbol{\cdot}\!
    \mathrm{Spin}(3)
    \ar@{->>}[dd]^-{ \mathrm{pr}_3 }
    \\
    \\
    &
    \mathbb{Z}_2
    \ar@{^{(}->}[rr]
    &&
    \mathrm{Spin}(3)
    \ar@{->>}[rr]
    &&
    \mathrm{SO}(3)
  }
$$
Since the bottom left morphism is an identity, it follows
that also after passing to classifying spaces and
forming connecting homomorphisms, the corresponding
morphism on the bottom right in the following diagram is
a weak homotopy equivalence:
$$
\hspace{-1cm}
  \xymatrix@R=.8em{
    &
    B \mathbb{Z}_2
    \ar@{=}[dd]
    \ar[rr]
    &&
    B
    \big(
      \mathrm{Spin}(5)
      \times
      \mathrm{Spin}(3)
    \big)
    \ar[dd]
    \ar[rr]
    &&
    B
    \big(
      \mathrm{Spin}(5)
      \!\boldsymbol{\cdot}\!
      \mathrm{Spin}(3)
    \big)
    \ar[rr]^-{ \varpi }
    \ar[dd]^-{ B \mathrm{pr}_3 }
    &&
    B^2 \mathbb{Z}
    \ar@{=}[dd]
    \\
    \\
    &
    B \mathbb{Z}_2
    \ar[rr]
    &&
    B \mathrm{Spin}(3)
    \ar[rr]
    &&
    B \mathrm{SO}(3)
    \ar[rr]^-{ w_2 }
    &&
    B^2 \mathbb{Z}
  }
$$
By the top homotopy fiber sequence, this exhibits $\varpi$ as the obstruction to the lift from central product Spin structure
to direct product Spin structure.
\end{proof}

\begin{example}
 \label{Sp1CentralProductGroups}
Applying Def. \ref{Def-dot} to three copies of $\mathrm{Sp}(1)$ yields the group
\begin{equation}
  \label{Sp1Sp1Sp1}
  \mathrm{Sp}(1)
    \boldsymbol{\cdot}
  \mathrm{Sp}(1)
    \boldsymbol{\cdot}
  \mathrm{Sp}(1)
  \;:=\;
  \big(
    \mathrm{Sp}(1)
    \times
    \mathrm{Sp}(1)
    \times
    \mathrm{Sp}(1)
  \big) / \big\{ (1,1,1), (-1,-1,-1) \big\}
  \,.
\end{equation}
The notation appears for instance in \cite{OrneaPiccinni01}\cite{BettiolMendes14}.
\begin{itemize}
\vspace{-2mm}
\item Observe that, due to the exceptional isomorphisms
$\mathrm{Spin}(3) \simeq \mathrm{Sp}(1)$ and
$\mathrm{Spin}(4) \simeq \mathrm{Spin}(3) \times \mathrm{Spin}(3)$
there are isomorphisms
\begin{equation}
  \label{Spin4DotSpin3}
  \mathrm{Spin}(4)\boldsymbol{\cdot}\mathrm{Spin}(3)
  \;\simeq\;
  \mathrm{Spin}(3)
    \boldsymbol{\cdot}
  \mathrm{Spin}(3)
    \boldsymbol{\cdot}
  \mathrm{Spin}(3)
  \;\simeq\;
  \mathrm{Sp}(1)
    \boldsymbol{\cdot}
  \mathrm{Sp}(1)
    \boldsymbol{\cdot}
  \mathrm{Sp}(1)
  \,.
\end{equation}
\vspace{-8mm}
\item
The group \eqref{Sp1Sp1Sp1} is acted upon via automorphisms interchange the three dot-factors
by the symmetric group on three elements:
\begin{equation}
  \label{S3ActionOnSpin4Spin3}
  \xymatrix@C=.1em{
  \ar@(lu,ld)_{ S_3 } & \big(
    \mathrm{Sp}(1)
      \boldsymbol{\cdot}
    \mathrm{Sp}(1)
      \boldsymbol{\cdot}
    \mathrm{Sp}(1)
    \big)
  }
\end{equation}
\vspace{-6mm}
\item
Beware that the central product of groups with central $\mathbb{Z}_2$-subgroup (Def. \ref{Def-dot})
is not a binary associative operation: for instance, we have
\begin{equation}
  \label{SO4AsCentralProduct}
  \mathrm{Sp}(1)
    \boldsymbol{\cdot}
  \mathrm{Sp}(1)
  \;\simeq\;
  \mathrm{Spin}(3)
    \boldsymbol{\cdot}
  \mathrm{Spin}(3)
  \;\simeq\;
  \mathrm{SO}(4)
  \,,
\end{equation}
which does not even contain the $\mathbb{Z}_2$-subgroup
anymore that one would diagonally quotient out in
\eqref{Spin4DotSpin3}, hence the would-be iterated binary expression
``$
  \big(
    \mathrm{Sp}(1)
      \boldsymbol{\cdot}
    \mathrm{Sp}(1)
  \big)
    \boldsymbol{\cdot}
  \mathrm{Sp}(1)
  $''
  does not even make sense.
Instead we have
\begin{equation}
  \label{TriplCentralProductAsCartesianFollowedByCentral}
  \mathrm{Sp}(1)
    \boldsymbol{\cdot}
  \mathrm{Sp}(1)
    \boldsymbol{\cdot}
  \mathrm{Sp}(1)
  \;\simeq\;
  \big(
    \mathrm{Sp}(1)
    \times
    \mathrm{Sp}(1)
  \big)
    \boldsymbol{\cdot}
  \mathrm{Sp}(1)
  \,.
\end{equation}
But it is useful to observe that
\begin{equation}
  \mathrm{Sp}(1)
  \;\simeq\;
  \mathrm{Sp}(1)
    \boldsymbol{\cdot}
  \mathbb{Z}_2
  \phantom{AA}
  \mbox{and}
  \phantom{AA}
  \mathrm{Sp}(1)
    \times
  \mathrm{Sp}(1)
  \;\simeq\;
  \mathrm{Sp}(1)
    \boldsymbol{\cdot}
  \mathbb{Z}_2
    \boldsymbol{\cdot}
  \mathrm{Sp}(1) \;.
\end{equation}
All of the above will play a role in Prop. \ref{hHsslash} below.
\end{itemize}
\end{example}

\begin{prop}[Triality of quaternionic subgroups of $\mathrm{Spin}(8)$]
\label{QuaternionicSubgroupTriality}
The subgroup inclusions into $\mathrm{Spin}(8)$ of
  $\mathrm{Sp}(2) \!\boldsymbol{\cdot}\! \mathrm{Sp}(1)$ via
  \eqref{ActionOfSp2Sp1},
  $\mathrm{Sp}(1) \!\boldsymbol{\cdot}\! \mathrm{Sp}(2)$ via
  \eqref{ActionOfSp1Sp2},
and
  $\mathrm{Spin}(5) \!\boldsymbol{\cdot}\! \mathrm{Spin}(3)$
  via \eqref{Spinn1Spinn2SubgroupInclusion}, represent three distinct conjugacy classes of subgroups,
and under the defining projection to $\mathrm{SO}(8)$ they map to subgroups of $\mathrm{SO}(8)$ as follows:
$$
  \xymatrix@R=.1em@C=1.5em{
    \mathrm{Sp}(1) \boldsymbol{\cdot} \mathrm{Sp}(2)
    \ar@{=}[dddddr]
    \ar@{_{(}->}[ddrr]^-{ \iota' }
    \\
    \\
    &&
    \mathrm{Spin}(8)
    \ar@{->>}[dddddr]
    &&
   \; \mathrm{Spin}(3) \boldsymbol{\cdot} \mathrm{Spin}(5)
    \ar@{->>}[dddddr]
     \ar@{_{(}->}[ll]_-{\iota}
    \\
    \\
    \mathrm{Sp}(2) \boldsymbol{\cdot} \mathrm{Sp}(1)
     \ar@{=}[dddddr]
     \ar@{^{(}->}[uurr]
       |<<<<<<<{ \phantom{AA \atop AA} }
       ^>>>>>>>>>>>>>>>>>{ \iota'' }
    \\
    &
    \mathrm{Sp}(1) \boldsymbol{\cdot} \mathrm{Sp}(2)
    \ar@{_{(}->}[ddrr]^-{ \iota' }
    \\
    \\
    &&&
    \mathrm{SO}(8)
    &&
   \; \mathrm{SO}(3) \times \mathrm{SO}(5) .
    \ar@{_{(}->}[ll]_-{\iota}
    \\
    \\
    &
    \mathrm{Sp}(2) \boldsymbol{\cdot} \mathrm{Sp}(1)
    \ar@{^{(}->}[uurr]^-{ \iota'' }
  }
$$
Moreover, the triality group $\mathrm{Out}( \mathrm{Spin}(8))$
acts transitively by permutation on the set of these three conjugacy classes.
$$
  \xymatrix@C=-12pt@R=.001pt{
    \mathrm{Spin}(8)
    \ar@/_2.9pc/@{<->}[dddddddddddddddddddd]_{\simeq}
    \ar@{<->}@/^4.9pc/[ddddddddddrrrrrrrrrrrr]^-{\simeq}_-{ \mathrm{tri} }
    \\
    \\
    \\
    \\
    &
    \mathrm{Sp}(2)\!\boldsymbol{\cdot}\! \mathrm{Sp}(1)
    \ar@{<->}@/_2pc/[dddddddddddd]_-{\simeq}
    \ar@{_{(}->}[uuuul]
    \ar@{<->}@/^3pc/[ddddddrrrrrrr]^-{\simeq}
    \\
    \\
    \\
    \\
    &
    &
    \mathrm{Sp}(1) \!\boldsymbol{\cdot}\! \mathrm{Sp}(1)\!\boldsymbol{\cdot}\!\mathrm{Sp}(1)
    \ar@{<->}[dddd]^-{\simeq}
    \ar@{_{(}->}[uuuul]
  \;\;\;\;  \ar@{<->}[ddrr]^-{\simeq}
    \\
    \\
    &
    &
    &&
  \hspace{1.2cm} \mathrm{Sp}(1)\!\boldsymbol{\cdot}\! \mathrm{Sp}(1)\!\boldsymbol{\cdot}\! \mathrm{Sp}(1)
    \;\ar@{^{(}->}[rrrr]
    &&
    {\phantom{AAAAAAAA}}
    &&
    \mathrm{Spin}(5)\!\boldsymbol{\cdot}\! \mathrm{Spin}(3)
   \; \ar@{^{(}->}[rrrr]
    &&
    {\phantom{AAAAAAAA}}
    &&
    \mathrm{Spin}(8)\;.
    \\
    \\
    &
    &
     \mathrm{Sp}(1)\!\boldsymbol{\cdot}\! \mathrm{Sp}(1)\!\boldsymbol{\cdot}\! \mathrm{Sp}(1)
    \ar@{^{(}->}[ddddl]
   \;\;\;\;\; \ar@{<->}[uurr]^-{\simeq}
    \\
    \\
    \\
    \\
    &
    \mathrm{Sp}(1)\!\boldsymbol{\cdot}\! \mathrm{Sp}(2)
    \ar@{<->}@/_3pc/[uuuuuurrrrrrr]_-{\simeq}
    \ar@{^{(}->}[ddddl]
    \\
    \\
    \\
    \\
    \mathrm{Spin}(8)
    \ar@{<->}@/_4.9pc/[uuuuuuuuuurrrrrrrrrrrr]_-{\simeq}
  }
$$
\end{prop}
\begin{proof}
  This follows by analysis of the action of triality on the
  corresponding Lie algebras;
   see \cite[Sec. 2]{CV97}, \cite[Prop. 3.3 (3)]{Ko02}.
\end{proof}

\begin{remark}[Subgroups]
  \label{SubgroupInclusions}
{\bf (i)}   For emphasis,
  notice that the subgroups appearing in Prop. \ref{QuaternionicSubgroupTriality} are all isomorphic
  as abstract groups
  $$
    \mathrm{Sp}(1)
    \boldsymbol{\cdot}
    \mathrm{Sp}(2)
    \;\simeq\;
    \mathrm{Sp}(2)
    \boldsymbol{\cdot}
    \mathrm{Sp}(1)
    \;\simeq\;
    \mathrm{Spin}(5)
    \boldsymbol{\cdot}
    \mathrm{Spin}(3)
    \;\simeq\;
    \mathrm{Spin}(3)
    \boldsymbol{\cdot}
    \mathrm{Spin}(5)
  $$
  due to the classical exceptional isomorphisms
  $$
    \mathrm{Sp}(1) \;\simeq\; \mathrm{Spin}(3)
    \,,
    \phantom{AA}
    \mathrm{Sp}(2) \;\simeq\; \mathrm{Spin}(5)
  $$
  and via the evident automorphisms that permutes central  product factors. However, when each is equipped with
  \emph{its} canononical subgroup inclusion into $\mathrm{Spin}(8)$,  via \eqref{ActionOfSp2Sp1},
  \eqref{ActionOfSp1Sp2} and \eqref{Spinn1Spinn2SubgroupInclusion}, then these are  distinct subgroups.
  Moreover, Prop. \ref{QuaternionicSubgroupTriality} says that the first three of these are even in
  distinct conjugacy classes of subgroups, while  the two $\mathrm{Spin}(3) \!\boldsymbol{\cdot}\! \mathrm{Spin}(5)$
  and   $\mathrm{Spin}(5) \!\boldsymbol{\cdot}\! \mathrm{Spin}(3)$
  are in the same conjugacy class.

 \item {\bf (ii)}  In the following, when considering these subgroup inclusions
  and their induced morphisms on classifying spaces,   we will always mean that \emph{canonical}
  inclusion of the subgroup of that name. When we need to refer to another,
  non-canonical embedding of any of these groups $G$,   then we will always make this
  explicit as a triality automorphism $G \overset{\simeq}{\to} G'$ followed by the  canonical inclusion of $G'$.
  See for instance \eqref{DiagramProvingHalfIntegralFluxQuantization} below
  for an example.
\end{remark}

For the development in \cref{CancellationFromCohomotopy} we need to know in particular how universal characteristic
classes behave under the triality automorphisms:
\begin{lemma}[Pullback of classes along triality]
  \label{PullbackOfClassesAlongTriality}
  The integral cohomology ring of $B \mathrm{Spin}(8)$
  is
  \begin{equation}
    \label{Spin8CohomologyRing}
    H^\bullet
    (
      B \mathrm{Spin}(8);
      \mathbb{Z}
    )
    \;\simeq\;
    \mathbb{Z}
    \big[
     \tfrac{1}{2}
     p_1
     ,\;
     \tfrac{1}{4}
     \big(
       p_2
       -
       \big(\tfrac{1}{2}p_1\big)^2
     \big)
     -
     \tfrac{1}{2}\rchi
     ,\;
     \rchi_8,
     \;
     \beta(w_6)
    \big]
    \big/
    \big( 2 \beta(w_6)\big)
    \,,
  \end{equation}
  where $p_k$ are Pontrjagin classes,   $\rchi_8$ is the Euler class, $w_6$ is a Stiefel-Whitney
  class, $\beta$ is the Bockstein homomorphism, so that   $W_7 := \beta(w_6)$ is an integral Stiefel-Whitney class.

  \item {\bf (i)} Under the delooping of the triality automorphism  from Prop. \ref{QuaternionicSubgroupTriality} to
  classifying spaces
  \begin{equation}
    \label{TrialityAutomorphismDelooped}
    \raisebox{20pt}{
    \xymatrix{
      B
      \big(
        \mathrm{Sp}(2)\boldsymbol{\cdot} \mathrm{Sp}(1)
      \big)
      \ar[rr]^-{\simeq}
      \ar[d]
      &&
      B
      \big(
        \mathrm{Spin}(5)\boldsymbol{\cdot} \mathrm{Spin}(3)
      \big)
      \ar[d]
      \\
      B \mathrm{Spin}(8)
      \ar[rr]^-{\simeq}_-{ B \mathrm{tri} }
      &&
      B \mathrm{Spin}(8)
    }
    }
  \end{equation}
  these classes pull back as follows:
\begin{equation}
  \label{PullbackAlongTriality}
  \big(
    B \mathrm{tri}
  \big)^\ast
  \;\colon\;
  \begin{aligned}
    \tfrac{1}{2} p_1 \quad & \longmapsto \quad \tfrac{1}{2} p_1
    \\
    \rchi_8 \quad
      & \longmapsto \quad
      -
      \tfrac{1}{4}
      \big(
        p_2
        -
        \big(\tfrac{1}{2}p_1\big)^2
      \big)
      +
      \tfrac{1}{2}\rchi_8
    \\
     \tfrac{1}{4}
     \big(
       p_2
       -
       \big(\tfrac{1}{2}p_1\big)^2
     \big)
     -
     \tfrac{1}{2}\rchi_8
     \quad
    & \longmapsto \quad
      -
      \rchi_8
  \end{aligned}
\end{equation}
\item {\bf (ii)}   Notice that, in particular,
  $$
    \big(
      (
        B \mathrm{tri}
      )^\ast
    \big)^{-1}
    \;=\;
    (
      B \mathrm{tri}
    )^\ast
    \,.
  $$
  and
  \begin{equation}
    \label{PullbackOfp2UnderTriality}
    (
      B \mathrm{tri}
    )^\ast
    \;:\;
    \tfrac{1}{4}p_2
    \;\longmapsto\;
    - \rchi_8
    +
   \big(
      \tfrac{1}{4} p_1
    \big)^2
    -
    \tfrac{1}{2}
    \Big(
     \tfrac{1}{4}
     \big(
       p_2
       -
       \big(\tfrac{1}{2}p_1\big)^2
     \big)
     -
     \tfrac{1}{2}\rchi
    \Big).
  \end{equation}
\end{lemma}
\begin{proof}
  This follows by combining \cite[Lemmas 2.5, 4.1, 4.2]{CV97},
  following \cite[Thm. 2.1]{GrayGreen70},
  and using the property $(\mathrm{tri}^\ast)^{-1} = \mathrm{tri}^\ast$,
  recalled in \cite[2.]{CV97}.
\end{proof}

Now we may have a closer look at the quaternionic Hopf fibration
$
    \xymatrix{
      S^7
      \simeq
      S( \mathbb{H}^2)
      \ar[r]^-{ h_{\mathbb{H}} }
      &
      \mathbb{H}P^1
      \simeq
      S^4
    }
  $:

\begin{prop}[Symmetries of the quaternionic Hopf fibration]
  \label{EquihH}
\item {\bf (i)} The symmetry group of  $h_{\mathbb{H}}$  and hence the group of twists for Cohomotopy
  jointly in degrees 4 and 7,
  is the group \eqref{SpnSp1},
  \begin{equation}
    \label{Sp2Sp1ActionOnS7}
    \mathrm{Sp}(2)
      \!\boldsymbol{\cdot}\!
    \mathrm{Sp}(1)
    \longhookrightarrow
    \mathrm{O}(8)\
    \,,
  \end{equation}
  with its canonical action \eqref{ActionOfSp2Sp1},
  in that this is the largest
  subgroup of $\mathrm{O}(8) \simeq \mathrm{O}(\mathbb{H}^2)$
  under which $h_{\mathbb{H}}$ is equivariant.

\item {\bf (ii)} The corresponding action on the codomain 4-sphere
$S^4 \simeq S\big(\mathbb{R}^5\big)$ is via the
canonical projection \eqref{Spinn1Spinn2ExactSequences}
to $\mathrm{SO}(5)$
\begin{equation}
  \label{Sp2Sp1ActionOnS4}
  \xymatrix{
    \mathrm{Sp}(2)
      \!\boldsymbol{\cdot}\!
    \mathrm{Sp}(1)
    \ar[r]^-{\simeq} &
    \mathrm{Spin}(5) \!\boldsymbol{\cdot}\! \mathrm{Spin}(3)
    \ar@{->>}[r]^{~~~~~~
      \mathrm{pr}_5
    }
    &
    \mathrm{SO}(5)
  }.
\end{equation}
\end{prop}
\begin{proof}
  This statement essentially
  appears as \cite[Prop. 4.1]{GluckWarnerZiller86}
  and also, somewhat more implicitly, in  \cite[p. 263]{Porteous95}.
  To make this more explicit, we may observe,
  with \hyperlink{TableT}{Table S},
  that the  quaternionic Hopf fibration
  has the following coset space description:
  \begin{equation}
    \label{CosethH}
    \raisebox{20pt}{
    \xymatrix@R=1em{
      S^3
      \ar@{=}[d]
      \ar[rr]^{ \mathrm{fib}(h_{\mathbb{H}}) }
      &&
      S^7
      \ar@{=}[d]
      \ar[rr]^{h_{\mathbb{H}}}
      &&
      S^4
      \ar@{=}[d]
      \\
      \frac{
        \mathrm{Spin}(4)
      }{
        \mathrm{Spin}(3)
      }
      \ar[rr]_-{
        \frac{\iota_{{}_{4}}}
        {
          \mathrm{id}
        }
      }
      &&
      \frac{
        \mathrm{Sp}(2)
      }{
        \mathrm{Sp}(1)
      }
      \ar[rr]_{
        \frac{
          \mathrm{id}
        }{
          q \mapsto (q,1)
        }
      }
      &&
      \frac{
        \mathrm{Sp}(2)
      }{
        \mathrm{Sp}(1) \times \mathrm{Sp}(1)
      }
    }
    }
  \end{equation}
  where
  $\iota_4
   : \mathrm{Spin}(4)
   \hookrightarrow \mathrm{Spin}(5) \simeq \mathrm{Sp}(2)$
  denotes the canonical inclusion.
  This can also be deduced from \cite[Table 1]{HatsudaTomizawa09}.
  In the octonionic case the analogous statement is noticed in
  \cite[p. 7]{OPPV12}.
\end{proof}
The following Prop. \ref{hHsslash}
gives the homotopy-theoretic version of Prop. \ref{EquihH},
which is the key for the discussion in \cref{CancellationFromCohomotopy} below.
In order to clearly bring out all subtleties, we first recall the following fact:
\begin{lemma}[${\rm Spin}(4)$-action on quaternions]
  \label{Spin4Action}
  Under the exceptional isomorphism
  $$
    \xymatrix@R=1.2em{
      \mathrm{Sp}(1)
       \times
      \mathrm{Sp}(1)
      \ar@{^{(}->}[d]
      \ar[r]^-{ \simeq }
      &
      \mathrm{Spin}(4)
      \ar@{^{(}->}[d]
      \ar@{->>}[r]
      &
      \mathrm{SO}(4)
      \ar@{^{(}->}[d]
      \\
      \mathrm{Sp}(2)
      \ar[r]_-{\simeq}
      &
      \mathrm{Spin}(5)
      \ar@{->>}[r]
      &
      \mathrm{SO}(5)
    }
  $$
  the action of $\mathrm{Sp}(1) \times \mathrm{Sp}(1)$
  on $\mathbb{R}^4 \simeq_{\mathbb{R}} \mathbb{H}$
  is the conjugation action of pairs $(q_1, q_2)$
  of unit quaternions on any quaternion $x$:
  \begin{equation}
    \label{Spin4ByConjugationAction}
    \raisebox{30pt}{
    \xymatrix@R=-2pt{
      \mathrm{Spin}(4)
        \times
      \mathbb{R}^4
      \ar[dddddd]_-{\simeq}
      \ar[rr]
      &&
      \mathbb{R}^4
      \ar[dddddd]^-{ \simeq }
      \\
      \\
      \\
      \\
      \\
      \\
      \big(
        \mathrm{Sp}(1)
        \times
        \mathrm{Sp}(1)
      \big)
      \times
      \mathbb{H}
      \ar[rr]^-{
        \mathrm{conj}(-,-)(-)
      }
      &&
      \mathbb{H}
      \\
      \big(
        (q_1,q_2), x
      \big)
      \ar@{|->}[rr]
      &&
      q_1 \cdot x \cdot \overline{q_2}
    }
    }
  \end{equation}
\end{lemma}

\begin{prop}
[The $\mathrm{Sp}(2)\boldsymbol{\cdot}\mathrm{Sp}(1)$-parametrized quaternionic Hopf fibration]
  \label{hHsslash}
  The homotopy quotient
  of the  quaternionic Hopf fibration
  $h_{\mathbb{H}}$ by its equivariance group  (Prop. \ref{EquihH}) is
  equivalently the map  of classifying spaces
  $$
    \xymatrix@R=1.2em@C=5em{
      S^7
       \!\sslash\!
      \mathrm{Sp}(2) \!\boldsymbol{\cdot}\! \mathrm{Sp}(1)
      \ar[dd]_{
        h_{\mathbb{H}}
        \sslash
        \mathrm{Sp}(2) \boldsymbol{\cdot} \mathrm{Sp}(1)
      }
      \ar@{<->}[r]^-{ \simeq }
      &
      B
      \big(
        \mathrm{Sp}(1)
          \boldsymbol{\cdot}
        \mathrm{Sp}(1)
      \big)
      \ar[dd]^{
        B
        (
          [q_1, q_2]
          \mapsto
          [q_1, q_2, q_2]
        )
      }
      \\
      \\
      S^4
       \!\sslash\!
      \mathrm{Sp}(2) \!\boldsymbol{\cdot}\! \mathrm{Sp}(1)
      \ar@{<->}[r]_-{ \simeq }
      &
      B
      \big(
        \mathrm{Sp}(1)
          \!\boldsymbol{\cdot}\!
        \mathrm{Sp}(1)
          \!\boldsymbol{\cdot}\!
        \mathrm{Sp}(1)
      \big)
    }
  $$
  which is induced by the following inclusion
  of central product groups from Example \ref{Sp1CentralProductGroups}:
  \begin{equation}
    \label{Sp1Sp1InsideSp1Sp1Sp1}
    \xymatrix@R=-2pt{
      \mathrm{Sp}(1)
        \boldsymbol{\cdot}
      \mathrm{Sp}(1)
      \ar@{^{(}->}[rr]^-{
      }
      &&
      \mathrm{Sp}(1)
        \boldsymbol{\cdot}
      \mathrm{Sp}(1)
        \boldsymbol{\cdot}
      \mathrm{Sp}(1)
      \\
      [q_1,\,q_2]
    \;\;\;  \ar@{|->}[rr]
      &&
      \big[ q_1 ,\, q_2,\, q_2 \big]
    }
  \end{equation}
\end{prop}
\begin{proof}
Consider the following diagram:
$$
  \hspace{-.2cm}
  \xymatrix@C=1.8em{
    S^7
    \ar@/^2pc/[rrrr]^{ h_{\mathbb{H}} }
    \ar[dd]_-{ \mathrm{quot} }
    \ar@{=}[r]
    &
    \frac{
      \mathrm{Sp}(2)
    }{
      \mathrm{Sp}(1)
    }
    \ar[dd]_-{ \mathrm{fib} }
    \ar[rr]_-{
      \frac{
        \mathrm{id}
      }{
        \mathrm{
          (\mathrm{id},\;e)
        }
      }
    }
    &&
    \frac{
      \mathrm{Sp}(2)
    }{
      \mathrm{Sp}(1)
        \times
      \mathrm{Sp}(1)
    }
    \ar[dd]^-{ \mathrm{fib} }
    \ar@{=}[r]
    &
    S^4
    \ar[dd]^-{ \mathrm{quot} }
    \\
    \\
    S^7
      \!\sslash\!
      \big(
        \mathrm{Sp}(2) \!\boldsymbol{\cdot}\! \mathrm{Sp}(1)
      \big)
      \ar@/_2pc/[rrrr]|-{\footnotesize
       \; h_{\mathbb{H}}
        \sslash
        (
          \mathrm{Sp}(2)\boldsymbol{\cdot} \mathrm{Sp}(1)
        ) \;
      }
    \ar@{=}[r]
    &
    B
    \big(
      \mathrm{Sp}(1)
       \boldsymbol{\cdot}
      \mathrm{Sp}(1)
    \big)
    \ar[ddr]|<<<<<<<{
      \phantom{ {AA} \atop {AA} }
    }_>>>>>>>>>>>>{
    }
    \ar[rr]^-{\small
      B
      (
        [q_1, q_2]
        \mapsto
        [q_1, q_2, q_2]
      )
    }
    &&
    B
    \big(
      \mathrm{Sp}(1)
      \!\boldsymbol{\cdot}\!
      \mathrm{Sp}(1)
      \!\boldsymbol{\cdot}\!
      \mathrm{Sp}(1)
    \big)
    \ar@{=}[r]
    \ar[ddl]|<<<<<<<<{
      \phantom{ {AA} \atop {AA} }
    }^>>>>>>>>>>>>{
    }
    &
    S^4
    \!\sslash\!
    \big(
      \mathrm{Sp}(2)\!\boldsymbol{\cdot}\!\mathrm{Sp}(1)
    \big)
    \\
    \\
    &
    &
    \!\!\!\!\!\!\!\!\!\!\!\!\!\!\!\!\!\!
    B
    \big(
      \mathrm{Sp}(2)
      \!\boldsymbol{\cdot}\!
      \mathrm{Sp}(1)
    \big)
    \!\!\!\!\!\!\!\!\!\!\!\!\!\!\!\!\!\!
  }
  $$
The outer rectangle exhibits the homotopy quotient of
$h_{\mathbb{H}}$ that we are after, and so we need to show
this factors as a pasting of homotopy commutative inner squares as
shown.

First, the factorization of the top
horizontal map follows as the right
half of diagram \eqref{CosethH} in Prop. \ref{EquihH}.
Moreover, the bottom triangle exhibits the
delooping of the factorization
\begin{equation}
  \xymatrix@R=-2pt@C=13pt{
    \mathrm{Sp}(1)
      \boldsymbol{\cdot}
    \mathrm{Sp}(1)
    \ar@{^{(}->}[rr]^-{
    }
    &&
    \mathrm{Sp}(1)
      \boldsymbol{\cdot}
    \mathrm{Sp}(1)
      \boldsymbol{\cdot}
    \mathrm{Sp}(1)
    \ar@{^{(}->}[rr]^-{
    }
    &&
    \mathrm{Sp}(2)
      \boldsymbol{\cdot}
    \mathrm{Sp}(1)
    \\
    [q_1,q_2]
    \ar@{|->}[rr]
    &&
    [q_1,q_2, q_2]
    \ar@{|->}[rr]
    &&
    \left[
      \left(
        {\begin{smallmatrix}
          q_1 & 0
          \\
          0 & q_2
       \end{smallmatrix}}
      \right)
      ,
      q_2
    \right]
  }
\end{equation}
and hence commutes by construction.
This implies, by functoriality of homotopy fibers,
that also the square of homotopy fibers commutes,
and hence the whole diagram commutes as soon as these
squares have top horizontal morphisms as shown.
Hence it remains to see that the induced morphism of homotopy fibers
is indeed as shown, and hence is indeed the quaternionic Hopf fibration.

For this, we invoke Lemma \ref{fibBiota}, which says that the
homotopy fibers here are the coset spaces of the corresponding
group inclusions, and hence the morphism of homotopy fibers
is the corresponding induced morphism of coset spaces.
With this we are reduced to showing that we have a commuting top square
as follows
\begin{equation}
  \label{ExtendedCosetSpaceRealizationOfQuaternionicHopfFibration}
  \raisebox{30pt}{
  \xymatrix@C=4em{
    \frac{
      \mathrm{Sp}(2)
        \boldsymbol{\cdot}
      \mathrm{Sp}(1)
    }{
      \mathrm{Sp}(1) \boldsymbol{\cdot} \mathrm{Sp}(1)
    }
    \ar[rr]^-{
      \frac{
        \mathrm{id}
      }{
        [q_1,q_2] \mapsto [q_1, q_2, q_2]
      }
    }
    \ar@{=}[d]
    &&
    \frac{
      \mathrm{Sp}(2)
        \cdot
      \mathrm{Sp}(1)
    }{
      \mathrm{Sp}(1)
        \boldsymbol{\cdot}
      \mathrm{Sp}(1)
        \boldsymbol{\cdot}
      \mathrm{Sp}(1)
    }
    \ar@{=}[d]
    \\
    \frac{
      \mathrm{Sp}(2)
    }{
      \mathrm{Sp}(1)
    }
    \ar@{=}[d]
    \ar[rr]^-{
      \frac{
        \mathrm{id}
      }{
        q \mapsto (q,1)
      }
    }
    &&
    \frac{
      \mathrm{Sp}(2)
    }{
      \mathrm{Sp}(1)
        \times
      \mathrm{Sp}(1)
    }
    \ar@{=}[d]
    \\
    S^7
    \ar[rr]^-{ h_{\mathbb{H}} }
    &&
    S^4
  }
  }
\end{equation}
because the bottom square already commutes by Prop. \ref{EquihH}.

For this, we observe that the groups
$\mathrm{Sp}(1)\boldsymbol{\cdot} \mathrm{Sp}(1)$
and
$\mathrm{Sp}(1)
  \boldsymbol{\cdot}
  \mathrm{Sp}(1)
  \boldsymbol{\cdot}
  \mathrm{Sp}(1)
$
are the
\emph{stabilizer subgroups} under the
respective $\mathrm{Sp}(2)\boldsymbol{\cdot}\mathrm{Sp}(1)$-actions
from Prop. \ref{EquihH}
on $S^7$ and $S^4$,
of any one point on $S^7$ and $S^4$, respectively:
For definiteness we consider the points
$$
\begin{bmatrix}
      0
      \\
      1
    \end{bmatrix}
  \;\in\;
    S^7
    \;\simeq\;
    S
    \left(\!\!\!
    {\begin{array}{c}
      \underset{
        \oplus
      }{
        \mathbb{H}
      }
      \\
      \mathbb{H}
    \end{array}}
  \!\!\!\right)
  \qquad
  \mbox{and}
  \qquad
  \begin{bmatrix}
        0
        \\
        1
      \end{bmatrix}
     \;\in\;
    S^4
    \;\simeq\;
    S
    \left(\!\!\!
    {\begin{array}{c}
      \underset{\oplus}{
        \mathbb{H}
      }
      \\
      \mathbb{R}
    \end{array}}
 \!\!\! \right)
$$
for which one sees by direct inspection of the matrix multiplications involved that their
stabilizer subgroups under the actions of Prop. \ref{EquihH} are as follows:
$$
  \xymatrix@C=34pt{
    \mathrm{Sp}(1)
    \boldsymbol{\cdot}
    \mathrm{Sp}(1)
    \ar[d]_-{\simeq}
    \ar[rr]
    &&
    \mathrm{Sp}(1)
    \boldsymbol{\cdot}
    \mathrm{Sp}(1)
    \boldsymbol{\cdot}
    \mathrm{Sp}(1)
    \ar[d]_-{\simeq}
    \\
  \left\{
  \left[
\left(
        {\begin{smallmatrix}
          q_1 & 0
          \\
          0 & q_2
       \end{smallmatrix}}
      \right)
  ,
  q_2
  \right]
  \;\vert\;
    q_i
    \;\in\;
    \mathrm{Sp}(1)
  \right\}
   \ar[rr]^{
     [q_1, q_2]
     \;\mapsto\;
     [ q_1, q_2, q_2 ]
   }
   \ar[d]_-{ \simeq }
   &&
  \Big\{
  \big[
    \mathrm{conj}(q_1,q_2)
    ,
    q_3
  \big]
  \;\vert\;
    q_i
    \;\in\;
    \mathrm{Sp}(1)
  \Big\}
  \ar[d]^-{ \simeq }
  \\
  \mathrm{Stab}_{
    \mathrm{Sp}(2)\boldsymbol{\cdot}\mathrm{Sp}(1)
  }
  \left(\!
    \left[\!\!\!\!
      {\begin{array}{c}
        0
        \\
        1
      \end{array}}
   \!\!\!\! \right]
      \in
      {\begin{array}{c}
        \underset{
          \oplus
        }{
          \mathbb{H}
        }
        \\
        \mathbb{H}
      \end{array}}
 \!\! \right)
    \ar@{^{(}->}[dr]
    &&
  \mathrm{Stab}_{
    \mathrm{Sp}(2)
      \boldsymbol{\cdot}
    \mathrm{Sp}(1)
  }
  \left(\!
    \left[\!\!\!\!
      {\begin{array}{c}
        0
        \\
        1
      \end{array}}
    \!\!\!\!\right]
    \in
    {\begin{array}{c}
      \underset{\oplus}{
        \mathbb{H}
      }
      \\
      \mathbb{R}
    \end{array}}
 \!\!\! \right)
   \ar@{^{(}->}[dl]
  \\
  &
  \mathrm{Sp}(2)
    \boldsymbol{\cdot}
  \mathrm{Sp}(1)
  }
$$
Here on the left we used the defining action by
quaternionic matrix multiplication from \eqref{ActionOfSp2Sp1},
while on the right we used the quaternionic conjugation
action $\mathrm{conj}(-,-)$ \eqref{Spin4ByConjugationAction}
of
$\mathrm{Spin}(4) \simeq \mathrm{Sp}(1) \times \mathrm{Sp}(1)$ by
Lemma \ref{Spin4Action}.

That our groups are thus stabilizer subgroups
implies the existence of top vertical isomorphisms in
\eqref{ExtendedCosetSpaceRealizationOfQuaternionicHopfFibration}.
Making these explicit
and chasing a coset through the
top square in \eqref{ExtendedCosetSpaceRealizationOfQuaternionicHopfFibration}
makes manifest that the square indeed commutes:
$$
  \raisebox{20pt}{
  \hspace{-.4cm}
  \footnotesize
  \xymatrix@C=4em@R=4em{
    \mathrm{Sp}(1)
    \boldsymbol{\cdot}
    \mathrm{Sp}(1)
    \ar@{^{(}->}[r]^-{
      {
      [q_1,q_2]
      \mapsto
      [\mathrm{diag}(q_1,q_2),q_2]
      }
      \atop
      \phantom{a}
    }
    &
    \mathrm{Sp}(2)
      \boldsymbol{\cdot}
    \mathrm{Sp}(1)
    \ar@{->>}[r]|-{ \mathrm{quot} }
    &
    \frac{
      \mathrm{Sp}(2)
        \boldsymbol{\cdot}
      \mathrm{Sp}(1)
    }{
      \mathrm{Sp}(1) \boldsymbol{\cdot} \mathrm{Sp}(1)
    }
    \ar[r]^-{
      {
      \frac{
        \mathrm{id}
      }{
        [q_1,q_2]
        \mapsto
        [q_1, q_2, q_2]
      }
      }
      \atop
      {
        \phantom{a}
      }
    }
    \ar@{<-}[d]_-{\simeq}
    &
    \frac{
      \mathrm{Sp}(2)
        \cdot
      \mathrm{Sp}(1)
    }{
      \mathrm{Sp}(1)
        \boldsymbol{\cdot}
      \mathrm{Sp}(1)
        \boldsymbol{\cdot}
      \mathrm{Sp}(1)
    }
    \ar@{<-}[d]^-{ \simeq }
    &
    \mathrm{Sp}(2)
      \boldsymbol{\cdot}
    \mathrm{Sp}(1)
    \ar@{->>}[l]|-{
      \mathrm{quot}
    }
    &
    \mathrm{Sp}(1)
      \boldsymbol{\cdot}
    \mathrm{Sp}(1)
      \boldsymbol{\cdot}
    \mathrm{Sp}(1)
    \ar@{_{(}->}[l]_-{
      {
       [
        \mathrm{diag}(q_1,q_2),
        q_3
       ]
       \mapsfrom
       [q_1, q_2, q_3]
       }
       \atop
       {
         \phantom{a}
       }
    }
    \\
    \mathrm{Sp}(1)
    \ar@{^{(}->}[r]_{
      q \mapsto \mathrm{diag}(q,1)
    }
    \ar@{^{(}->}[u]|-{
      \;\;q \mapsto [q,1]
    }
    &
    \mathrm{Sp}(2)
    \ar@{->>}[r]|-{ \mathrm{quot} }
    \ar@{^{(}->}[u]|-{
      \;\;A \mapsto [A,1]
    }
    &
    \frac{
      \mathrm{Sp}(2)
    }{
      \mathrm{Sp}(1)
    }
    \ar[r]_-{
      \frac{
        \mathrm{id}
      }{
        q \mapsto (q,1)
      }
    }
    &
    \frac{
      \mathrm{Sp}(2)
    }{
      \mathrm{Sp}(1)
        \times
      \mathrm{Sp}(1)
    }
    &
    \mathrm{Sp}(2)
    \ar@{^{(}->}[u]|-{
      \;\;
      A
      \mapsto
      [A,1]
    }
    \ar@{->>}[l]|-{
      \mathrm{quot}
    }
    &
    \mathrm{Sp}(1)
      \times
    \mathrm{Sp}(1)
    \ar@{^{(}->}[u]|-{
      (q_1, q_2)
      \mapsto
      [q_1,q_2, 1]
    }
    \ar@{_{(}->}[l]^-{
      {
        \phantom{a}
      }
      \atop
      {
      \mathrm{diag}(q_1,q_2)
      \mapsfrom
      (q_1, q_2)
      }
    }
    }
    }
  $$

$$
\footnotesize
  \xymatrix{
    [A,1]
    \cdot
    \big(
      \mathrm{Sp}(1)
      \boldsymbol{\cdot}
      \mathrm{Sp}(1)
    \big)
    \ar@{|->}[r]
    &
    [A,1]
    \cdot
    \big(
      \mathrm{Sp}(1)
      \boldsymbol{\cdot}
      \mathrm{Sp}(1)
      \boldsymbol{\cdot}
      \mathrm{Sp}(1)
    \big)
    \\
    A
    \cdot
    \big(
      \mathrm{Sp}(1)
    \big)
    \ar@{|->}[r]
    \ar@{|->}[u]
    &
    A
    \cdot
    \big(
      \mathrm{Sp}(1)
        \boldsymbol{\cdot}
      \mathrm{Sp}(1)
    \big)
    \ar@{|->}[u]
  }
 $$
This completes the proof.

\vspace{-4mm}
\end{proof}

\medskip

\medskip

\subsection{Twisted Cohomotopy in degree 7 alone}
\label{TwistedCohomotopyInDegreeSeven}

If we do not require the twists of Cohomotopy in degree 7
to be compatible with the quaternionic Hopf fibration
(as we did in the previous section, \cref{TwistedCohomotopyInDegrees})
then there are more exceptional twists.
We give a homotopy-theoretic classification of these in
Prop. \ref{Cohomotopy7GStructure} below.
In Prop. \ref{NIs1GStructures} below
we highlight how this recovers precisely
the special holonomy structures of $\mathcal{N}=1$ compactifications
of M/F-theory.

\medskip
Further below in  \cref{TwistedPT}, we explain how
these $N=1$ structures are \emph{fluxless} in
a precise cohomotopical sense, which crucially enters the M2-tadpole
cancellation in  \cref{M2BraneTadpoleCancellation}.

\medskip

\begin{prop}[G-structures induced by Cohomotopy in degree 7]
\label{Cohomotopy7GStructure}
We have the following sequence of homotopy
pullbacks of universal 7-spherical fibrations, hence of twists
for Cohomotopy in degree 7 (see \hyperlink{FigureD}{\it Figure D}):
$$
  \xymatrix@R=.7em@C=4em{
    S^7
    \ar@{=}[dd]
    \ar[rr]^-{ \mathrm{fib} }
    &&
    B \mathrm{SU}(2)
    \ar[r]
    \ar[dd]
    \ar@{}[ddr]|-{
      \mbox{
        \tiny
        {\rm (pb)}
      }
    }
    &
    B \mathrm{Spin}(5)
    \ar[dd]
    \\
    \\
    S^7
    \ar@{=}[dd]
    \ar[rr]^-{ \mathrm{fib} }
    &&
    B \mathrm{SU}(3)
    \ar[r]
    \ar[dd]
    \ar@{}[ddr]|-{
      \mbox{
        \tiny
        {\rm (pb)}
      }
    }
    &
    B \mathrm{Spin}(6)
    \ar[dd]
    \\
    \\
    S^7
    \ar@{=}[dd]
    \ar[rr]^-{ \mathrm{fib} }
    &&
    B \mathrm{G}_2
    \ar[r]
    \ar[dd]
    \ar@{}[ddr]|-{
      \mbox{
        \tiny
      {\rm  (pb)}
      }
    }
    &
    B \mathrm{Spin}(7)
    \ar[dd]^-{ B \iota }
    \\
    \\
    S^7
    \ar[rr]^-{ \mathrm{fib} }
    &&
    B \mathrm{Spin}(7)
    \ar[r]_-{B \iota' }
    &
    B \mathrm{Spin}(8)
  }
$$
\end{prop}
\begin{proof}
First, observe that there is the following analogous commuting diagram of Lie groups:
\begin{equation}
  \label{SomeSpin8Subgroups}
  \xymatrix@C=25pt@R=9pt{
    \mathrm{SU}(2)
  \;  \ar@{^{(}->}[rr]
    \ar@{^{(}->}[dd]
    \ar@{}[ddrr]|-{
      \mbox{
        \tiny
        (pb)
      }
    }
    &&
    \mathrm{SU}(3)
   \; \ar@{^{(}->}[rr]
     \ar@{^{(}->}[dd]
    \ar@{}[ddrr]|-{
      \mbox{
        \tiny
        (pb)
      }
    }
    &&
    \mathrm{G}_2
  \;  \ar@{^{(}->}[rr]
    \ar@{^{(}->}[dd]
    \ar@{}[ddrr]|-{
      \mbox{
        \tiny
        (pb)
      }
    }
    &&
    \mathrm{Spin}(7)
    \ar@{^{(}->}[dd]^-{ \iota' }
    \\
    \\
    \mathrm{Spin}(5)
    \ar@{->>}[dd]
 \;   \ar@{^{(}->}[rr]
    \ar@{}[ddrr]|-{
      \mbox{
        \tiny
        (pb)
      }
    }
    &&
    \mathrm{Spin}(6)
  \;  \ar@{^{(}->}[rr]
    \ar@{->>}[dd]
    \ar@{}[ddrr]|-{
      \mbox{
        \tiny
        (pb)
      }
    }
    &&
    \mathrm{Spin}(7)
 \;   \ar@{^{(}->}[rr]
    \ar@{->>}[dd]
    \ar@{}[ddrr]|-{
      \mbox{
        \tiny
        (pb)
      }
    }
    &&
    \mathrm{Spin}(8)
    \ar@{->>}[dd]
    \\
    \\
    \mathrm{SO}(5)
   \; \ar@{^{(}->}[rr]
    &&
    \mathrm{SO}(6)
  \;  \ar@{^{(}->}[rr]
    &&
    \mathrm{SO}(7)
  \;  \ar@{^{(}->}[rr]
    &&
    \mathrm{SO}(8) \;.
  }
 \end{equation}
Here the bottom squares evidently commute and are pullback
squares by the definition of Spin groups, while the
three total vertical rectangles commute and are pullback
squares by  \cite[Table 2, p. 144]{Onishchik93}.
By the pasting law,
\footnote{\label{foot}
Recall that this says that if
$$
  \xymatrix@R=1em@C=3em{
    A
    \ar[r]
    \ar[d]
    & B
    \ar[r]
    \ar[d]
    \ar@{}[dr]|-{
      \mbox{
        \tiny
         (pb)
      }
    }
    &
    C
    \ar[d]
    \\
    D
    \ar[r]
    &
    E
    \ar[r]
    &
    F
  }
$$
is a commuting diagram, where the right square is a pullback, then the left square is a pullback
precisely if the full outer rectangle is a pullback. The same holds for homotopy-commutative
diagrams and homotopy-pullback squares.}
this implies that also the top squares are pullbacks, hence exhibiting intersections
of subgroup inclusions.
Notice that the top right vertical inclusion $\iota^\prime$ is \emph{not} the canonical inclusion
of $\mathrm{Spin}(7)$ in $\mathrm{Spin}(8)$, but is a subgroup inclusion in a distinct
$\mathrm{Spin}(7)$-conjugacy class, of which there are three \cite[Thm. 5 on p. 6]{Varadarajan01}.
The intersection in the top right square is also proven in \cite[Thm. 5 on p. 13]{Varadarajan01}, and
that of the middle square in \cite[Lem. 9 on p. 10]{Varadarajan01}.
Again, by the pasting law, this implies that also the top squares are pullbacks,
hence exhibiting intersections of subgroup inclusions.

Applying delooping (passage to classifying spaces) to these top squares,
this shows that we have a homotopy commuting diagram as follows:
\begin{equation}
  \xymatrix@C=13pt@R=4pt{
    &
    &&
    S^7
    \ar[dddd]^-{\mathrm{fib}}
    \ar@{=}[rr]
    &&
    S^7
    \ar@{=}[rr]
    \ar[dddd]^-{\mathrm{fib}}
    &&
    S^7
    \ar@{=}[rr]
    \ar[dddd]^-{\mathrm{fib}}
    &&
    S^7
    \ar[dddd]^-{\mathrm{fib}}
    \\
    \\
    \\
    S^3
    \ar[dr]^-{\mathrm{fib}}
    &
    &
    S^5
    \ar[dr]^-{\mathrm{fib}}
    & &
    S^6
    \ar[dr]^-{\mathrm{fib}}
    & &
    S^7
    \ar[dr]^-{\mathrm{fib}}
    \\
    &
    \ast
    \ar[rr]
    \ar[dddd]
    &&
    B \mathrm{SU}(2)
    \ar[rr]
    \ar[dddd]
    \ar@{}[ddddrr]|-{
      \mbox{
        \tiny
        (pb)
      }
    }
    &&
    B \mathrm{SU}(3)
    \ar[rr]
    \ar[dddd]
    \ar@{}[ddddrr]|-{
      \mbox{
        \tiny
        (pb)
      }
    }
    &&
    B \mathrm{G}_2
    \ar[rr]
    \ar[dddd]
    \ar@{}[ddddrr]|-{
      \mbox{
        \tiny
        (pb)
      }
    }
    &&
    B \mathrm{Spin}(7)
    \ar[dddd]^-{ B \iota^\prime }
    \\
    \\
    \\
    \\
    &
    B\mathrm{Spin}(4)
    \ar[rr]
    &&
    B \mathrm{Spin}(5)
    \ar[rr]
    &&
    B \mathrm{Spin}(6)
    \ar[rr]
    &&
    B \mathrm{Spin}(7)
    \ar[rr]
    &&
    B \mathrm{Spin}(8)
    \\
    S^4
    \ar[ur]_-{\mathrm{fib}}
    &
    &
    S^5
    \ar[ur]_-{\mathrm{fib}}
    & &
    S^6
    \ar[ur]_-{\mathrm{fib}}
    & &
    S^7
    \ar[ur]_-{\mathrm{fib}}
  }
\end{equation}
The spherical homotopy fibers shown in this diagram follow
by using Lemma \ref{fibBiota}
with classical results
about coset space structures of topological spheres,
as summarized in \hyperlink{TableS}{Table S}.

In order to see that each square in the diagram of classifying spaces is a homotopy pullback,
we now use the following basic fact from homotopy theory (see e.g. \cite[5.2]{CPS05}):
Assume that $Y_1, Y_2$ are connected spaces, and we are given a homotopy-commutative
square as on the right in the following diagram
$$
  \xymatrix@C=15pt@R=10pt{
    \mathrm{fib}(f_1)
    \ar[rr]
    \ar@{..>}[dd]_-{\simeq}
    &&
    X_1
    \ar[rr]^-{f_1}
    \ar[dd]
    \ar@{}[ddrr]|-{\
      \mbox{
        \tiny
        (pb)
      }
    }
    &&
    Y_1
    \ar[dd]
    \\
    \\
    \mathrm{fib}(f_2)
    \ar[rr]
    &&
    X_2
    \ar[rr]^-{f_2}
    &&
    Y_2\;.
  }
$$
Then the square is a homotopy pullback square if and only if the induced left vertical morphism
between horizontal homotopy fibers is a weak homotopy equivalence; as indicated.
To see that in our case these induced left vertical morphisms are indeed weak homotopy equivalences,
we first observe that for each of the squares above the horizontal homotopy fibers are $n$-spheres
of the same dimension $n$:
$$
  \xymatrix@C=15pt@R=11pt{
    \mathllap{
      S^7 \simeq \,
    }
    \frac{
      \mathrm{Spin}(7)
    }{
      \mathrm{G}_2
    }
    \ar[rr]
    \ar@{-->}[dd]_-{\simeq}
    &&
    B \mathrm{G}_2
    \ar[rr]^-{}
    \ar[dd]
    \ar@{}[ddrr]|-{\
      \mbox{
        \tiny
      }
    }
    &&
    B \mathrm{Spin}(7)
    \ar[dd]
    \\
    \\
    \mathllap{
      S^7 \simeq \,
    }
    \frac{
      \mathrm{Spin}(8)
    }{
      \mathrm{Spin}(7)
    }
    \ar[rr]
    &&
    B \mathrm{Spin}(7)
    \ar[rr]_-{}
    &&
    B \mathrm{Spin}(8)
  }
$$
and
$$
  \xymatrix@C=15pt@R=11pt{
    \mathllap{
      S^6 \simeq \,
    }
    \frac{
      \mathrm{G}_2
    }{
      \mathrm{SU}(3)
    }
    \ar[rr]
    \ar@{-->}[dd]_-{\simeq}
    &&
    B \mathrm{SU}(3)
    \ar[rr]^-{}
    \ar[dd]
    \ar@{}[ddrr]|-{\
      \mbox{
        \tiny
      }
    }
    &&
    B \mathrm{G}_2
    \ar[dd]
    \\
    \\
    \mathllap{
      S^6 \simeq \,
    }
    \frac{
      \mathrm{Spin}(7)
    }{
      \mathrm{Spin}(6)
    }
    \ar[rr]
    &&
    B \mathrm{Spin}(6)
    \ar[rr]_-{}
    &&
    B \mathrm{Spin}(7)
  }
$$
(for the coset realization of $S^6$ on the top left see \cite{FuIsh55})
and
$$
  \xymatrix@C=15pt@R=11pt{
    \mathllap{
      S^5 \simeq \,
    }
    \frac{
      \mathrm{SU}(3)
    }{
      \mathrm{SU}(2)
    }
    \ar[rr]
    \ar@{-->}[dd]_-{\simeq}
    &&
    B \mathrm{SU}(2)
    \ar[rr]^-{}
    \ar[dd]
    \ar@{}[ddrr]|-{\
      \mbox{
        \tiny
      }
    }
    &&
    B \mathrm{SU}(3)
    \ar[dd]
    \\
    \\
    \mathllap{
      S^5 \simeq \,
    }
    \frac{
      \mathrm{Spin}(6)
    }{
      \mathrm{Spin}(5)
    }
    \ar[rr]
    &&
    B \mathrm{Spin}(5)
    \ar[rr]_-{}
    &&
    B \mathrm{Spin}(6)\;.
  }
$$
To see in detail that the homotopy fibers on the left are not only pairwise weakly homotopy equivalent,
but that the universally induced dashed morphism exhibits such a weak homotopy equivalence, we proceed
as follows. For $G := \mathrm{Spin}(n)$ one of the  $\mathrm{Spin}$ groups appearing above,
pick any one topological space $E G$ modelling the total space of the universal $G$ bundle
(hence any weakly contractible topological space equipped with a free continuous $G$-action).
Then for $G' \overset{\iota}{\hookrightarrow} G$ any subgroup, we have that the projection
 $(E G)/G' \to (E G)/G$ is a Serre fibration modelling
$B G' \overset{B \iota}{\longrightarrow} B G$ (e.g. \cite[11.4]{Mitchell11}).
Since ordinary pullbacks of Serre fibrations are already homotopy pullbacks,
this means that the above homotopy pullback squares are represented by actual
pullback squares of topological spaces in the following diagram:
$$
  \xymatrix@C=25pt@R=11pt{
    \mathllap{
      S^n \simeq \,
    }
    \frac{
      G'
    }{
      G' \cap G''
    }
    \ar[rr]
    \ar@{-->}[dd]_-{\simeq}
    \ar@{}[ddrr]|-{\
      \mbox{
        \tiny
        (pb)
      }
    }
    &&
    (
      E G
    )
    /
    (
      G' \cap G''
    )
    \ar[rr]^-{}
    \ar[dd]
    \ar@{}[ddrr]|-{\
      \mbox{
        \tiny
        (pb)
      }
    }
    &&
    (
      E G
    )/ G'
    \ar[dd]
    \\
    \\
    \mathllap{
      S^n \simeq \,
    }
    \frac{
      G
    }{
      G''
    }
    \ar[rr]
    &&
    (
     E G
    )/ G''
    \ar[rr]_-{}
    &&
    (
     E G
    )/ G\;.
  }
$$
Here the dashed morphism is the canonical continuous
function induced by the given group inclusions,
so that it is now sufficient to observe that this is a homeomorphism.

While this does not follow for general subgroup intersections,
it does follow as soon as the given coset spaces
are homeomorphic, as is the case here.
Namely, pick any point $x \in S^n$ and observe
that we have a commuting square of continuous functions as follows.
$$
  \xymatrix@C=4em@R=1.5em{
    S^n
    \ar@{<-}[rr]^-{ [g'] \mapsto g'(x) }_-{\simeq_{\mathrm{homeo}}}
    \ar@{=}[d]
    &&
    \tfrac{
      G'
    }{
      G' \cap G''
    }
    \ar[d]
    \\
    S^n
    \ar@{<-}[rr]_-{ [g] \mapsto g(x) }^-{\simeq_{\mathrm{homeo}}}
    && \;
    \tfrac{
      G
    }{
      G''
    }\;.
  }
$$
Since in this diagram the top, bottom and left maps are homeomorphisms,
it follows that the right map is also a homeomorphism.
\end{proof}

\begin{remark}
[Twisted generalized cohomotopy]
  \label{TwistedGeneralizedCohomotopy}
One may also consider twisted Cohomotopy with coefficients
in fibrations of \emph{pairs} of spheres:
$$
  \left\{
    \raisebox{23pt}{
    \xymatrix{
    &&
    \big(
      S^p \times S^q
    \big)
    \!\sslash\!
    \big(
     {\rm O}(p) \times {\rm O}(q)
    \big)
    \ar[d]
    \\
    X
      \ar@{-->}[urr]
      \ar[rr]_{ TX }
      &&
    B \mathrm{O}(n)
    }
    }
  \right\}_{\raisemath{16pt}{\Bigg/\sim}}
$$
\item {\bf (i)} Corresponding twists arise from ``doubly exceptional geometry'',
in that we have the following pasting diagram of homotopy
pullbacks, further refining those of Prop. \ref{Cohomotopy7GStructure}:
$$
  \raisebox{49pt}{
  \xymatrix@R=1em{
    S^7\times S^7
    \ar@{}[ddrr]|-{
      \mbox{
        \tiny
        (pb)
      }
    }
    \ar[rr]
    \ar[dd]
    &&
    B G_2
    \ar[dd]^-{ i_{G_2} }
    \\
    \\
    S^7\ar[dd]
    \ar[rr]
    \ar@{}[ddrr]|-{
      \mbox{
        \tiny
        (pb)
      }
    }
    &&
    B \mathrm{Spin}(7)\ar[dd]^-{ i_{\mathrm{Spin(7)}} }
    \\
    \\
    \ast
    \ar[rr]
    &&
    B \mathrm{Spin}(8)
  }
  }
 \phantom{AAAA}
 \mbox{equivalently}
 \phantom{AAAA}
 \raisebox{49pt}{
  \xymatrix@R=1em{
    S^7\times S^7
    \ar[rr]
    \ar[dd]
    \ar@{}[ddrr]|-{
      \mbox{
        \tiny
        (pb)
      }
    }
    &&
    S^7 \!\sslash\! \mathrm{Spin}(7)
    \ar[dd]^-{ i_{G_2} }
    \\
    \\
    S^7\ar[dd]
    \ar[rr]
    \ar@{}[ddrr]|-{
      \mbox{
        \tiny
        (pb)
      }
    }
    &&
    S^7 \!\sslash\! \mathrm{Spin}(8)
    \ar[dd]^-{ i_{\mathrm{Spin(7)}} }
    \\
    \\
    \ast \ar[rr]
    &&
    B \mathrm{Spin}(8)
  }
  }
$$
This follows analogously as in Prop. \ref{Cohomotopy7GStructure},
with \cite[p. 146]{Onishchik93}.

\item {\bf (ii)} Further twists for Cohomotopy with coefficients in $S^p \times S^p$
arise from topological $G$-structure for rotation groups ${\rm O}(p,p)$ in split signature,
and hence from \emph{generalized geometry} (e.g. \cite{Hull07}).
This is because indefinite orthogonal groups are homotopy equivalent to their
maximal compact subgroups via the polar decomposition
$$
  \mathrm{O}(p,p)
  \;\simeq_{\mathrm{wh}}\;
  \mathrm{O}(p) \times \mathrm{O}(p)
$$
(see, e.g., \cite[Sec. 17.2]{HN}) and similarly for higher connected covers (see \cite{SSh}).
Therefore, we might call Cohomotopy with coefficients in
$S^p \times S^p$, and twisted by generalized geometry,
\emph{generalized Cohomotopy} (not to be confused with older terminology \cite{Jaworowski62}).
We will discuss the details elsewhere.
\end{remark}


\medskip

\subsection{Twisted Cohomotopy via Poincar{\'e}-Hopf}
\label{EulerCharacteristicAndM2BraneWorldvolumes}

We characterize here the $T X$-twisted Cohomotopy of
compact orientable smooth manifolds $X$ in terms of the
``Cohomotopy charge'' carried by a finite number of
point singularities in $X$.
This is the content of
Prop. \ref{TwistedCohomotopyViaPoincareHopf} below.
The proof is a cohomotopical restatement of the
classical Poincar{\'e}-Hopf (PH) theorem (see e.g. \cite[Sec. 15.2]{DNF85}),
but the perspective of twisted Cohomotopy
is noteworthy in itself and is crucial for the discussion
of M2-brane tadpole cancellation in \cref{M2BraneTadpoleCancellation}
below.

\medskip

\begin{prop}[Twisted cohomotopy and the Euler characteristic]
  \label{TwistedCohomotopyViaPoincareHopf}
  Let $X$ be an orientable compact smooth manifold.
  Then:
\vspace{-2mm}
\item  {\bf (i)} A cocycle in the $T X$-twisted Cohomotopy of $X$ (Def. \ref{TwistedCohomotopy})
exists if and only if the Euler  characteristic of $X$ vanishes:
  $$
    \pi^{T X}( X )
    \;\neq\;
    \varnothing
    \;\;\;\;
    \Longleftrightarrow
    \;\;\;\;
    \rchi[X]
    \;\neq\;
    0
    \,.
  $$

 \item {\bf (ii)} Generally, there exists a finite set of
  points $\{x_i \in X\}$ such that
  the operation of
  restriction to open neighbourhoods of these
  points exhibits an injection of the $T X$-twisted Cohomotopy
  of their complement
  $
      \pi^{T X}
      \big(
        X \setminus \underset{i}{\coprod}\{x_i\}
      \big)
  $
  (Def. \ref{TwistedCohomotopy})
  into the
  product of untwisted Cohomotopy sets \eqref{CohomotopyUntwisted}
  $
    \pi^{ \mathrm{dim}(X) }\big( U_{x_i} \setminus \{x_i\}\big)
  $
  of these
  pointed neighborhoods.
  Moreover, the latter are integers
  which sum to the Euler characteristic $\rchi[X]$
  of $X$:
  \begin{equation}
    \label{PHViaCohomotopy}
    \xymatrix{
      \pi^{T X}
      \big(
        X \setminus \underset{i}{\coprod}\{x_i\}
      \big)
      \ar[d]
      \; \ar@{^{(}->}[rr]^-{ \mathrm{restr.} }
      &&
      \underset{i}{\prod}
      \pi^{ \mathrm{dim}(X)-1 }
      \big(
        U_{x_i} \setminus \{x_i\}
      \big)
      \ar[r]^-{ \simeq }
      &
      \underset{i}{\prod}
      \mathbb{Z}
      \ar[d]^-{ \underset{i}{\sum} }
      \\
      \ast
      \ar[rrr]^-{ \chi[X] }
      && &
      \mathbb{Z}
    }
  \end{equation}
\end{prop}
\begin{proof}
This follows with the classical {Poincar{\'e}-Hopf theorem}, \eqref{PHTheorem} below.
We recall the relevant terminology:

 \begin{enumerate}[{\bf (i)}]
 \vspace{-2mm}
  \item For $v$ a vector field on $X$, a point $x \in X$ is
  called an
  \emph{isolated zero} of $v$ if there exists an open
  contractible neighborhood
   $U_x \subset X$ such that
  the restriction $v|_{U_x}$ of $v$ to this neighborhood vanishes at $x$ and only at $x$.
  \vspace{-2mm}
  \item This means that on $U_x \setminus \{x\}$ the vector field $v$
   induces a map to the $( \mathrm{dim}(X)-1)$-sphere
   \begin{equation}
     \label{vOnUxAsMapToSphere}
     v/{\left\vert v\right\vert}
     \;:\;
     \xymatrix@C=4em{
      U_x \setminus \{x\}
      \ar[r]^-{ v/{\vert v\vert }}
      &
      S(T_x X)
      \simeq
      S^{ \mathrm{dim}(X)-1 }
      }
      \,.
   \end{equation}
   Here the equivalence on the right is to highlight that
   the sphere arises as the fiber of the unit sphere bundle
   of the tangent bundle $T U_x$, which may be
   identified with the unit sphere in $T_x X$, by the
   assumed contractibility of $U_x$.
\vspace{-2mm}
  \item Given an isolated zero $x$, the
   \emph{Poincar{\'e}-Hopf index} of $v$ at that point is
   the degree of the associated map \eqref{vOnUxAsMapToSphere}
   to the sphere, for any choice of local chart:
  \begin{equation}
    \label{PHIndexAtIsolatedZero}
    \xymatrix@C=3em{
    \mathrm{index}_x(v)
    \;:=\;
    \mathrm{deg}
    \big(
      U_x \setminus \{x\}
      \ar[r]^-{ v/{\vert v\vert }}
      &
      S(T_x X)
      \simeq
      S^{ \mathrm{dim}(X)-1 }
    \big).
    }
  \end{equation}
\end{enumerate}

 Now for $X$ orientable and compact,
 the {\it Poincar{\'e}-Hopf theorem} (e.g. \cite[Sec. 15.2]{DNF85})
 says that for any vector field $v \in \Gamma(T X)$
  with a finite set $\{x_i \in X\}$ of isolated zeros,
  the sum of the indices \eqref{PHIndexAtIsolatedZero} of $v$
  equals the
  Euler characteristic $\rchi[X]$ of $X$:
  \begin{equation}
    \label{PHTheorem}
    \underset{
      \mathclap{
        \mbox{
          \tiny
          isolated zero
        }
        \atop
        x_i \in X
      }
    }{\sum}
    \;
    \mathrm{index}_{x_i}(v)
    \;=\;
    \rchi[X]
    \,.
  \end{equation}

To conclude, observe that the  maps to spheres in
\eqref{vOnUxAsMapToSphere} are but the restriction of
the corresponding cocycle in the $T X$-Cohomotopy of
$X \setminus \underset{i}{\coprod} \{x_i\}$:
$$
  \xymatrix{
    &&
    S^{\mathrm{dim}(X)} \!\sslash\! \mathrm{SO}(\mathrm{dim}(X))
    \ar[d]
    \\
    X \setminus \underset{i}{\coprod} \{x_i\}
    \ar@/^1pc/@{-->}[urr]^-{ v/{\left\vert v \right\vert} }
    \ar[rr]_-{ T X }
    &&
    B \mathrm{SO}(\mathrm{dim}(X))
  }
$$
Finally, the identification of the PH-index
with an integer is via the Hopf degree theorem \eqref{HopfDegreeTheorem},
now understood
as the characterization of untwisted Cohomotopy in \eqref{HopfDegreeTheorem}.
\end{proof}

  We may equivalently use the differential form  data that underlies a cocycle in twisted Cohomotopy,
  by Prop. \ref{SectionsOfRationalSphericalFibrations}, to   re-express the cohomotopical PH-theorem,
  Prop. \ref{TwistedCohomotopyViaPoincareHopf},  via Stokes' theorem.
  Let $X$ be an orientable compact smooth manifold  of even dimension $\mathrm{dim}(X) = 2n + 2$,
  for $n \in \mathbb{N}$   and let $v \in TX$ be a vector field   with isolated zeros $\{x_i \in X\}$.
  For any fixed choice of Riemannian metric   on $X$ and any small enough positive real number $\epsilon$,
  write
  $$
    D^\epsilon_{x_i}
    \;:=\;
    \big\{
      x \in X
      \;\vert\;
      d(x,x_i) \lt \epsilon
    \big\} \subset X
  $$
  for the open ball of radius $\epsilon$ around $x_i$.
  The complement of these open balls is hence a manifold
  with boundary a disjoint union of $(2n+1)$-spheres:
  $$
    \partial
    \big(
      X \setminus \underset{i}{\coprod} \{x_i\}
    \big)
    \;\simeq\;
    \underset{i}{\coprod}
    S^{2n+1}.
  $$
  Then, by Prop. \ref{SectionsOfRationalSphericalFibrations},   the cocycle in twisted Cohomotopy
  on $X \setminus \underset{i}{\coprod} \{x_i\}$ which  corresponds to the vector field $v$ has
  underlying it a differential $(2n+1)$-form $G_{2n+1}$ which satisfies
  $$
    d G_{2n+1}
    \;=\;
    -
    \rchi_{2n+2}(\nabla)
    \,.
  $$
  By Stokes' theorem we thus have
  $$
    \begin{aligned}
    \rchi[X]
    & =
    \underset{\epsilon \to 0}{\mathrm{lim}}
    \,
    \int_{X \setminus \underset{i}{\coprod} D^\epsilon_{x_i}}
    \rchi
    \\
    & =
    -
    \underset{\epsilon \to 0}{\mathrm{lim}}\,
    \underset{i}{\sum} \int_{ \partial D^\epsilon_{x_i} }
    G_{2n+1}
  \end{aligned}
  $$
  We may summarize the above by the following.

  \begin{lemma}[Cohomological PH-theorem]
  \label{CohomologicalPH}
  In the above setting, the Euler characteristic is given by
   the integral of $-G_{2n+1}$
  over the boundary components
  around the zeros of $v$:
  \begin{equation}
    \label{CohomologicalPHViaStokes}
    -
    \underset{i}{\sum}
    \int_{ S^{2n+1}_i }
    G_{2n+1}
    \;=\;
    \rchi[X]
    \,.
  \end{equation}
\end{lemma}


\medskip

\subsection{Twisted Cohomotopy via Pontrjagin-Thom}
\label{TwistedPT}

We recall the unstable Pontrjagin-Thom theorem
relating untwisted Cohomotopy to normally framed submanifolds,
\eqref{ClassicalPTIsIso} below.
Then we show that twisted Cohomotopy jointly in degrees
4 and 7 (as per \cref{TwistedCohomotopyInDegrees})
knows about \emph{calibrated submanifolds} in 8-manifolds, Prop. \ref{ModuliSpaceOfCalibratedFromTwistedCohomotopy} below.
Finally we observe that in this case vanishing submanifolds under a
twisted Pontrjagin-Thom construction means, equivalently, a
factorization through the quaternionic Hopf fibration,
\eqref{PTVanishingChargeAsLift} below.

\medskip

\noindent {\bf Framed submanifolds from untwisted Cohomotopy.}
One striking aspect of \hyperlink{HypothesisH}{\it Hypothesis H},
is that unstable Cohomotopy of a manifold $X$
is exactly the cohomology theory
which classifies (cobordism classes of)
\emph{submanifolds} $\Sigma \subset X$, subject to
constraints on the normal bundle $N_X \Sigma$ of the embedding.

\medskip
In the case of vanishing twist, this is the statement of the classical
\emph{unstable Pontrjagin-Thom isomorphism} (e.g. {\cite[IX.5]{Kosinski93}})
\begin{equation}
  \label{ClassicalPTIsIso}
  \xymatrix{
    \pi^n(X)
    \; \; \;
    \ar@{<-}@<+6pt>[rr]^-{ \mathrm{PT}^n }
    \ar@<-6pt>[rr]_-{
      \mathrm{fib}_0 \; \circ \; \mathrm{reg}
    }^-{
      \raisebox{1pt}{
        $\simeq$
      }
    }
    &&
    \;\; \;
    \mathrm{FrSubMfd}^{\mathrm{codim} = n}(X)_{\raisemath{0pt}{\big/ \sim_{\mathrm{bord}}}}.
  }
\end{equation}
For a closed smooth manifold $X$ and any degree $n \in \mathbb{N}$, this
identifies degree $n$ cocycles
$$
  \big[
    X
      \overset{c}{\longrightarrow}
    S^n
  \big]
  \;\in\;
  \pi^{n}(X)
$$
in the untwisted unstable Cohomotopy \eqref{CohomotopyUntwisted} of $X$
with the cobordism classes of normally framed submanifolds
$\Sigma$  of codimension $n$
$$
  \big(
    \Sigma \hookrightarrow X
    \,,
    \;\;
    \xymatrix{
      N_{X} \Sigma \ar[r]^-{\mathrm{fr}}_{ \simeq }
      &
     \Sigma \times \mathbb{R}^n
    }\,,
    \;\;
    \mathrm{dim}(\Sigma) = \mathrm{dim}(X) - n
     \; \big)
$$
given as the preimage of a chosen base point
\begin{equation}
  \label{BasePointInSphere}
  \mathrm{pt}
  \;\in\; S^n
\end{equation}
under a smooth function representative $c$ of $[c]$ for which $\mathrm{pt}$ is a regular value
$
  c^{-1}\big(
    \{\mathrm{pt}\}
  \big)
  \;=:\;
  \Sigma \subset X
  $.

\medskip
As advocated in \cite{S-top}, we may naturally think of the submanifolds
$\Sigma \subset X$  appearing in the unstable Pontrjagin-Thom isomorphism
\eqref{ClassicalPTIsIso} as \emph{branes} whose charge is given by the Cohomotopy class $[c]$.
This reveals Cohomotopy as the canonical cohomology theory for measuring charges of branes
given as (cobordism classes of) submanifolds. To see this in full detail one needs
to consider the refinement of \eqref{ClassicalPTIsIso} to  twisted and \emph{equivariant}
Cohomotopy. In the rational approximation this is discussed in \cite{ADE},
the full non-rational theory of M-branes at singularities classified by
equivariant Cohomotopy will be discussed
elsewhere \cite{SS19a}\cite{SS19b}.

\medskip
Here we content ourselves with highlighting two
related
facts, which are needed for the discussion in \cref{CancellationFromCohomotopy}.

\medskip

\noindent {\bf Calibrated submanifolds from twisted Cohomotopy.}
The manifold $\mathbb{R}^8$ carries an
exceptional \emph{calibration} by the \emph{Cayley 4-form}
$\Phi \in \Omega^4(\mathbb{R}^8)$ \cite{HarveyLawson82},
which singles out 4-dimensional
submanifold embeddings $\Sigma_4 \hookrightarrow \mathbb{R}^8$
as the corresponding \emph{calibrated submanifolds}.
The space of all such \emph{Cayley 4-planes},
canonically a subspace of the Grassmannian space
$\mathrm{Gr}(4,8)$ of \emph{all} 4-planes in 8 dimensions,
is denoted
\begin{equation}
  \mathrm{CAY}
  \;\subset\;
  \mathrm{Gr}(4,8)
\end{equation}
in \cite[(2.19)]{BryantHarvey89}\cite[(5.20)]{GluckMackenzieMorgan95}.
We will write
\begin{equation}
  \mathrm{CAY}_{\mathrm{sL}}
  \;\subset\;
  \mathrm{CAY}
  \;\subset\;
  \mathrm{Gr}(4,8)
\end{equation}
for the further subspace of those Cayley 4-planes which are
also special Lagrangian submanifolds.
There are canonical symmetry actions of $\mathrm{Spin}(7)$
and of $\mathrm{Spin}(6)$, respectively, on these spaces
\cite[Prop. 1.36]{HarveyLawson82}:
\begin{equation}
  \label{GroupActionsOnCayleyPlanes}
  \xymatrix{
    \mathrm{CAY}
    \ar@(ul,ur)^{ \mathrm{Spin}(7) }
  }
\quad
\text{and}
\quad
  \xymatrix{
    \mathrm{CAY}_{\mathrm{sL}}
    \ar@(ul,ur)^{ \mathrm{Spin}(6) }
  }.
\end{equation}
Hence the corresponding homotopy quotients
\begin{equation}
  \label{ModuliSpacesOfCayley4Planes}
  \mathrm{CAY} \!\sslash\! \mathrm{Spin}(7)
  \quad
  \text{and}
  \quad
  \mathrm{CAY}_{\mathrm{sL}} \!\sslash\! \mathrm{Spin}(6)
\end{equation}
are the \emph{moduli spaces} for Cayley 4-planes
and for special Lagrangian Cayley 4-planes, respectively:
for $X$ a $\mathrm{Spin}(7)$-manifold, a dashed lift in
$$
  \xymatrix{
    &&
    \mathrm{CAY} \!\sslash\! \mathrm{Spin}(7)
    \ar[d]
    \\
    X
    \ar@/^1pc/@{-->}[urr]
    \ar[rr]
    &&
    B \mathrm{Spin}(7)
  }
  \phantom{AAAAAAAAA}
  \xymatrix{
    &&
    \mathrm{CAY}_{\mathrm{sL}} \!\sslash\! \mathrm{Spin}(6)
    \ar[d]
    \\
    X
    \ar@/^1pc/@{-->}[urr]
    \ar[rr]
    &&
    B \mathrm{Spin}(6)
  }
$$
is a distribution on $X$ by tangent spaces to (special Lagrangian)
calibrated submanifolds.

\begin{prop}[Calibrations from twisted cohomotopy]
  \label{ModuliSpaceOfCalibratedFromTwistedCohomotopy}
  The moduli spaces of (special Lagrangian) Cayley 4-planes   \eqref{ModuliSpacesOfCayley4Planes}
  are compatibly weakly homotopy equivalent to the   coefficient spaces for twisted Cohomotopy
  jointly in degrees 4 and 7,   according to Prop. \ref{EquihH}:
  $$
    \xymatrix{
      \mathrm{CAY}_{\mathrm{sL}} \!\sslash \! \mathrm{Spin}(6)
      \ar@{}[r]|-{\simeq}
      \ar[d]
      &
      S^7
        \!\sslash\!
      \big(
        \mathrm{Sp}(2)
        \boldsymbol{\cdot}
        \mathrm{Sp}(1)
      \big)
      \ar[d]
      \\
      \mathrm{CAY} \!\sslash \! \mathrm{Spin}(7)
      \ar@{}[r]|-{\simeq}
      &
      S^4
        \!\sslash\!
      \big(
        \mathrm{Sp}(2)
        \boldsymbol{\cdot}
        \mathrm{Sp}(1)
      \big)
    }
  $$
\end{prop}
\begin{proof}
  By \cite[Theorem 1.38]{HarveyLawson82}
  (see also \cite[(3.19)]{BryantHarvey89}, \cite[(5.20)]{GluckMackenzieMorgan95})
  we have a coset space realization
$$
  \mathrm{CAY}
  \;\simeq\;
  \mathrm{Spin}(7)
  /
  \big(
    \mathrm{Spin}(4)
    \boldsymbol{\cdot}
    \mathrm{Spin}(3)
  \big)\;.
$$
and by  \cite[p. 7]{BBMOOY96} we have a coset space realization
$$
  \mathrm{CAY}_{\mathrm{sl}}
  \;\simeq\;
  \mathrm{Spin}(6)
  /
  \big(
    \mathrm{Spin}(3)
    \boldsymbol{\cdot}
    \mathrm{Spin}(3)
  \big)
  \;\simeq\;
  \mathrm{SU}(6)/ \mathrm{SO}(4)\;.
$$
By Lemma \ref{fibBiota} this means equivalently that
there are weak homotopy equivalences
$$
  \mathrm{CAY} \!\sslash\! \mathrm{Spin}(7)
  \;\simeq\;
  B
  \big(
    \mathrm{Spin}(4)
    \boldsymbol{\cdot}
    \mathrm{Spin}(3)
  \big)
  \;\simeq\;
  B
  \big(
    \mathrm{Sp}(1)
    \boldsymbol{\cdot}
    \mathrm{Sp}(1)
    \boldsymbol{\cdot}
    \mathrm{Sp}(2)
  \big)
$$
and
$$
  \mathrm{CAY}_{\mathrm{sL}}
  \!\sslash\! \mathrm{Spin}(6)\
  \;\simeq\;
  B
  \big(
    \mathrm{Spin}(3)
    \boldsymbol{\cdot}
    \mathrm{Spin}(3)
  \big)
  \;\simeq\;
  B
  \big(
    \mathrm{Sp}(1)
    \boldsymbol{\cdot}
    \mathrm{Sp}(1)
  \big)
  \,.
$$
This then implies the claim by Prop. \ref{hHsslash}.
\end{proof}

\medskip

\noindent {\bf Vanishing PT-charge in twisted Cohomotopy.}
Even without discussing a full generalization of the
untwisted Pontrjagin-Thom theorem \eqref{ClassicalPTIsIso}
to the case of twisted Cohomotopy (Def. \ref{TwistedCohomotopy}),
we may say what it means for a cocycle in twisted Cohomotopy
to correspond to the empty submanifold, hence to
correspond to vanishing brane charge in the sense discussed
above. This is all that we will need to refer to below in
\cref{CFieldBackgroundCharge}
and
\cref{M2BraneTadpoleCancellation}:

\begin{enumerate}[{\bf (i)}]
\vspace{-2mm}
\item In the case of untwisted cohomotopy it is immediate that the zero-charge cocycle is simply
the one represented by any function that does not meet the given
base point $\mathrm{pt} \in S^n$ \eqref{BasePointInSphere}.
\vspace{-2mm}
\item
In the case of twisted Cohomotopy
according to Def. \ref{TwistedCohomotopy}, this chosen point must be a chosen \emph{section}
of the given spherical fibration corresponding to the given twist $\tau$:
$$
  \xymatrix@R=1em{
    &&
    S^n \!\sslash\! \mathrm{O}(n+1)
    \ar[dd]
    \\
    \\
    X
    \ar@{-->}[uurr]^{ \mathrm{pt} }
    \ar[rr]_-{\tau}
    &&
    B \mathrm{O}(n+1)
  }
$$
which serves over each $x \in X$ as the point $\mathrm{pt}_x \in E_x \simeq S^4$ at which
we declare to form the inverse image of another given section,
under a parametrized inverse Pontrjagin-Thom construction.
\vspace{-2mm}
\item
With that section $\mathrm{pt}$ chosen, any
other twisted Cohomotopy cocycle
$[c_0] \in \pi^\tau(X)$  which will yield the
empty submanifold under parametrized Pontrjagin-Thom
must be represented by a section
$c_0$ which is everywhere distinct from $\mathrm{pt}$,
$$
  c_0(x) \neq \mathrm{pt}_x
$$
so that $c_0^{-1}(\mathrm{pt}(x)) = \varnothing$
for all $x \in X$.
\vspace{-2mm}
\item
But such a choice of a pair of pointwise distinct sections is equivalently a reduction of the
structure group not just along $\mathrm{O}(4) \hookrightarrow \mathrm{O}(5)$
as in Remark \ref{TwistedCohomotopyIsReductionOfStructureGroup}, but
is rather a reduction all the way along $\mathrm{O}(3)\hookrightarrow \mathrm{O}(5)$.
\end{enumerate}
\vspace{-1.5mm}

Specified to the
$\mathrm{Sp}(2) \boldsymbol{\cdot} \mathrm{Sp}(1)$-twisted
Cohomotopy jointly in degrees 4 and 7, from \cref{TwistedCohomotopyInDegrees}
this says that vanishing of the brane charge
seen by degree 4 Cohomotopy cocycle via a putative parameterized PT theorem
is witnessed by a lift from
$B \big(  \mathrm{Spin}(5)  \boldsymbol{\cdot} \mathrm{Spin}(3) \big) $
all the way to
$B \big(  \mathrm{Spin}(3)  \boldsymbol{\cdot} \mathrm{Spin}(3) \big) $.
But comparison with
Prop. \ref{hHsslash}
(see also \hyperlink{FigureT}{\it Figure T}) shows the following:

\begin{lemma}[Vanishing of Cohomotopy charge means factorization through $h_{\mathbb{H}}$]
\label{VanishingBraneChargeInTermsOfPontrjaginThom}
The vanishing of cohomotopical brane charge
of $\mathrm{Sp}(2)\boldsymbol{\cdot}\mathrm{Sp}(1)$-twisted
Cohomotopy in degree 4 (\cref{TwistedCohomotopyInDegrees}),
in the sense of the above parametrized Pontrjagin-Thom
construction of corresponding branes,
is exhibited by factorizations of the degree-4 cocycle
through degree-7 Cohomotopy, via the
equivariant quaternionic Hopf fibration $h_{\mathbb{H}}$
of Prop. \ref{hHsslash}:
\begin{equation}
  \label{PTVanishingChargeAsLift}
  \raisebox{60pt}{
  \xymatrix@C=6em@R=1.5em{
    &&
    S^7
    \!\sslash\!
    \big(
      \mathrm{Sp}(2)
        \!\boldsymbol{\cdot}\!
      \mathrm{Sp}(1)
    \big)
    \ar[dd]^-{
      h_{\mathbb{H}}
      \sslash
      (
        \mathrm{Sp}(2)
          \boldsymbol{\cdot}
        \mathrm{Sp}(1)
      )
    }
    \ar[r]^-{ \simeq }
    &
    B
    \big(
      \mathrm{Spin}(3)
        \boldsymbol{\cdot}
      \mathrm{Spin}(3)
    \big)
    \ar[dd]
    \\
    \\
    &&
    S^4
      \!\sslash\!
    \big(
      \mathrm{Sp}(2)
        \!\boldsymbol{\cdot}\!
      \mathrm{Sp}(1)
    \big)
    \ar[dd]
    \ar[r]^-{ \simeq }
    &
    B
    \big(
      \mathrm{Spin}(4)
        \!\boldsymbol{\cdot}\!
      \mathrm{Spin}(3)
    \big)
    \ar[dd]
    \\
    \\
    X
    \ar[rr]_-{ \tau }
    \ar@/^1pc/@{-->}[uurr]|-{
      \mbox{\!\!
        \tiny
        \color{blue}
        \begin{tabular}{c}
          cocycle in
          \\
          twisted
          \\
          Cohomotopy
          \\
          in degree 4
        \end{tabular}
      \!}
    }
    \ar@/^2pc/@{-->}[uuuurr]|>>>>>>>>>>>>{
      \mbox{
        \tiny
        \color{blue}
        \begin{tabular}{c}
          PT-vanishing of
          \\
          cocycle in
          \\
          twisted Cohomotopy
          \\
          in degree 4
        \end{tabular}
      }
    }
    &&
    B
    \big(
      \mathrm{Sp}(2)
        \!\boldsymbol{\cdot}\!
      \mathrm{Sp}(1)
    \big)
    \ar[r]_-{ B \mathrm{tri} }^-{ \simeq }
    &
    B
    \big(
      \mathrm{Spin}(5)
        \!\boldsymbol{\cdot}\!
      \mathrm{Spin}(3)
    \big)\;.
  }
  }
\end{equation}
\end{lemma}

\medskip
We come back to this in Prop. \ref{PTVanishing4Flux}
and Prop. \ref{DifferentialFormDataOnExtendedSpacetime} below.

\medskip

This concludes our discussion of general properties of
twisted Cohomotopy theory. Now we turn,
in  \cref{CancellationFromCohomotopy}, to discussing how,
under \hyperlink{HypothesisH}{\it Hypothesis H},
these serve to yield anomaly cancellation in M-theory.

\section{C-field charge-quantized in twisted Cohomotopy}
\label{CancellationFromCohomotopy}

We consider now the setup of M-theory on 8-manifolds:
\begin{remark}
 \label{MTheoryOn8Manifolds}
For M-Theory on 8-manifolds
\cite{Witten95b}\cite{BeckerBecker96}\cite{SethiVafaWitten96},
spacetime is of the form
$\mathbb{R}^{2,1} \times X^8$, corresponding to a
background of parallel M2-branes which appear as singular
points in the 8-dimensional space $X^8$, or else
as points that would be singular were they included in $X^8$.
See also
\cite{Tsimpis06}
\cite{CMP13}\cite{PrinsTsimpis13}\cite{BabalicLazaroiu14a}\cite{Shahbazi15}\cite{BabalicLazaroiu14b}
\cite{BabalicLazaroiu15c}\cite{BabalicLazaroiu15d}.

M-theory on 8-manifolds with $\mathrm{Sp}(2)\cdot \mathrm{Sp}(1)$-structure
(as in Def. \ref{HypothesisHFor8Manifolds} below),
specifically on the quaternionic projective plane $\mathbb{H}P^2$
\cite[4.3]{McNamaraVafa19} (see also Example \ref{The8Manifold} below),
has been argued in \cite[pp. 75]{AtiyahWitten01} to be dual to
4d M-theory on $G_2$-manifolds in three different ways,
such as to plausibly yield proof of confinement in 4d gauge theory.

If the 8-manifold $X^8$ is elliptically fibered then
M-theory on $X^8$ has been argued to be T-dual to
phenomenologically interesting F-theory compactifications
on spacetimes of the form $\mathbb{R}^{3,1} \times \widetilde X^8$
\cite{CMP13}\cite{BGPP13}:
$$
  \overset{
    \mathclap{
    \mbox{
      \tiny
      \color{blue}
      \begin{tabular}{c}
        M-theory
        \\
        on 8-manifolds
      \end{tabular}
    }
    }
  }{
    \mathbb{R}^{2,1}
    \times
    X^8
  }
  \;\;\;\;
  \xleftrightarrow{
    \;\;\;
      \mbox{
        \tiny
        \color{blue}
        \begin{tabular}{c}
          T-duality
        \end{tabular}
      }
    \;\;\;
  }
  \;\;\;\;
  \overset{
    \mbox{
      \tiny
      \color{blue}
      \begin{tabular}{c}
        F-theory on
        \\
        8-manifolds
      \end{tabular}
    }
  }{
    \mathbb{R}^{3,1}
    \times
    \widetilde X^8
  }
$$
In particular, for $X^8$ of $\mathrm{Spin}(7)$-structure, the
resulting $\mathcal{N} = 1$ supersymmetry in 3d on the left is argued
\cite{Witten95b}\cite[3]{Witten95c}\cite[4.3]{Vafa96}\cite[p. 7]{Witten00}
to be dual to a peculiar ``$\mathcal{N} = \sfrac{1}{2}$'' supersymmetry
in 4d on the right,
which does enforce a vanishing cosmological constant, but does
not constrain the finite energy particle spectrum to be supersymmetric.
This is developed in \cite{BonettiGrimmPugh13}\cite{BGPP13}\cite{HLLZ19}\cite{HLLSZ19}.
\end{remark}

\medskip

For our purposes, the following is concretely the data in question:
\begin{defn}[The 8-manifold $X^8$]
\label{The8Manifold}
We consider $X^8$ to be a smooth 8-dimensional spin-manifold,
possibly with boundary, which is connected and simply connected.
Let $\nabla$ be any affine connection
on the tangent bundle $T X^8$.
We assume that $H^2(X^8, \mathbb{Z}_2) = 0$.
\end{defn}

\begin{remark}[Role of technical assumptions on the 8-manifold]
\label{AssumptionsOn8Manifold} We highlight the following:

\vspace{-1mm}
 \item {\bf (i)} The assumption in Def. \ref{The8Manifold} that $X^8$ be connected
is convenient but immaterial and easily dropped.

\vspace{-1mm}
\item {\bf (ii)}  The assumption that $X^8$ be simply connected should also be
immaterial, but is not so easily dropped: All proofs invoking
Sullivan models in the following should generalize
at least to nilpotent fundamental groups, but will be
much harder without this assumption.

\vspace{-1mm}
\item {\bf (iii)}  The choice of affine connection $\nabla$
in Def. \ref{The8Manifold} is just to bring in Chern-Weil theory
and only affects the explicit representatives
of characteristic forms in the following, not any of the
gauge/homotopy invariant statements.

\vspace{-1mm}
\item {\bf (iv)}  The assumption $H^2(X,\mathbb{Z}_2) = 0$
bluntly ensures that any specific obstruction classes that could appear
in this group vanishes.
This is used only in \cref{W7Cancellation}
and \cref{HalfIntegralCFieldFluxQuantization} below,
and in \cref{HalfIntegralCFieldFluxQuantization} we
only need that the specific obstruction class
$\varpi \in H^2(X,\mathbb{Z}_2)$ to direct product
$\mathrm{Sp}(2) \times \mathrm{Sp}(1)$-structure vanishes
(from Prop. \ref{ObstructionToDirectProductStructure}).
With this class thus assumed to vanish,
there is no essential loss of generality in assuming
$\mathrm{Sp}(2)$-structure.
\end{remark}

\begin{example}
  \label{HP2}
  The quaternionic projective plane $X^8 = \mathbb{H}P^2$
  satisfies the assumptions of Def. \ref{The8Manifold}.
  To see this, it is sufficient to observe that
  it is homotopy equivalent to the result of gluing
  an 8-cell to a 4-cell (with attaching map being the
  quaternionic Hopf fibration)
  $$
    \xymatrix@C=4em@R=1.5em{
      S^7
      \ar@{}[dr]|-{ \mbox{\tiny (po)} }
      \ar[r]^{ h_{\mathbb{H}} }
      \ar[d]
      &
      S^4
      \ar[d]
      \\
      \mathbb{D}^8
      \ar[r]
      &
      \mathbb{H}P^2
    }
  $$
  This cell structure immediately implies vanishing of
  all cohomology in degree $\leq 3$.
\end{example}

\medskip

\begin{defn}[{\it Hypothesis H} for M-theory on 8-manifolds]
\label{HypothesisHFor8Manifolds}
Given an 8-manifold $X^8$ (Def. \ref{The8Manifold})
we say that a pair of differential forms $(G_4, G_7)$
on $X^8$ satisfies {\it Hypothesis H}
if it is in the image
of the non-abelian Chern character map
(Def. \ref{RationalTwistedCohomotopy})
from J-twisted 4-Cohomotopy
$$
  \xymatrix@R=-3pt@C=4em{
    \pi^{i_7 \circ \tau}\big( X^8 \big)
    \ar[dddddd]_{ h_\ast }
    \\
    \\
    \\
    \\
    \\
    \\
    \pi^{i_4 \circ \tau}\big( X^8 \big)
    \ar[rr]^{
      L_{\mathbb{R}}
    }_-{
    }
    &&
    \pi^{i_4 \circ \tau}\big( X^8 \big)_{\mathbb{R}}
    \ar@{<<-}[r]^-{
    }
    &
     \;\,
  \big\{
    (G_4, G_7)
    \,\vert\,
    \cdots
  \big\}
    \; \ar@{_{(}->}[r]^-{
    }
    &
    \Omega^4(X^8) \times \Omega^7(X^8)
    \\
    \big[ c \big]
    \ar@{|->}[rr]
    &&
    \big[
      (
      G_4, G_7
      )
    \big]
  }
$$
for twists compatible with the quaternionic Hopf fibration,
which by Prop. \ref{EquihH} means
that  $\tau$ is a
topological $\mathrm{Sp}(2)\cdot \mathrm{Sp}(1)$-structure
on $X^8$, via \eqref{ActionOfSp2Sp1}.
\end{defn}

We now discuss some consequences of Hypothesis H,
as summarized in \hyperlink{Table1}{\it Table 1}.

\subsection{Special $G$-structures}

We discuss how \hyperlink{HypothesisH}{Hypothesis H} implies $\mathcal{N}=1$ $G$-structure as in
\eqref{CohomotopyAndSpin7}.

\medskip

\noindent {\bf Parallel spinors and $G$-structure.}
Conditions on a compactification manifold to admit
suitably parallel spinor sections and hence preserve
some amount of supersymmetry have commonly been phrased in
terms of special holonomy metrics (see e.g. \cite{Gubser02}).
But more generally, in the potential presence of flux,
an alternative characterization
is in terms of \emph{$G$-structure},
i.e. reductions of the structure group
of the tangent/frame bundle.
This was used already in the classical \cite{IPW}
but received more attention after it was re-amplified
in the context of  flux compactifications in \cite[Sec. 2]{GMPW04},
see also \cite[Sec. 2]{Koerber11}\cite{Gaillard11}\cite{DDG14}.
Discussion of $G$-structure specifically in the context
of M-theory on 8-manifolds (Remark \ref{MTheoryOn8Manifolds}) includes
\cite{Tsimpis06}\cite{CMP13}\cite{PrinsTsimpis13}\cite{BabalicLazaroiu14a}\cite{Shahbazi15}\cite{BabalicLazaroiu14b}
\cite{BabalicLazaroiu15c}\cite{BabalicLazaroiu15d}.

\begin{prop}
  \label{NIs1GStructures}
  Let $X^d$ be a spin-manifold of dimension $d \in \{5,6,7,8\}$.
  Then cocycles in J-twisted 7-Cohomotopy (Def. \ref{TwistedCohomotopy})
  are equivalent to topological $G$-structures
  on $X^d$ as follows:
  \begin{equation}
    \left[
    \raisebox{20pt}{
    \xymatrix@C=0pt{
      X^d
      \ar@{-->}[rr]^-{ \exists }_{\ }="s"
      \ar[dr]_{T X}^{\ }="t"
      &&
      S^7 \!\sslash\! \mathrm{Spin}(d)
      \ar[dl]
      \\
      & B \mathrm{Spin}(d)
      \ar@{=>} "s"; "t"
    }
    }
    \right]
    \;\;\in\;\;
    \pi^{ i_7 \circ T X }(X^d)
    \;\;\;
    \Leftrightarrow
    \;\;\;
    \mbox{
      \begin{tabular}{l}
        topological
        \\
        $G$-structure
      \end{tabular}
      for
      $
        G =
        \left\{
        \begin{array}{lcl}
          \mathrm{Spin}(7) & \vert & d = 8
          \\
          \mathrm{G}_2 & \vert & d = 7
          \\
          \mathrm{SU}(3) & \vert & d = 6
          \\
          \mathrm{SU}(2) & \vert & d = 5
        \end{array}
        \right.
      $
  }
  \end{equation}
  hence are equivalent precisely to
  those $G$-structures that correspond to $\mathcal{N}=1$ compactifications
 of F-theory, M-theory, and string theory, respectively.
 (e.g. \cite{AcharyaGukov04}\cite{BBS10}\cite{GSZ14}).
\end{prop}
\begin{proof}
  This is Prop. \ref{Cohomotopy7GStructure} used in Prop. \ref{TwistedCohomotopyIsReductionOfStructureGroup}
\end{proof}

\medskip

\subsection{DMW anomaly cancellation}
\label{W7Cancellation}

We prove that \hyperlink{Hypothesis H}{Hypothess H}
implies the DMW anomaly cancellation condition \eqref{SP1Sp2Structure}:

\medskip

\begin{prop}
\label{CohomologicalCharacterizationOfSpin3TimesSpin5StructureIn8d}
  Let $X^8$ be an 8-manifold as in Def. \ref{The8Manifold}.
  Then existence
  of topological $\mathrm{Sp}(2)\cdot \mathrm{Sp}(1)$-structure
  on $X^8$, as in Hypothesis H (Def. \ref{HypothesisHFor8Manifolds})
  implies the following:
\item {\bf (i)}
  The Euler class of the tangent bundle is
  proportional to the one-loop anomaly polynomial
    \begin{equation}
      \label{Spin5DotSpin3ImpliesI8RelatedToChi}
      \tfrac{1}{24}
      \rchi_8(T X^8)
      \;=\;
      I_8(T X^8)
      \;:=\;
      \tfrac{1}{48}
      \big(
        p_2(T X^8)
        +
        \tfrac{1}{4}
        \big(
          p_1(T X^8)
        \big)^2
      \big)
      \;\in\;
      H^8\big(X^8, \mathbb{R}\big)
      \;.
    \end{equation}
  \item {\bf (ii)}
  The degree-6 Stiefel-Whitney class vanishes:
    \begin{equation}
      \label{Vanishingw6}
      w_6
      \big(
        T X^8
      \big)
      =
      0\;
      \;\in\;
      H^6\big( X^8, \mathbb{Z}_2 \big)\,,
    \end{equation}
    and hence so does the
    degree-7 integral Stiefel-Whitney class $W_7 := \beta(w_6)$:
    \begin{equation}
      \label{VanishingW7}
      W_7
      \big(
        T X^8
      \big)
      =
      0\;
      \;\in\;
      H^7\big( X^8, \mathbb{Z} \big)\,.
    \end{equation}
\end{prop}
\begin{proof}
  This follows by applying \cite[Thm. 8.1 \& Rem. 8.2]{CV98b}.
\end{proof}

\medskip

\subsection{Curvature-corrected Bianchi identity}
\label{CFieldBackgroundCharge}

We prove that \hyperlink{Hypothesis H}{Hypothesis H}
implies the higher curvature corrected Bianchi identities
\eqref{RationalTwists} \eqref{CompatibleRationalTwists}.

\begin{prop}[Higher curvature corrections via Cohomotopy]
  Let $X^8$ be an 8-manifold as in Def. \ref{The8Manifold}. Then:

\item {\bf (i)} The general form of the rationally twisted rational Cohomotopy
  sets in degrees 4 and 7 is as in \eqref{RationalTwists}
  and \eqref{CompatibleRationalTwists}.

 \item {\bf (ii)}   If the
  differential forms $(G_4,G_7)$ moreover satisfy
  Hypothesis H (Def. \ref{HypothesisHFor8Manifolds}),
  then the Cohomotopy set is concretely
  given as follows:
  $$
    \xymatrix{
    \pi^{i_4 \circ \tau}(X^8)_{\mathbb{R}}
    \simeq
    \left\{
    {\begin{aligned}
      d\,G_4 & = 0
      \\
      d\,G_7
        & =
        -
        \tfrac{1}{2}
        \widetilde G_4
          \wedge
        \big(
          \widetilde G_4 - \tfrac{1}{2}p_1(\nabla)
        \big)
        \;-\;
        12 \cdot I_8(\nabla)
    \end{aligned}}
    \right\}_{\big/\sim}
    \ar@{<<-}[r]
    &
     \;
    \big\{
      G_4, G_7
        \,\vert\,
      \cdots
    \big\}
    \;   \ar@{^{(}->}[r]
    &
      \Omega^4(X^8) \times \Omega^7(X^8)
    }\!,
  $$
  where $\widetilde G_4 := G_4 + \tfrac{1}{4}p_1(\nabla)$ from \eqref{Shifted4Flux}
  and $I_8 = \tfrac{1}{48}\big(p_2 + \tfrac{1}{4}p_1^2 \big)$ from \eqref{Spin5DotSpin3ImpliesI8RelatedToChi}.
\end{prop}
\begin{proof}
  The first statement is
  the specialization of Prop. \ref{SectionsOfRationalSphericalFibrations}
  to degrees 4\&7.
  For the second statement it then remains to
  re-express the class $\tfrac{1}{4}p_2$ of the
  effective $\mathrm{O}(5)$-twist \eqref{RationallyTwistedCohomtopyInEvenDegree}
  to the corresponding class of the given tangential
  $\mathrm{Sp}(2)$-twist
  as we pass through triality (Prop. \ref{QuaternionicSubgroupTriality})
  \begin{equation}
    \label{ATopologicalSp2Structure}
    \xymatrix@C=4em{
      &&
      B
        \mathrm{Sp}(2)
      \ar[d]
      \ar[r]^-{ \simeq }
      &
      B \mathrm{Spin}(5)
      \ar[d]
      \\
      X^8
      \ar[rr]_-{
        T X^8
      }
      \ar@{-->}[urr]^-{
        \tau
      }
      &&
      B \mathrm{Spin}(8)
      \ar[r]^-{\simeq}_-{ B \mathrm{tri} }
      &
      B \mathrm{Spin}(8)
      \,,
    }
  \end{equation}
  with the delooped triality automorphism \eqref{TrialityAutomorphismDelooped}
  shown on the right.
  We claim that this is the difference between the Euler class
  and the squared
  first fractional Pontrjagin class of $X^8$:
  \begin{equation}
    \label{Fractionalp2OfSpin5InTermsOfSp2}
    \tfrac{1}{4}p_2\big( B \mathrm{tri}_\ast(\tau)\big)
    \;=\;
    \Big(
      \tfrac{1}{4}p_1\big(
        T X^8
      \big)
    \Big)^2
    -
    \rchi_8\big( T X^8\big).
  \end{equation}
This follows by combining \eqref{PullbackOfp2UnderTriality} from
Prop. \ref{PullbackOfClassesAlongTriality} with
the $\mathrm{Sp}(2)$-structure relation
\eqref{Spin5DotSpin3ImpliesI8RelatedToChi}
\begin{equation}
  \label{EulerRelationOnSp2}
  \tfrac{1}{4}p_2 =  \big(\tfrac{1}{4}p_1 \big)^2 + \tfrac{1}{2} \rchi_8
  \;\;\;\;\;\;\;\;
  \mbox{on $B \mathrm{Sp}(2)$}
\end{equation}
from Prop. \ref{CohomologicalCharacterizationOfSpin3TimesSpin5StructureIn8d}.

In conclusion, this yields the claim by completing the square
on the right of the Bianchi identity:
$$
  \begin{aligned}
    d\, G_7
    & =
    -
    \tfrac{1}{2} G_4 \wedge G_4
    +
    \tfrac{1}{4}p_2(B \mathrm{tri}_\ast \nabla)
    \\
    & =
    -
    \tfrac{1}{2} G_4 \wedge G_4
    +
    \tfrac{1}{4}p_1(\nabla)
    \wedge
    \tfrac{1}{4}p_1(\nabla)
    -
    \rchi_8(\nabla)
    \\
    & =
    -
    \tfrac{1}{2}
    \underset{
      = \widetilde G_4
    }{
    \underbrace{
      \big(
        G_4 + \tfrac{1}{4}p_1(\nabla)
      \big)
    }
    }
    \wedge
    \underset{
      \mathclap{
      = \widetilde G_4 - \frac{1}{2}p_1(\nabla)
      }
    }{
    \underbrace{
      \big(
        G_4 - \tfrac{1}{4}p_1(\nabla)
      \big)
    }
    }
    -
    \underset{
      = 12 \cdot I_8(\nabla)
    }{
    \underbrace{
      \rchi_8(\nabla)
    }
    }.
  \end{aligned}
$$

\vspace{-7mm}
\end{proof}

\medskip

\subsection{Shifted 4-flux quantization}
\label{HalfIntegralCFieldFluxQuantization}

We prove that \hyperlink{HypothesisH}{Hypothesis H}
implies the shifted flux quantization condition \eqref{ShiftedIntegrality}.
The result is Prop. \ref{CohomologicalCharacterizationOfTwistedCohomotopyDegree4}
below.
The basic observation that makes this work is highlighted in Remark \ref{CFieldUniversal} below.
To put this to full use we need to
go into some technicalities in Lemma \ref{FreeIntegralCohomologyOfBSpin4Spin3} and
Lemma \ref{IntegralUniversalCfieldOnBSpin4xSpin3} below.

\medskip

First we recall some classical facts about the
integral cohomology of $B \mathrm{Spin}(n)$ for low $n$:

\begin{lemma}
\label{BSpinRationalHomotopy}
{\bf (i)} The integral cohomology ring of $B \mathrm{SO}(3)$
is
\begin{equation}
  \label{HBSO3}
  H^\bullet
  (
    B \mathrm{SO}(3);
    \mathbb{Z}
  )
  \;\simeq\;
  \mathbb{Z}
  \big[
    p_1, W_3
  \big]/(2 W_3)
  \,,
\end{equation}
and the integral cohomology of
$B \mathrm{Spin}(3)$
is free on one generator
\begin{equation}
  H^\bullet
  \big(
    B \mathrm{Spin}(3);
    \mathbb{Z}
  \big)
  \;\cong\;
  \mathbb{Z}
  \big[
    \tfrac{1}{4}p_1
  \big],
\end{equation}
while the integral cohomology ring of $B \mathrm{Spin}(4)$ is
free on two generators
\begin{equation}
  \label{IntegralCohomologyOfSpin4}
  H^\bullet
  \big(
    B \mathrm{Spin}(4)
    ;
    \mathbb{Z}
  \big)
  \;\simeq\;
  \mathbb{Z}
  \big[
    \tfrac{1}{2}p_1
    ,
    \underset{
      =: \widetilde \Gamma_4
    }{
    \underbrace{
      \underset{
        =:
        \Gamma_4
      }{
      \underbrace{
        \tfrac{1}{2}\rchi_{4}
      }
      }
      +
      \tfrac{1}{4} p_1
    }
    }
  \big]
  \,,
\end{equation}
where $p_1$ is the first Pontrjagin class and $\rchi_{4}$ the Euler class.
\item {\bf (ii)} Under the exceptional isomorphism
$\vartheta \;\colon\; \mathrm{Spin}(3) \times \mathrm{Spin}(3) \overset{\simeq}{\longrightarrow} \mathrm{Spin}(4)$
these classes are related by
\begin{equation}
  \label{FirstPontrjaginFromBSpin4ToBSpin3BSpin3}
  \begin{aligned}
  \vartheta^\ast \left( \tfrac{1}{2}p_1 \right)
  & =
  \phantom{-} \tfrac{1}{4}p_1 \otimes 1 + 1 \otimes \tfrac{1}{4} p_1\;,
  \\
  \vartheta^\ast
  \big(
    \tfrac{1}{2}\rchi
    +
    \tfrac{1}{4} p_1
  \big)
  & =\phantom{-}
   \tfrac{1}{4}p_1 \otimes 1\;,
  \phantom{+1 \otimes \tfrac{1}{4}p_1 }
\\
 \text{hence}
 \phantom{AAAAAA}
 \vartheta^\ast(\rchi)
 \phantom{ 1 }\;\;
 &
 =\phantom{-} \tfrac{1}{4}p_1 \otimes 1 - 1 \otimes \tfrac{1}{4} p_1\;.
 \end{aligned}
\end{equation}
\end{lemma}
\begin{proof}
  This follows from classical results \cite{Pittie91}.   More explicitly, \eqref{HBSO3}
  is a special case of   \cite[Thm. 1.5]{Brown82},   recalled for instance as
  \cite[Thm. 4.2.23 with Remark 4.2.25]{RudolphSchmidt17}.
  The other statements are recalled for instance in \cite[Lemma 2.1]{CV98a}.
\end{proof}

\begin{remark}[Universal avatar of the integral C-field]
  \label{CFieldUniversal}
  We highlight from \eqref{IntegralCohomologyOfSpin4}, under the braces, the universal integral class
  \begin{equation}
    \label{UniversalCField}
    \widetilde \Gamma_4
    \;:=\;
    \underset{
      =: \Gamma_4
    }{
      \underbrace{
        \tfrac{1}{2}\rchi_{4}
      }
    }
    +
    \tfrac{1}{4}p_{{}_1}
    \;\in\;
    H^4
    (
      B \mathrm{Spin}(4)
      ;
      \mathbb{Z}
    )
  \end{equation}
  for use below. Prop. \ref{CohomologicalCharacterizationOfTwistedCohomotopyDegree4}
  below says that,
  under \hyperlink{HypothesisH}{Hypothesis H}, these universal characteristic classes
  are the avatars of the
  half-integral shifted C-field flux $\widetilde G_4$.
  Since $\widetilde \Gamma_4$ is an integral cohomology class,
  its pullback to any given spacetime
  is an integral class, and such that its image in de Rham cohomology
  is $[G_4 + \tfrac{1}{4}p_1(\nabla)]$.
  This integral lift is what implements
  the shifted C-field flux quantization condition in M-theory
  \cref{HalfIntegralCFieldFluxQuantization}.
\end{remark}

We now trace the integral generator
$\widetilde \Gamma_4$ in \eqref{UniversalCField} to
the larger group $\mathrm{Spin}(5)\boldsymbol{\cdot} \mathrm{Spin}(3)$.
\begin{lemma}[Cohomology of the central group]
  \label{FreeIntegralCohomologyOfBSpin4Spin3}
  The integral cohomology in degree 4 of the classifying space
  of the central product group \eqref{Spin4DotSpin3}
  $$
    \mathrm{Spin}(4) \boldsymbol{\cdot} \mathrm{Spin}(3)
    \;\simeq\;
    \mathrm{Spin}(3) \boldsymbol{\cdot}
    \mathrm{Spin}(3) \boldsymbol{\cdot} \mathrm{Spin}(3)
  $$
  is the free lattice
  \vspace{-4mm}
  \begin{equation}
    \label{TheSublattice}
    H^4
    \big(
      B
      (
        \mathrm{Spin}(4)\boldsymbol{\cdot} \mathrm{Spin}(3)
      );
      \mathbb{Z}
    \big)
    \;\simeq\;
    \mathbb{Z}
    \left\langle
    \renewcommand{\arraystretch}{.15}
    \small
      \begin{array}{ccccc}
      \tfrac{1}{4}p^{(1)}_{{}_1}
      &+&
      \tfrac{1}{4}p^{(2)}_{{}_1}
      &+&
      \tfrac{2}{4}p^{(3)}_{{}_1}
      ,
      \\
      \\
      \tfrac{1}{4}p^{(1)}_{{}_1}
      &+&
      \tfrac{2}{4}p^{(2)}_{{}_1}
      &+&
      \tfrac{1}{4}p^{(3)}_{{}_1}
      ,
      \\
      \\
      \tfrac{2}{4}p^{(1)}_{{}_1}
      &+&
      \tfrac{1}{4}p^{(2)}_{{}_1}
      &+&
      \tfrac{1}{4}p^{(3)}_{{}_1}
      \end{array}
    \right\rangle
  \end{equation}
  where
  $
    p^{(k)}_{{}_{1}}
    \;:=\;
    (B \mathrm{pr}_k)^\ast
    (
      p_{{}_1}
    )
  $
  is the pullback of the first Pontrjagin class
  along the projection \eqref{Spinn1Spinn2ExactSequences}
  $$
    B
    \big(
      \mathrm{Spin}(4)\boldsymbol{\cdot}
      \mathrm{Spin}(3)
    \big)
    \simeq
    B
    \big(
      \mathrm{Spin}(3)
        \boldsymbol{\cdot}
      \mathrm{Spin}(3)
        \boldsymbol{\cdot}
      \mathrm{Spin}(3)
    \big)
    \xrightarrow{
      B \mathrm{pr}_k
    }
    B \mathrm{SO}(3)
    \,.
  $$
\end{lemma}
\begin{proof}
  The defining short exact sequence of groups (Def. \ref{Def-dot})
  $$
    1
    \longrightarrow
      \mathbb{Z}_2
    \longrightarrow
      \mathrm{Spin}(3)
         \boldsymbol{\cdot}
      \mathrm{Spin}(3)
         \boldsymbol{\cdot}
      \mathrm{Spin}(3)
    \longrightarrow
      \mathrm{Spin}(3)
         \times
      \mathrm{Spin}(3)
         \times
      \mathrm{Spin}(3)
    \longrightarrow
    1
  $$
  induces a homotopy fiber sequence of classifying spaces
  (e.g. \cite[11.4]{Mitchell11})
  $$
    \xymatrix{
      B \mathbb{Z}_2
      \ar[r]
      &
      B
      \big(
        \mathrm{Spin}(3)
          \times
        \mathrm{Spin}(3)
          \times
        \mathrm{Spin}(3)
      \big)
      \ar[r]
            &
      B
      \big(
        \mathrm{Spin}(3)
          \boldsymbol{\cdot}
        \mathrm{Spin}(3)
          \boldsymbol{\cdot}
        \mathrm{Spin}(3)
      \big).
    }
  $$
  The corresponding Serre spectral sequence  shows that
  $$
    \begin{aligned}
    H^4
    \big(
      B
      (
        \mathrm{Spin}(3)
          \boldsymbol{\cdot}
        \mathrm{Spin}(3)
          \boldsymbol{\cdot}
        \mathrm{Spin}(3)
      );
      \mathbb{Z}
    \big)
        \xymatrix{\; \ar@{^{(}->}[r]&} &
    H^4
    \big(
      B
      (
        \mathrm{Spin}(3)
          \times
        \mathrm{Spin}(3)
          \times
        \mathrm{Spin}(3)
      ),
      \mathbb{Z}
    \big)
    \\
    & \simeq
    \mathbb{Z}
    \big\langle
      \tfrac{1}{4}p^{(1)}_{{}_1},
      \tfrac{1}{4}p^{(2)}_{{}_1},
      \tfrac{1}{4}p^{(3)}_{{}_1}
    \big\rangle
    \end{aligned}
  $$
  is a sublattice of index 4. This sublattice must include the integral class $\tfrac{1}{2}p_{{}_1}$
  pulled back along the inclusion into $\mathrm{Spin}(7)$,  which by Lemma \ref{BSpinRationalHomotopy} is
  \begin{equation}
    \label{TripleForp1}
    \xymatrix@R=-1pt{
      B
      \big(
        \mathrm{Spin}(4)
          \boldsymbol{\cdot}
        \mathrm{Spin}(3)
      \big)
      \ar[rr]
      &&
      B \mathrm{Spin}(7)\;.
      \\
      \tfrac{1}{4}p_{{}_1}^{(1)}
      +
      \tfrac{1}{4}p_{{}_1}^{(2)}
      +
      \tfrac{2}{4}p_{{}_1}^{(3)}
      &&
      \tfrac{1}{2}p_{{}_1}
      \ar@{|->}[ll]
    }
  \end{equation}
  But then it must also  contain the  images of this element under the delooping of the
  $S_3$-automorphisms \eqref{S3ActionOnSpin4Spin3}.  This yields
  the other two elements shown in \eqref{TheSublattice}.
  Finally, it is clear that the sublattice spanned by these  three elements already has full rank and index 4:
  \begin{equation}
    \label{SublatticeOfTriplesSummingTo4}
    \mathbb{Z}
    \left\langle
    \renewcommand{\arraystretch}{.15}
 \small      \begin{array}{ccccc}
      \tfrac{1}{4}p^{(1)}_{{}_1}
      &+&
      \tfrac{1}{4}p^{(2)}_{{}_1}
      &+&
      \tfrac{2}{4}p^{(3)}_{{}_1}
      ,
      \\
      \\
      \tfrac{1}{4}p^{(1)}_{{}_1}
      &+&
      \tfrac{2}{4}p^{(2)}_{{}_1}
      &+&
      \tfrac{1}{4}p^{(3)}_{{}_1}
      ,
      \\
      \\
      \tfrac{2}{4}p^{(1)}_{{}_1}
      &+&
      \tfrac{1}{4}p^{(2)}_{{}_1}
      &+&
      \tfrac{1}{4}p^{(3)}_{{}_1}
      \end{array}
    \right\rangle
    \;\simeq\;
    \Big\{
      \tfrac{a}{4}p^{(1)}_{{}_1}
      +
      \tfrac{b}{4}p^{(2)}_{{}_1}
      +
      \tfrac{c}{4}p^{(3)}_{{}_1}
      \;\vert\;
      a,b,c \in \mathbb{Z}, \;
      a + b + c = 0 \,\; \mathrm{mod}\, 4
    \Big\}
  \end{equation}
  which means that there are no further generators.
\end{proof}

As a direct consequence we obtain the following identification.
\begin{lemma}[Integral classes]
  \label{IntegralUniversalCfieldOnBSpin4xSpin3}
  The following cohomology class
  on the classifying space of the group
  $\mathrm{Spin}(4) \boldsymbol{\cdot} \mathrm{Spin}(3)$ \eqref{Spin4DotSpin3},
  which a priori is in
  rational cohomology, is in fact integral:
  $$
    \underset{
      =: \widetilde \Gamma_4
    }{
    \underbrace{
      \tfrac{1}{2}\rchi_4
      +
      \tfrac{1}{4}p_{{}_1}
    }
    }
    \;+\;
    \tfrac{1}{2}p^{(3)}_{{}_1}
    \;\in\;
    H^4\big(
      \mathrm{Spin}(4)
       \!\boldsymbol{\cdot}\!
      \mathrm{Spin}(3)
      ;
      \mathbb{Z}
    \big)
  $$
  and hence so is its image on the classifying space of
  $\mathrm{Sp}(1) \!\boldsymbol{\cdot}\! \mathrm{Sp}(1) \!\boldsymbol{\cdot}\! \mathrm{Sp}(1)$ \eqref{Sp1Sp1Sp1}
  under the delooping of the triality isomorphism
  from Prop. \ref{QuaternionicSubgroupTriality}, which we
  will denote by the same symbols:
  \begin{equation}
    \label{tildeGamma4OnBSp1Sp1Sp1}
    \underset{
      =: \widetilde \Gamma_4
    }{
    \underbrace{
      \tfrac{1}{2}\rchi_4
      +
      \tfrac{1}{4}p_{{}_1}
    }
    }
    \;+\;
    \tfrac{1}{2}p^{(3)}_{{}_1}
    \;\in\;
    H^4\big(
      \mathrm{Sp}(1)
        \!\boldsymbol{\cdot}
      \mathrm{Sp}(1)
        \!\boldsymbol{\cdot}
      \mathrm{Sp}(1)
      ;
      \mathbb{Z}
    \big)
    \;\simeq\;
    H^4\big(
      \mathrm{Spin}(4)
       \!\boldsymbol{\cdot}\!
      \mathrm{Spin}(3)
      ,
      \mathbb{Z}
    \big)
    \,.
  \end{equation}
  Here $\tfrac{1}{2}\rchi_4$ is the Euler class  pulled back back from the left
  $B \mathrm{SO}(4)$ factor and   $p^{(3)}_{{}_1}$ is the   first Pontrjagin
  class pulled back from   the right $B \mathrm{SO}(3)$ factor, both along the respective
  projections \eqref{Spinn1Spinn2ExactSequences}, while   $p_{{}_1}$ is the first
  Pontrjagin class pulled back from the ambient  $B \mathrm{Spin}(8)$ along the
  canonical inclusion \eqref{Spinn1Spinn2SubgroupInclusion}:
  $$
    \xymatrix@R=2pt{
      &
      B
      \big(
        \mathrm{Spin}(4) \!\boldsymbol{\cdot}\! \mathrm{Spin}(3)
      \big)
      \ar[ddddl]_-{\!\!\!\! B \mathrm{pr}_4 }
      \ar[ddddr]^-{~B \mathrm{pr}_3 }
      \ar[dddd]|-{ B \iota_8 }
      \\
      \\
      \\
      \\
      B \mathrm{SO}(4)
      &
      B \mathrm{Spin}(8)
      &
      B \mathrm{SO}(3)
      \\
      \rchi_4
      &
      p_{{}_1}
      &
      p^{(3)}_{{}_1}
    }
  $$
\end{lemma}
\begin{proof}
  In terms of the contributions from the three factors
  under the identification
  $\mathrm{Spin}(4)\boldsymbol{\cdot}\mathrm{Spin}(3)
  \simeq \mathrm{Spin}(3)\boldsymbol{\cdot}\mathrm{Spin}(3)
  \boldsymbol{\cdot}\mathrm{Spin}(3)$
  the class in question is
  $$
    \underset{
      = \tfrac{1}{2}\rchi_4
    }{
    \underbrace{
      \tfrac{1}{8}p^{(1)}_{{}_1}
      -
      \tfrac{1}{8}p^{(2)}_{{}_1}
    }
    }
    +
    \underset{
      = \tfrac{1}{4}p_{{}_1}
    }{
    \underbrace{
      \tfrac{1}{8}p^{(1)}_{{}_1}
      +
      \tfrac{1}{8}p^{(2)}_{{}_1}
      +
      \tfrac{1}{4}p^{(3)}_{{}_1}
    }
    }
    +
    \tfrac{2}{4}p^{(3)}_{{}_1}
    \;=\;
    \tfrac{1}{4}p^{(1)}_{{}_1}
    +
    \tfrac{3}{4}p^{(3)}_{{}_1}
    \,,
  $$
  where under the braces we used Lemma \ref{BSpinRationalHomotopy}
  as in \eqref{TripleForp1}.
  The equivalent expression on the right
  makes manifest that this is in the sublattice
  \eqref{SublatticeOfTriplesSummingTo4}. Therefore,
  Lemma \ref{FreeIntegralCohomologyOfBSpin4Spin3} implies the
  claim.
\end{proof}

Now we may finally state and prove the main result of this section.

\begin{prop}[Integrality of the shifted class]
  \label{CohomologicalCharacterizationOfTwistedCohomotopyDegree4}
  Let $X^{8}$ be a 8-manifold as in Def. \ref{The8Manifold}.
  If a differential 4-form $G_4$ on $X^8$ satisfies
  Hypothesis H (Def. \ref{HypothesisHFor8Manifolds}), then
  its shift by 1/4th the first Pontrjagin form
  \begin{equation}
    \label{Shifted4Flux}
    \widetilde G_4
    \;:=\;
    G_4 \;+\; \tfrac{1}{4}p_1(\nabla)
  \end{equation}
  is integral:
  \begin{equation}
    \label{ShiftedIntegrality}
    \big[
      G_4 + \tfrac{1}{4}p_1(\nabla)
    \big]
    \;=\;
    [G_4] + \tfrac{1}{4}p_1\big(T X^8\big)
    \;\in\;
    H^4\big(X^8; \mathbb{Z}\big)
    \to
    H^4\big(X^8; \mathbb{R}\big)
    \,.
  \end{equation}
\end{prop}
\begin{proof}
  Notice that,
  by the assumption $H^2(X^8, \mathbb{Z}_2) = 0$ in Def. \ref{The8Manifold},
  it follows in particular that
  the characteristic class $\varpi$,
  from Def. \ref{Epsilon}, vanishes:
  \begin{equation}
    \label{VanishingEpsilonForIntegralFlux}
    \varpi(\tau) \in H^2\big(X^{8}; \mathbb{Z}_2\big) \;=\; 0\;.
  \end{equation}
With this, the proof proceeds by considering the following diagram,
which we will discuss below in stages:
\begin{equation}
  \label{DiagramProvingHalfIntegralFluxQuantization}
  \raisebox{60pt}{
  \xymatrix@R=5pt{
    && &
    \ar@{|->}@/_4pc/[llldddd]
    \overset{
      \widetilde \Gamma_4
    }{
    \overbrace{
      \tfrac{1}{2}\rchi_4
      +
      \tfrac{1}{4}p_1
    }
    }
    \;
    +
    \;
    \;\;
    \tfrac{1}{2}p_1^{(3)}
    &
    \\
    &&
    B\big(
      \mathrm{Sp}(1)
        \!\boldsymbol{\cdot}\!
      \mathrm{Sp}(1)
        \!\boldsymbol{\cdot}\!
      \mathrm{Sp}(1)
    \big)
    \ar[r]^-{ \simeq }_-{ }
    \ar[dd]
    &
    B
    \big(
      \mathrm{Spin}(4)
      \!\boldsymbol{\cdot}\!
      \mathrm{Spin}(3)
    \big)
    \ar[dd]
    \ar[r]^-{ B \mathrm{pr}_{4} }
    &
    B
      \mathrm{SO}(4)
    \ar[dd]
    \\
    \\
    &&
    B\big(
      \mathrm{Sp}(2)
        \!\boldsymbol{\cdot}\!
      \mathrm{Sp}(1)
    \big)
    \ar[r]^-{ \simeq }_-{  }
    \ar[dd]
    &
    B
    \big(
      \mathrm{Spin}(5)
      \!\boldsymbol{\cdot}\!
      \mathrm{Spin}(3)
    \big)
    \ar[dd]
    \ar[r]^-{ B \mathrm{pr}_5 }
    &
    B
      \mathrm{SO}(5)
    \ar[dd]
    \\
    \mathllap
    {
    [G_4]
    +
    \tfrac{1}{4}p_1(T X^8)
    \!\!
    }
    \\
    X
    \ar[rr]|-{ T X }
    \ar@{-->}[uurr]^-{ \tau }
    \ar@/^1pc/@{-->}[uuuurr]|-{
      \mbox{
        \tiny
        \color{blue}
        \begin{tabular}{c}
          cocycle in
          \\
          $\tau$-twisted
          \\
          Cohomotopy
        \end{tabular}
      }
    }^>>>>>{ c }
    &&
    B \mathrm{Spin}(8)
    \ar[r]_-{\simeq}^-{ B \mathrm{tri} }
    &
    B \mathrm{Spin}(8)
    \ar[r]
    &
    B \mathrm{SO}(8)
    \\
    p_1\big( T X^8 \big)
    &&
    \ar@{|->}[ll]
    p_1
    &
    p_1
    \ar@{|->}[l]
  }
  }
\end{equation}
Here  the vertical maps are the deloopings of the canonical group inclusions
(Remark \ref{SubgroupInclusions})
and the horizontal equivalences $B\mathrm{tri}$ are the  deloopings
\eqref{TrialityAutomorphismDelooped}
of the  respective triality automorphism from Prop. \ref{QuaternionicSubgroupTriality},
while the horizontal  maps $B \mathrm{pr}_n$ are the deloopings of the canonical  projections \eqref{Spinn1Spinn2ExactSequences}.
On the left we used that, by Def. \ref{TwistedCohomotopy}, an element
$$
  [c] \in
  \pi^\tau\big( X^8 \big)
$$
in the
$\tau$-twisted Cohomotopy of $X^8$ is the homotopy class of
a section $c$ of the $S^4$-bundle classified by $ B \mathrm{pr}_5 \circ B \mathrm{tri}\circ \tau$:
$$
  \xymatrix@R=1em@C=3.5em{
    S^4
    \ar[r]
    &
    E
    \ar[dd]_-{\pi}
    \ar[rr]^{  }
    \ar@{}[ddrr]|-{
      \mbox{
        \tiny
        (pb)
      }
    }
    &&
    B \mathrm{SO}(4)
    \mathrlap{
      \; \simeq S^4 \!\sslash\!  \mathrm{SO}(5)
    }
    \ar[dd]
    \\
    \\
    &
    X^{8}
    \ar@/^1.5pc/@{-->}[uu]^c
    \ar[rr]_-{
      B \mathrm{pr}_5 \;\circ \;  B \mathrm{tri}\; \circ \; \tau
    }
    &&
    B \mathrm{SO}(5)
  }
$$
and we used Prop. \ref{hHsslash} to identify various homotopy quotients
of $S^4$ with classifying spaces, as shown.
   This shows that    $E$ is  the unit sphere bundle of a  rank 5 real vector bundle $V$
   classified by $B \mathrm{pr}_5 \circ B \mathrm{tri} \circ c$.   Therefore, by
   Prop. \ref{SectionsOfRationalSphericalFibrations} we have
  $$
    \pi^*[G_4] = \tfrac{1}{2}\rchi_4\big(\widehat{V})
    \,,
  $$
  where $\widehat{V}$ is defined by the splitting $\pi^*V=\mathbb{R}_E\oplus \widehat{V}$ determined
  by the tautological section of $\pi^*V$ over $E$,   i.e., it is the rank 4 real vector bundle on $E$
  classified by $E \to B \mathrm{SO}(4)$.
  Hence, by \eqref{tildeGamma4OnBSp1Sp1Sp1}  in
  Lemma \ref{IntegralUniversalCfieldOnBSpin4xSpin3},
  we have that
  \[
    \pi^*\Big(
      \;
      \underset{
        =:
        K
      }{
      \underbrace{
        [G_4]
          +
        \tfrac{1}{4}p_1\big( B \mathrm{tri} \circ T X^8 \big)
          +
        \tfrac{1}{2}p_1^{(3)}\big(  B \mathrm{tri} \circ \tau  \big)
      }
      }
      \;
    \Big)
    \;\in\;
    H^4(E;\mathbb{Z})
  \]
  is an integral class.

 We now claim that the class $K$ is integral
 already before the pullback, as a class on $X$.
 For this, consider the commutative diagram
 \[
 \xymatrix@C=3em{
 \cdots
   \ar[r]
   &
   H^4\big(X^8; \mathbb{Z}\big)
   \ar[r]
   \ar[d]^{\pi^*}
   &
   H^4\big(X^8; \mathbb{R}\big)
   \ar[r]^-{q}
   \ar[d]^{\pi^*}
   &
   H^4\big(X^8; \mathbb{R}/\mathbb{Z}\big)
   \ar[r]
   \ar[d]^{\pi^*}
   &
  \cdots
  \\
  \cdots
  \ar[r]
  &
  H^4\big(E; \mathbb{Z}\big)
  \ar[r]
  &
  H^4\big(E; \mathbb{R}\big)\ar[r]^-{q}
  &
  H^4\big(E; \mathbb{R}/\mathbb{Z}\big)
  \ar[r]
  &
 \cdots
 }
 \]
 induced by the short exact sequence $0\to \mathbb{Z} \to \mathbb{R} \to \mathbb{R}/\mathbb{Z}\to 0$.
 From the Serre spectral sequence for the fibration $\pi\colon E\to X$ one sees that the vertical maps in the
 above diagram are injective. Consequently, from
 \[
   \begin{aligned}
      \pi^*
      q (K)
      & =
      q  \pi^*(K)
      \\
      & = 0
   \end{aligned}
 \]
 it follows that already
 $q(K)=0$,
 which means that $K$ itself is integral:
 \begin{equation}
   \label{AClassOnX}
   [G_4]
     +
   \tfrac{1}{4}p_1\big( B \mathrm{tri} \circ T X^8 \big)
     +
   \tfrac{1}{2}p_1^{(3)}(  B \mathrm{tri} \circ \tau)
   \;\in\;
   H^4(
     X^8; \mathbb{Z}
   )
   \,.
 \end{equation}
 Now observe that the third summand in \eqref{AClassOnX}  is the  first fractional Pontrjagin class of the
 underlying $\mathrm{SO}(3)$-bundle. By the assumption \eqref{VanishingEpsilonForIntegralFlux}
 this admits Spin structure, by Lemma \ref{Epsilon}.  This in turn implies that its first Pontrjagin class
 is divisible by two, hence that the last summand in  \eqref{AClassOnX} is integral by itself
 $$
   \tfrac{1}{2}p_1^{(3)}( B \mathrm{tri}\circ \tau)
   \;\in\;
   H^4\big( X^8; \mathbb{Z}\big)
   \,,
 $$
 and hence that also the remaining summand
 \begin{equation}
   \label{AlmostThere}
   [G_4]
     +
   \tfrac{1}{4}p_1
   \big(
     B \mathrm{tri} \circ T X^8
   \big)
   \;\in\;
   H^4\big(X^8; \mathbb{Z}\big)
 \end{equation}
 is integral by itself.
 Finally, pullback along the triality automorphism preserves
 the first Pontrjagin class, by Lemma \ref{PullbackOfClassesAlongTriality}
 \begin{equation}
   \label{BphiPreservesp1}
   p_1( B \mathrm{tri} \circ \tau)
     \;=\;
   p_1\big(T X^8\big)
 \end{equation}
 and hence \eqref{AlmostThere} indeed becomes
 $
    [G_4]
     +
   \tfrac{1}{4}p_1
   \big(
     T X^8
   \big)
   \in
   H^4\big(X^8; \mathbb{Z}\big)
 $.
\end{proof}

\medskip

\subsection{Background charge}
\label{BackgroundCharge}

We prove that \hyperlink{HypothesisH}{Hypothesis H} implies
the background charge \eqref{BackgroundChargeValue} of the 4-flux.

\medskip
\begin{prop}[Cohomotopically vanishing 4-flux form]
  \label{PTVanishing4Flux}
 Let $X^8$ be a smooth 8-manifold which is simply connected
 (Remark \ref{SimplyConnected}) and equipped with
 topological
 $\mathrm{Sp}(2)$-structure
 $\tau$ (Example \ref{CentralProductOfSymplecticGroups}).
 Then, if a cocycle in $\tau$-twisted Cohomotopy
 (Def. \ref{TwistedCohomotopy})
 has a factorization through the quaternionic Hopf fibration,
 exhibiting its vanishing PT-charge according to
 \eqref{PTVanishingChargeAsLift} in \cref{TwistedPT},
 it follows that the differential 4-form
 $G_4$ by Def. \ref{HypothesisHFor8Manifolds}
 has value
 $$
   G_4 \;=\; \tfrac{1}{4}p_1\big( \nabla_\tau \big)
   \,.
 $$
 Consequently, the corresponding integral 4-form $\widetilde G_4$
 \eqref{ShiftedIntegrality} from Prop. \ref{CohomologicalCharacterizationOfTwistedCohomotopyDegree4}
 has class $\tfrac{1}{2}p_1$:
 $$
   \raisebox{60pt}{
   \xymatrix@C=4em{
     &&
     S^7
     \!\sslash\!
       \mathrm{Sp}(2)
     \ar[d]^{\small
       h_{\mathbb{H}}
       \sslash
         \mathrm{Sp}(2)
     }
     \\
     &&
     S^4
     \!\sslash\!
       \mathrm{Sp}(2)
     \ar[d]
     \\
     X^8
     \ar@/^2pc/@{-->}[uurr]^-{ \exists }
     \ar@/^1pc/@{-->}[urr]^-{ c }
     \ar[rr]^-{ \tau }
     \ar[drr]_-{ T X^8 }
     &&
     B
       \mathrm{Sp}(2)
     \ar[d]
     \\
     &&
     B \mathrm{Spin}(8)
   }
   }
   \phantom{AA}
   \Longrightarrow
   \phantom{AA}
   [\widetilde G_4]
   \;=\;
   \tfrac{1}{2} p_1(T X^8)
   \;\in\;
   H^4
   (
    X^8; \mathbb{R}
   )
   \,.
 $$
\end{prop}
\begin{proof}
  By Prop. \ref{hHsslash}, the cocycle $c$
  in degree 4 twisted Cohomotopy itself
  is equivalently further reduction of $\tau$ to
  topological
  $
   \mathrm{Sp}(1)
   \!\boldsymbol{\cdot}\!
   \mathrm{Sp}(1)
   \!\boldsymbol{\cdot}\!
   \mathbb{Z}_2
   $-structure
   (Example \ref{Sp1CentralProductGroups}).
   Similarly, the assumed factorization through degree-7 Cohomotopy
   is equivalently existence of yet further reduction to topological
   $
   \mathrm{Sp}(1)
   \boldsymbol{\cdot}
   \mathbb{Z}_2
   $-structure,
   via inclusion of the first factor
   $$
     \xymatrix@R=1.3em{
       S^7
       \!\sslash\!
       \mathrm{Sp}(2)
       \ar[rr]^-{ \simeq }
       \ar[dd]_{
         h_\mathbb{H}
         \sslash
           \mathrm{Sp}(2)
       }
       &&
       B
       \big(
         \mathrm{Sp}(1)
           \boldsymbol{\cdot}
         \mathbb{Z}_2
       \big)
       \ar[dd]|-{
         B(
           [q, 1] \mapsto [q,1,1]
         )
       }
       \ar@{}[r]|-{\simeq}
       &
       B \mathrm{Sp}(1)
       \ar[dd]|-{ B( q \mapsto (q,1) )  }
       \\
       \\
       S^4
       \!\sslash\!
         \mathrm{Sp}(2)
       \ar[rr]_-{ \simeq }
       &&
       B
       \big(
         \mathrm{Sp}(1)
           \boldsymbol{\cdot}
         \mathrm{Sp}(1)
           \boldsymbol{\cdot}
         \mathbb{Z}_2
       \big)
       \ar@{}[r]|-{ \simeq }
       &
       B
       \big(
         \mathrm{Sp}(1)
         \times
         \mathrm{Sp}(1)
       \big)
     }
   $$
   This means that the pullback along the equivariant
   quaternionic Hopf fibration is given by projection to the first component $p_1^{(1)}$
   (in the notation of Lemma \ref{FreeIntegralCohomologyOfBSpin4Spin3}).
   But, by \eqref{FirstPontrjaginFromBSpin4ToBSpin3BSpin3} in
   Lemma \ref{BSpinRationalHomotopy}, the difference between the
   universal avatar class $\widetilde \Gamma_4$ of the half integral shifted flux
   (Prop. \ref{CohomologicalCharacterizationOfTwistedCohomotopyDegree4}
   and Remark \ref{CFieldUniversal}) and the class $\tfrac{1}{2}p_1$
   has no such first component:
   $$
   \xymatrix{
     \widetilde \Gamma_4 - \tfrac{1}{2}p_1
     \;=\;
     -\tfrac{1}{4}p^{(2)}_1
     \ar@{|->}[rr]^-{(h_{\mathbb{H}}\sslash \mathrm{Sp}(2))^\ast
     }
     &&
     0
     }.
   $$
   With this, the statement follows from \eqref{DiagramProvingHalfIntegralFluxQuantization}
   in the proof of Prop. \ref{CohomologicalCharacterizationOfTwistedCohomotopyDegree4}.
\end{proof}

\medskip

\subsection{Integral equation of motion}
\label{IntegralEquationOfMotionFromCohomotopy}

We prove that \hyperlink{HypothesisH}{Hypothesis H}
implies the C-field's ``integral equation of motion''
\eqref{SteenrodSquareVanishes}.

\medskip

\begin{prop}[$Sq^2$-closedness of shifted 4-flux]
  \label{DerivingIntegralEquationOfMotion}
  If $G_4$ is a differential 4-form on
  an 8-manifold $X^8$ as in Def. \ref{The8Manifold}
  and satisfying Hypothesis H (Def. \ref{HypothesisHFor8Manifolds}),
  then the class of
  the shifted form $\widetilde G_4 := G_4 + \tfrac{1}{4}p_1(\nabla)$,
  which is integral by Prop. \ref{CohomologicalCharacterizationOfTwistedCohomotopyDegree4},
  is annihilated  by (mod 2 reduction followed by) the second Steenrod operation:
  \begin{equation}
    \label{Spin4Sq2AnnihilatestildeG4}
    \mathrm{Sq}^2
    \big(
      [\widetilde G_4]
    \big)
    \;=\;
    0\;.
  \end{equation}
\end{prop}
\begin{proof}
  By Prop. \ref{hHsslash} and under triality  (Prop. \ref{QuaternionicSubgroupTriality}) the
  $\tau$-twisted Cohomotopy cocycle  exhibits reduction to $\mathrm{Spin}(4)$-structure:
    $$
    \xymatrix@C=3em{
      &&
      S^4 \!\sslash\! \mathrm{Sp}(2)
      \ar[d]
      \ar[rr]^-{ \simeq }
      &&
      B \mathrm{Spin}(4)
      \ar[d]
      \\
      &&
      B \mathrm{Sp}(2)
      \ar[d]
      \ar[rr]^-{ \simeq }
      &&
      B \mathrm{Spin}(5)
      \ar[d]
      \\
      X^8
      \ar@{-->}@/^1pc/[uurr]^>>>>>{
        c
      }|-{
        \mbox{
          \tiny
          \color{blue}
          \begin{tabular}{c}
            cocycle in
            \\
            $\tau$-twisted
            \\
            Cohomotopy
          \end{tabular}
        }
      }
      \ar@{->}[urr]^-{ \tau }
      \ar[rr]_-{ T X^8 }
      &&
      B \mathrm{Spin}(8)
      \ar[rr]^-{ \simeq }_-{ B \mathrm{tri} }
      &&
      B \mathrm{Spin}(8)
    }
  $$
  But, by Prop. \ref{CohomologicalCharacterizationOfTwistedCohomotopyDegree4},
  the class of $\widetilde G_4$ is the pullback of the class
  $\widetilde \Gamma_4 \in H^4( B \mathrm{Spin}(4); \mathbb{Z})$
  \eqref{IntegralCohomologyOfSpin4}
  along this reduction:
  $$
    [\widetilde G_4]
    \;=\;
    ( B \mathrm{tri} \circ c)^\ast
    \big(
      \widetilde \Gamma_4
    \big)
    \;\in\;
    H^4(
      X^8; \mathbb{Z}
    )
    \,.
  $$
  Under these identifications, the statement
  follows upon using \cite[Cor. 4.2 (1)]{CV98a}, where the element corresponding to
  $\widetilde \Gamma_4$ is denoted $s$,
  while the class $[\widetilde G_4]$ is denoted $S$.
\end{proof}

\medskip

\subsection{7-Flux quantization}
\label{PageCharge}

We prove that \hyperlink{HypothesisH}{Hypothesis H} implies
integrality \eqref{HalfIntegral7Flux} of the Page 7-flux \eqref{PageCharge7Form}.

\medskip
The main result is Theorem\ref{IntegralityOfPageCharge} below.
In order to formulate this,
the key ingredient in the expression for the page charge is
a 3-flux $H_3$ which locally trivializes the C-field 4-flux $G_4$.
The homotopy-theoretic manifestation of this local trivialization
is the homotopy pullback of the quaternionic Hopf fibration.
We first introduce this now in Def. \ref{ExtendedSpacetime}, and
then prove a few technical lemmas about it.

\medskip
\begin{defn}[Extended spacetime]
  \label{ExtendedSpacetime}
  Let $X^8$ be an 8-manifold as in Def. \ref{The8Manifold},
  equipped with topological $\mathrm{Sp}(2)$-structure $\tau$
  and with a cocycle in $c$ in $\tau$-twisted Cohomotopy as in
  Def. \ref{HypothesisHFor8Manifolds}.
  Then we say that the corresponding
  \emph{extended spacetime} is the  fibration $\widehat X^8 \to X^8$
  arising as the homotopy pullback
  of the $\mathrm{Sp}(2)$-equivariant quaternionic Hopf fibration
  (Prop. \ref{hHsslash}) along $c$:
  \begin{equation}
    \label{PullExtendedSpacetime}
    \xymatrix{
      \widehat X^7
      \ar@{}[drr]|-{
        \mbox{
          \tiny
          (pb)
        }
      }
      \ar[d]_{ c^\ast(h_{\mathbb{H}}\sslash \mathrm{Sp}(2)) }
      \ar[rr]
      &&
      S^7 \!\sslash\! \mathrm{Sp}(2)
      \ar[d]^-{
        h_{\mathbb{H}}
        \sslash
        \mathrm{Sp}(2)
      }
      \\
      X^8
      \ar@{-->}[rr]|-{\; c \; }
      \ar[dr]_-{\tau}
      &&
      S^4 \!\sslash\! \mathrm{Sp}(2)
      \ar[dl]
      \\
      &
      B \mathrm{Sp}(2)
    }
      \end{equation}
\end{defn}

\begin{remark}[Nature of the extended spacetime]
  \label{NatureOfExtendedSpacetime}
\item  {\bf (i)}
  The extended spacetime $\widehat X^8$ in Def. \ref{ExtendedSpacetime}
  is an $S^3$-fibration over $X^8$, since the homotopy fiber
  of $h_{\mathbb{H}} \sslash \mathrm{Sp}(2)$ over any point
  is $S^3$, by the pasting law for homotopy pullbacks:
  $$
    \xymatrix{
      S^3
      \ar@{}[drr]|-{
        \mbox{
          \tiny
          (pb)
        }
      }
      \ar[d]
      \ar[rr]
      &&
      S^7
      \ar[rr]
      \ar[d]^-{
        h_{\mathbb{H}}
      }
      \ar@{}[drr]|-{
        \mbox{
          \tiny
          (pb)
        }
      }
      &&
      S^7 \!\sslash\! \mathrm{Sp}(2)
      \ar[d]^-{
        h_{\mathbb{H}}
        \sslash
        \mathrm{Sp}(2)
      }
      \\
      \ast
      \ar[rr]
      &&
      S^4
      \ar[d]
      \ar[rr]
      \ar@{}[drr]|-{
        \mbox{
          \tiny
          (pb)
        }
      }
      &&
      S^4 \!\sslash\! \mathrm{Sp}(2)
      \ar[d]
      \\
      &&
      \ast
      \ar[rr]
      &&
      B \mathrm{Sp}(2)
    }
  $$
  As such, this is the incarnation in non-rational
  parameterized homotopy theory of the
  rational superspace $S^3$-fibration
  over 11-dimensional super-spacetime
  discussed in detail in \cite{Higher-T}\cite{PLB},
  which is classified by the bifermionic component  $\mu_{{}_{\rm M2}}$
  of the C-field super flux form
  \cite[p. 12]{FSS13}\cite[(2.1)]{FSS15}\cite{FSS19d}:
\begin{equation}
  \label{SuperRationalExendedSpacetime}
  \xymatrix@R=1.3em{
    \widehat{
      \mathbb{T}^{10,1\vert \mathbf{32}}
    }
    \ar[rr]^-{
      \mu_{{}_{\rm M5}}
      \;+\;
      h_3 \wedge \mu_{{}_{\rm M2}}
    }
    \ar[dd]
    \ar@{}[ddrr]|-{
      \mbox{
        \tiny
        (pb)
      }
    }
    &&
    S^7_{\mathbb{R}}
    \ar[dd]^{
      (
        h_{\mathbb{H}}
      )_{\mathbb{R}}
    }
    \\
    \\
    \mathbb{T}^{10,1\vert \mathbf{32}}
    \ar[rr]^-{
      (
        \mu_{{}_{\rm M2}},\;
        \mu_{{}_{\rm M5}}
      )
      }
    \ar[dr]_-{ \tau }
    &&
    S^4_{\mathbb{R}}
    \ar[dl]
    \\
    &
    K(\mathbb{R},4)
  }
\end{equation}
 \item {\bf (ii)}
  By the universal property of homotopy pullbacks,
  the extended spacetime $\widehat{X}$ in
  Def. \ref{ExtendedSpacetime} is the classifying space for maps $\phi$ to $X$   equipped with
  a cocycle $\widehat c$ in degree 7 twisted Cohomotopy that exhibits the  degree 4 twisted
  Cohomotopy cocycle $\phi^\ast(c)$ as factoring
  through the quaternionic Hopf fibration, via a homotopy $H_3$:
  \begin{equation}
    \label{ExtendedSpacetimeAsClassifyingSpace}
    \raisebox{30pt}{
    \xymatrix@C=5em{
      \widehat Q_{M5}
      \ar@/^1pc/[drrr]^-{ \widehat c }_>>>>>{\ }="s"
      \ar@/_1pc/[ddr]_-{\phi}
      \ar@{-->}[dr]|-{ (\phi,\widehat{c}, H_3) }
      \\
      &
      \widehat X^8
      \ar[d]
      \ar[rr]|>>>>>>>>>>>>>>>{ \phantom{AAA} }
      &&
      S^7 \!\sslash\! \mathrm{Sp}(2)
      \ar[d]^-{
        h_{\mathbb{H}}
        \sslash
        \mathrm{Sp}(2)
      }
      \\
      &
      X^8
      \ar@{->}[rr]^-{ c }^<<<{\ }="t"
      \ar[dr]_-{\tau}
      &&
      S^4 \!\sslash\! \mathrm{Sp}(2)
      \ar[dl]
      \\
      &
      &
      B \mathrm{Sp}(2)
      \ar@/^.4pc/@{=>}^{H_3} "s"; "t"
    }
    }
  \end{equation}

  By Lemma \ref{VanishingBraneChargeInTermsOfPontrjaginThom}
  and Prop. \ref{PTVanishing4Flux},
  factorization through the quaternionic Hopf fibration is
  the intrinsic cohomotopical meaning of
  the concept of ``vanishing 4-flux'', and here this
  is reflected by the trivializing homotopy $H_3$.
  But this means that,
  under \hyperlink{HypothesisH}{\it Hypothesis H},
  the extended spacetime of Def. \ref{ExtendedSpacetime}
  is really
  the classifying space for fundamental M5-brane sigma-model configurations
  in $X$ with extended worldvolume $\hat Q_{\rm M5}$
  Hence the extended spacetime $\widehat X^8$
  is the classifying space space for the fluxed M5-brane sigma model in the
  M2-brane background $X^8$.
  This is discussed in detail in
  \cite{FSS19c}.
  For the super-rational analog \eqref{SuperRationalExendedSpacetime}
  this was discussed in
  in \cite[Rem. 3.11]{FSS13}\cite[p. 4]{FSS15}.

\end{remark}

Next we characterize,
in Prop. \ref{DifferentialFormDataOnExtendedSpacetime} below,
the differential form data
encoded in \eqref{ExtendedSpacetimeAsClassifyingSpace}.
For that we need the following two lemmas.
The statement of Lemma \ref{SullivanModelOfHopfFibration} is standard
but rarely made fully explicit.
We spell it out since it is crucial
for our new result, Lemma \ref{SullivanModelOfSp2EquivariantHopfFibration}.
For background on Sullivan models see e.g. \cite[Sec. 12]{FHT00}.
\

\begin{lemma}[Sullivan model of the Hopf fibration]
  \label{SullivanModelOfHopfFibration}
  The Sullivan model of the quaternionic Hopf fibration $h_{\mathbb{H}}$,
  with explicit normalization of its generators,
  is:
  $$
    \xymatrix@R=1.5em{
      S^7
      \ar[dd]^-{h_{\mathbb{H}}}
      &&
      \mathbb{R}
      [
        \omega_7
      ]
      /
      (
        {\begin{aligned}
          d \omega_7 & = 0
        \end{aligned}} \!\!\!\!\!\!\!
      )
      \ar@{<-}[dd]_{
        ( h_{\mathbb{H}})^\ast
      }^{
        \mbox{
          \footnotesize
          $
          {\begin{aligned}
            \omega_4 & \mapsto 0
            \\
            \omega_7 & \mapsto \omega_7
          \end{aligned}}
          $
        }
      }
      &
      {\begin{aligned}
        \big\langle \omega_7, [S^7]\big\rangle
        &
        =
        1
      \end{aligned}}
      \\
      \\
      S^4
      &&
      \mathbb{R}
      [
        \omega_4
        ,
        \omega_7
      ]
      \Big/
      \left(
        {\begin{aligned}
          d \omega_4 & = 0
          \\
          d \omega_7 & = - \omega_4 \wedge \omega_4
        \end{aligned}}
      \right)
      &
      {\begin{aligned}
        \big\langle \omega_4, [S^4]\big\rangle
        &
        =
        1
      \end{aligned}}
    }
  $$
\end{lemma}
\begin{proof}
  One way to see this is with \cite[Theorem 6.1]{AnAr78},  by which, under the identification
  of Sullivan generators with linear duals of homotopy groups, the co-binary component of the
  Sullivan differential equals the linear dual of the Whitehead product,
  $[\; \cdot \;, \cdot \; ]_{\mathrm{Wh}}$:
  $$
    \big[ d \omega \big]\vert_{\wedge^2}
    \;=\;
   -  \big[\; \cdot \;, \; \cdot \;\big]_{\mathrm{Wh}}^\ast( \omega)
    \,.
  $$
 Note that both the Whitehead product gives a factor of 2
  $$
    \big[
      [\mathrm{id}_{S^4}],\, [\mathrm{id}_{S^4}]
    \big]_{\mathrm{Wh}}
    \;=\;
    2 \cdot [h_{\mathbb{H}}]
  $$
  as does the evaluation $\langle \; \cdot \; , \; \cdot \; \rangle$ of the wedge square of $\omega_4$
  (by \cite[top of p. 976]{AnAr78}):
  $$
    \big\langle
      \omega_4 \wedge \omega_4
      ,
      S^4
      \wedge
      S^4
    \big\rangle
    \;=\;
    (-1)^{2 \cdot 2}
    \big\langle
      \omega_4, S^4
    \big\rangle^2
    +
    \big\langle
      \omega_4, S^4
    \big\rangle^2
    \;=\;
    2\,,
  $$
  which hence cancel out.
  See also \cite[Example 1 on p. 178]{FHT00}.

  Alternatively, this follows
  by considering the homotopy cofiber of $h_{\mathbb{H}}$,
  whose Sullivan model is the fiber product
  $$
  \xymatrix@C=9em@R=1.2em{
    &
    \left(
      {\begin{aligned}
        d\omega_4 & = 0
        \\
        d \omega_7 & = h \cdot \omega_4 \wedge \omega_4 + \omega_8
        \\
        d \omega_8 & = 0
      \end{aligned}}
    \right)
    \ar[dl]_-{
      \mbox{
        \footnotesize
        $
        {\begin{aligned}
          \omega_4 & \mapsto \omega_4
          \\
          \omega_7 & \mapsto \omega_7
          \\
          \omega_8 & \mapsto 0
        \end{aligned}}
        $
      }
    }
    \ar[dr]^-{
      \mbox{
        \footnotesize
        $
        {\begin{aligned}
          \omega_4 & \mapsto 0
          \\
          \omega_7 & \mapsto \omega_7
          \\
          \omega_8 & \mapsto \omega_8
        \end{aligned}}
        $
      }
    }
    \\
    \left(
      {\begin{aligned}
        d \omega_4 & = 0
        \\
        d \omega_7 & = h \cdot \omega_4 \wedge \omega_4
      \end{aligned}}
    \right)
    \ar[dr]_{
      \mbox{
        \footnotesize
        $
        {\begin{aligned}
          \omega_4 & \mapsto 0
          \\
          \omega_7 & \mapsto \omega_7
        \end{aligned}}
        $
      }
    }
    &&
    \left(
      {\begin{aligned}
        d \omega_7 & = \omega_8
        \\
        d \omega_8 & = 0
      \end{aligned}}
    \right)
    \ar[dl]^-{
      \mbox{
        \footnotesize
        $
        {\begin{aligned}
          \omega_7 & \mapsto \omega_7
          \\
          \omega_8 & \mapsto 0
        \end{aligned}}
        $
      }
    }
    \\
    &
    \left(
      {\begin{aligned}
        d \omega_7 & = 0
      \end{aligned}}
    \right)
  }
$$
and then using the Hopf invariant one theorem
\cite{Adams60} which implies that
$h = \pm 1$.
\end{proof}
\begin{lemma}[Sullivan model of ${\rm Sp}(2)$-equivariant Hopf fibration]
  \label{SullivanModelOfSp2EquivariantHopfFibration}
  The Sullivan model
  of the $\mathrm{Sp}(2)$-parametrized quaternionic Hopf fibration
  $h_{\mathbb{H}} \!\sslash\! \mathrm{Sp}(2)$
  (Prop. \ref{hHsslash})
  is as shown here:
  \begin{equation}
  \label{hHRationalEquivariant}
  \raisebox{60pt}{
  \xymatrix@C=7pt{
    S^7
      \!\sslash\!
    \mathrm{Sp}(2)
    \ar[dd]|-{
      h_{\mathbb{H}}
      \sslash
      \mathrm{Sp}(2)
    }
    \ar[dr]
    &&&
    \mathrm{CE}
    \big(
      \mathfrak{l}
      B
      \mathrm{Sp}(2)
    \big)
    \otimes
    \mathbb{R}[
      \omega_7
    ]\big/
    \big(
      {\begin{aligned}
        d \omega_7
          & =
        -
        \rchi_8 \!\!\!\!\!\!\!
      \end{aligned}}
    \big)
    \ar@{<-}[dd]^-{
      \mbox{
        \footnotesize
        $
        \begin{aligned}
          \omega_4 & \mapsto \tfrac{1}{4}p_1
          \\
          \omega_7 & \mapsto \omega_7
        \end{aligned}
        $
      }
    }_-{
      (
        h_{\mathbb{H}}
        \sslash
        \mathrm{Sp}(2)
      )^\ast
    }
    &
    \;\;
    \big\langle
      \omega_7, [S^7]
    \big\rangle
    =
    1
    \\
    &
    B \mathrm{Sp}(2)
    \\
    S^4
      \!\sslash\!
    \mathrm{Sp}(2)
    \ar[ur]
    &&&
    \mathrm{CE}
    \big(
      \mathfrak{l}
      B
      \mathrm{Sp}(2)
    \big)
    \otimes
    \mathbb{R}[
      \omega_4, \omega_7
    ]\Bigg/
    \left(
      {\begin{aligned}
        d \omega_4 & = 0
        \\
        d \omega_7
          & =
        - \omega_4 \wedge \omega_4
        \\
        & \;\phantom{=}
        -
        \rchi_8
        +
        \big(
          \tfrac{1}{4}p_1
        \big)^2
      \end{aligned}}
    \right)
    &
    \;\;
    \big\langle
      \omega_4, [S^4]
    \big\rangle
    =
    1
  }
  }
\end{equation}
where $\mathrm{CE}\big( \mathfrak{l} B \mathrm{Sp}(2) \big)$
denotes the Sullivan model of the classifying space of $\mathrm{Sp}(2)$.
\end{lemma}
\begin{proof}
  That the domain and codomain Sullivan algebras
  are as shown follows by {\cite[Sec. 15, Example 4]{FHT00}}
  as in the proof of Prop. \ref{SectionsOfRationalSphericalFibrations},
  where the normalization of the generators is
  from Lemma \ref{SullivanModelOfHopfFibration}.
  Here in the bottom right we translated,
  the summand
  $\tfrac{1}{4}p_2$ \eqref{c4pk}
  from the $\mathrm{Spin}(5)$-structure
  for which Prop. \ref{SectionsOfRationalSphericalFibrations} applies,
  to the given $\mathrm{Sp}(2)$-structure, by pullback
  along $B \mathrm{tri}$ \eqref{TrialityAutomorphismDelooped},
  using \eqref{Fractionalp2OfSpin5InTermsOfSp2}
\begin{equation}
  \label{TrialityPullbackOfp2InSp2Structure}
  \big(
    B \mathrm{tri}
  \big)^\ast
  \big(
    \tfrac{1}{4}p_2
  \big)
  \;=\;
  -\rchi_8
  +
  \big( \tfrac{1}{4}p_1 \big)^2
  \,.
\end{equation}
  Now to see that the map $( h_\mathbb{H} \sslash \mathrm{Sp}(2))^\ast$
  in \eqref{hHRationalEquivariant} is given on generators as claimed,
  we use that over any base point of $B \mathrm{Sp}(2)$ the parameterized quaternionic Hopf fibration
  restricts to the ordinary quaternionic Hopf fibration,
  making the following diagram homotopy commutative:
  $$
    \xymatrix@C=15pt{
      S^7
      \ar[dd]_<<<<<<{ h_{\mathbb{H}} }
      \ar[rrr]
      \ar[dr]
      &&&
      S^7 \!\sslash\! \mathrm{Sp}(2)
      \ar[dd]|<<<<<<{
        h_{\mathbb{H}} \sslash \mathrm{Sp}(2)
      }
      \ar[dr]
      \\
      &
      \ast
      \ar[rrr]|>>>>>>>>>>>>>>{ \phantom{AA} }
      &&&
      B \mathrm{Sp}(2)\;.
      \\
      S^4
      \ar[ur]
      \ar[rrr]
      &&&
      S^4 \!\sslash\! \mathrm{Sp}(2)
      \ar[ur]
    }
  $$
This means that the Sullivan model of $h_{\mathbb{H}}\sslash \mathrm{Sp}(2)$
must be a dashed homomorphism that makes the following diagram of dg-algebras commute:
$$
  \xymatrix{
    \mathbb{R}[
      \omega_7
    ]\big/
    \left(
     {\begin{aligned}
        d \omega_7
          & =
        0
      \end{aligned}\!\!\!\!\!\!}
    \right)
    \ar@{<-}[rr]
    \ar@{<-}[dd]_{
      \mbox{
        \footnotesize
        $
        \begin{aligned}
          \omega_4 & \mapsto 0
          \\
          \omega_7 & \mapsto \omega_7
        \end{aligned}
        $
      }
    }
    &&
    \mathrm{CE}
    \big(
      \mathfrak{l}
      B
      \mathrm{Sp}(2)
    \big)
    \otimes
    \mathbb{R}[
      \omega_7
    ]\big/
    \big(
      {\begin{aligned}
        d \omega_7
          & =
        \rchi_8
      \end{aligned}\!\!\!\!\!\!\!}
    \big)
    \ar@{<--}[dd]^-{
      \mbox{
         \footnotesize
         $
         \begin{aligned}
           \omega_4 & \mapsto \tfrac{1}{4}p_1
           \\
           \omega_7 & \mapsto \omega_7
         \end{aligned}
         $
      }
    }
    \\
    \\
    \mathbb{R}[
      \omega_4, \omega_7
    ]\Big/
    \left(
      {\begin{aligned}
        d \omega_4 & = 0
        \\
        d \omega_7
          & =
        - \omega_4 \wedge \omega_4
      \end{aligned}}
    \right)
    \ar@{<-}[rr]
    &&
    \mathrm{CE}
    \big(
      \mathfrak{l}
      B
      \mathrm{Sp}(2)
    \big)
    \otimes
    \mathbb{R}[
      \omega_4, \omega_7
    ]\Bigg/
    \left(
      {\begin{aligned}
        d \omega_4 & = 0
        \\
        d \omega_7
          & =
        - \omega_4 \wedge \omega_4
        \\
        &
        \phantom{=}\;
        \underset{
          =
          ( B \mathrm{tri})^\ast
          (
            \frac{1}{4}p_2
          )
        }{
        \underbrace{
          - \rchi_8 + \big(\tfrac{1}{4}p_1\big)^2
        }
        }
      \end{aligned}}
      \;\;\;\;\;
    \right)
  }
$$
where the horizontal morphisms project away the base algebra
$\mathrm{CE}\big( \mathfrak{l} B \mathrm{Sp}(2)\big)$.

The commutativity of this diagram requires that
the dashed morphism sends $\omega_7 \mapsto \omega_7$.
and by degree reasons it must send
$\omega_4 \mapsto c \cdot p_1$, for some $c \in \mathbb{R}$.
The unique choice for $c$ that makes the map respect the differentials,
in that the second summand in \eqref{TrialityPullbackOfp2InSp2Structure}
cancels out, is clearly $c = \tfrac{1}{4}$.
Alternatively, this follows also by Prop. \ref{PTVanishing4Flux}.
\end{proof}

\begin{prop}[Differential form data on extended spacetime]
  \label{DifferentialFormDataOnExtendedSpacetime}
  Let $X^8$ be an 8-manifold as in Def. \ref{The8Manifold}
  equipped with differential forms $(G_4. G_7)$
  that satisfy Hypothesis H (Def. \ref{HypothesisHFor8Manifolds}),
  hence
  equipped with  topological $\mathrm{Sp}(2)$-structure $\tau$ \eqref{ActionOfSp2Sp1}
  and equipped with a cocycle $c$ in $\tau$-twisted Cohomotopy  (Def. \ref{TwistedCohomotopy})
  with underlying differential forms $(G_4, 2G_7)$ according
  to Def. \ref{HypothesisHFor8Manifolds}
  $$
    \xymatrix@R=1.3em{
      X
      \ar[rr]^-{ (G_4, 2G_7) }
      \ar[dr]_-{ \tau }
      &&
      S^4 \!\sslash\! \mathrm{Sp}(2)\;.
      \ar[dl]
      \\
      &
      B \mathrm{Sp}(2)
    }
  $$
  Then  the pullback of these differential forms
  to the corresponding extended spacetime
  $\widehat X$ from Def. \ref{ExtendedSpacetime}  satisfies
  \begin{align}
    \label{TwistingOfH3univ}
    d\,H_3^{\mathrm{univ}}
    \; &=\;
    \widetilde G_4 - \tfrac{1}{2}p_1(\nabla)\;,
\\
    \label{FluxedCohomologicalPoincareThom}
    d\,\widetilde G_7
    \; & =\;
    -\tfrac{1}{2}\rchi_8(\nabla)\;,
  \end{align}
  where $H_3^{\mathrm{univ}}$
  is the universal 3-form $H_3^{\mathrm{univ}}$
  \eqref{ExtendedSpacetimeAsClassifyingSpace} on $\widehat{X}$,
  and where $\widetilde G_7$ the \emph{shifted 7-flux form}
  or \emph{Page flux}
  \begin{equation}
    \label{Shifted7Flux}
    \widetilde G_7
    \;:=\;
      G_7
        +
      \tfrac{1}{2} H^{\mathrm{univ}}_3 \wedge \widetilde G_4\;.
  \end{equation}
\end{prop}
\begin{proof}
To extract the differential form data
we may compute the defining homotopy pullback \eqref{PullExtendedSpacetime}
in rational homotopy theory and read off the resulting assignment
of generators in the Sullivan model.
By general facts of rational homotopy theory
(recalled e.g. in \cite[A]{FSS16a})
the Sullivan model for
$\widehat X^8$ is given as the pushout along the map corresponding to $(G_4, 2G_7)$ of a
minimal cofibration resolution of the Sullivan model for the equivariant quaternionic Hopf fibration
$h_{\mathbb{H}}\sslash \mathrm{Sp}(2)$. The latter was obtained in Lemma
\ref{SullivanModelOfSp2EquivariantHopfFibration}. By direct inspection one sees that the
minimal cofibration resolution is given as shown on the right of the following diagram:
$$
  \xymatrix@C=2.3em{
    &&
    \mathrm{CE}
    \big(
      \mathfrak{l}
      B
      \mathrm{Sp}(2)
    \big)
    \otimes
    \mathbb{R}[
      \omega_7
    ]\big/
    (
      {\begin{aligned}
        d \omega_7
        & =
        -\rchi_8
      \end{aligned}} \!\!\!\!\!\!\!
    )
    \ar@{<-}[dd]_-{\simeq}^-{
      \mbox{
        \footnotesize
        $
        \begin{aligned}
          h_3 & \mapsto 0
          \\
          \omega_4 & \mapsto \tfrac{1}{4}p_{{}_1}
          \\
          \omega_7 & \mapsto \omega_7
        \end{aligned}
        $
      }
    }
    \ar@/_3pc/@{<-}[dddd]
      |-{ \phantom{ {A \atop {A \atop A} } \atop {{A \atop A} \atop A} } }
      _<<<<<<<<<<<{
        ( h_{\mathbb{H}}  \sslash \mathrm{Sp}(2))^\ast
      }
    \\
    \\
    \mathrm{CE}
    \big(
      \mathfrak{l}
      \widehat{
        X
      }
    \big)
    \ar@{}[ddrr]|-{
      \mbox{
        \tiny
        (po)
      }
    }
    \ar@{<-}[dd]
    \ar@{<-}[rr]|-{
      \mbox{
        \footnotesize
        $
        \begin{aligned}
          \omega_4 & \mapsto G_4
          \\
          \omega_7 & \mapsto 2G_7
          \\
          h_3 & \mapsto H_3^{\mathrm{univ}}
        \end{aligned}
        $
      }
    }
    &&
    \mathrm{CE}
    \big(
      \mathfrak{l}
      B
      \mathrm{Sp}(2)
    \big)
    \otimes
    \mathbb{R}[
      h_3, \omega_4, \omega_7
    ]\Bigg/
    \left(
      {\begin{aligned}
        d h_3 & = \omega_4 - \tfrac{1}{4}p_1
        \\
        d \omega_4 & = 0
        \\
        d \omega_7
          & =
        -
        d h_3 \wedge (\omega_4 + \tfrac{1}{4}p_1)
        \\
        &
        \;\phantom{=}
        -
        \rchi_8
      \end{aligned}}
    \right)
    \ar@{<-^{)}}[dd]^-{
      \mbox{
      \footnotesize
      $
      \begin{aligned}
        \omega_4 & \mapsto \omega_4
        \\
        \omega_7 & \mapsto \omega_7
      \end{aligned}
      $
      }
    }
    \\
    \\
    \mathrm{CE}
    \big(
      \mathfrak{l}
      X
    \big)
    \ar@{<-}[rr]|-{
      \mbox{
        \footnotesize
        $
        \begin{aligned}
          \omega_4 & \mapsto G_4
          \\
          \omega_7 & \mapsto 2G_7
        \end{aligned}
        $
      }
    }
    \ar@{<-}[dr]_-{
      \tau^\ast
    }
    &&
    \mathrm{CE}
    \big(
      \mathfrak{l}
      B
      \mathrm{Sp}(2)
    \big)
    \otimes
    \mathbb{R}[
      \omega_4, \omega_7
    ]\Bigg/
    \left(
      {\begin{aligned}
        d \omega_4 & = 0
        \\
        d \omega_7
          & =
        - \omega_4 \wedge \omega_4
        +
        \big( \tfrac{1}{4}p_1\big)^2 - \rchi_8
      \end{aligned}}
    \right)
    \ar@{<-}[dl]
    \\
    &
    \mathrm{CE}
    \big(
      \mathfrak{l}
      B
      \mathrm{Sp}(2)
    \big)
  }
$$
The differential relations appearing on the right
now immediately imply the claim.
\end{proof}

\begin{prop}[Integrality of Page charge]
  \label{IntegralityOfPageCharge}
  Let $X^8$ be an 8-manifold as in Def. \ref{The8Manifold},
  equipped with differential forms $(G_4, G_7)$ that
  satisfy Hypothesis H, Def. \ref{HypothesisHFor8Manifolds}
  with respect to a topological $\mathrm{Sp}(2)$-structure $\tau$
  and a cocycle in $c$ in $\tau$-twisted Cohomotopy.
  Then for every map
  $$
    i \;:\; S^7 \longrightarrow \widehat X^8
  $$
  from the 7-sphere to the corresponding extended spacetime
  (Def. \ref{ExtendedSpacetime}), the integration of
  the pullback of the Page flux \eqref{Shifted7Flux}
  over the 7-sphere is half-integral:
  $$
    2
    \int_{{}_{S^7}}
      i^\ast \widetilde G_7
    \;\in\;
    \mathbb{Z}
    \,.
  $$
\end{prop}
\begin{proof}
  This is proven as \cite[Theorem 4.6]{FSS19c}.
\end{proof}

\medskip

\subsection{Tadpole cancellation}
\label{M2BraneTadpoleCancellation}

We discuss here how \hyperlink{HypothesisH}{Hypothesis H}
implies the fluxless C-field tadpole cancellation condition \eqref{NIsI}.

\medskip

The key point is to see what precisely ``4-fluxless'' is to mean.
For this, recall that we discussed Cohomotopy cocycles at four levels
of approximation, from the coarse approximation of rational/de Rham cohmology
on the left to full non-abelian Cohomotopy on the right:
\begin{center}
  \begin{tabular}{|c||c|c|c|c|}
    \hline
    \begin{tabular}{c}
  \bf    Cohomology
      \\
  \bf    theory
    \end{tabular}
    &
    \begin{tabular}{c}
      Rational
      \\
      cohomology
    \end{tabular}
    &
    \begin{tabular}{c}
      Integral
      \\
      cohomology
    \end{tabular}
    &
    \begin{tabular}{c}
      Stable
      \\
      Cohomotopy
    \end{tabular}
    &
    \begin{tabular}{c}
      Non-abelian
      \\
      Cohomotopy
    \end{tabular}
    \\
    \hline
    \hline
  \rule{0pt}{2.5ex}  \bf Cocycle
    &
    $G_4$
    &
    $\widetilde G_4$
    &
    $\Sigma^\infty c$
    &
    $c$
    \\
   \hline
  \end{tabular}
\end{center}
On the far left,
for rational and integral cohomology we had found
in Prop. \ref{PTVanishing4Flux} that
cohomotopical fluxlessness is reflected by any factorization
of the 4-cocycle through 7-Cohomotopy
via the quaternionic Hopf fibration $h_{\mathbb{H}}$,
and that this means that the differential flux 4-form
takes its background charge value:
$G_4 = \tfrac{1}{4}p_1(\nabla)$.

\medskip
But stable Cohomotopy theory is finer than its approximation by
de Rham cohomology. Indeed, not \emph{every} factorization through
$h_{\mathbb{H}}$ gives zero in 4-Cohomotopy, instead
there are in general torsion side effects
which disappear only in de Rham cohomology,
as shown in \eqref{QuaternionicHopfFibration}:
It is only those factorizations through $h_{\mathbb{H}}$
which occur in multiples of 24 that are strictly 4-fluxless
as seen in stable Cohomotopy.
Together with the 7-flux quantization of Prop. \ref{IntegralityOfPageCharge}
this means that the cohomotopically normalized 7-flux,
measuring the number of M2-branes, is as in \eqref{CohomotopicalNormalization}.
With this we have:
\begin{prop}[C-field tadpole cancellation via Cohomotopy]
  Let $X^8$ be an 8-manifold as in Def. \ref{The8Manifold}
  which is

  \vspace{-1mm}
  \item {\bf (a)} the complement of a tubular neighborhood
  around a finite number of points in a closed 8-manifold
  (the M2-brane loci) as in \eqref{ComplementOfM2s},
  and

  \vspace{-1mm}
  \item {\bf (b)} equipped with a twisted 7-Cohomotopy cocycle as in
  \eqref{ExtendedSpacetime}, hence with a section $i$
  of the corresponding extended spacetime (Def. \ref{ExtendedSpacetime}).

  \noindent Then the fluxless C-field tadpole cancellation condition
  \eqref{NIsI} holds:
  $$
    N_{{}_{\mathrm{M2}}}
    \;=\;
    I_8\big[X^8\big]
    \,.
  $$
\end{prop}
\begin{proof}
  By definition of
  the cohomotopically normalized 7-flux \eqref{CohomotopicalNormalization},
  we compute as follows:
  $$
    \begin{aligned}
    N_{{}_{M2}}
    & =
    \tfrac{-1}{12}\int_{\partial X^8} i^\ast \widetilde G_7
    \\
    & =
    \tfrac{-1}{12}\int_{{}_{X^8}} i^\ast d \widetilde G_7
    \\
    & =
    \tfrac{-1}{12}\int_{{}_{X^8}} \tfrac{-1}{2}\rchi_{8}(\nabla)
    \\
    & =
    \int_{{}_{X^8}} I_8(\nabla)
    \\
    & =
    I_8\big[X^8\big]\;.
    \end{aligned}
  $$
  Here the third step is by
  \eqref{FluxedCohomologicalPoincareThom}
  from Prop. \ref{DifferentialFormDataOnExtendedSpacetime}
  and the fourth step by \eqref{Spin5DotSpin3ImpliesI8RelatedToChi}
  from Prop. \ref{CohomologicalCharacterizationOfSpin3TimesSpin5StructureIn8d}.
\end{proof}

\medskip

\newpage

\section{Conclusion}
\label{Conclusions}

Perturbative string theory has a precise definition via 2d worldsheet SCFT.
In contrast, the formulation of its non-perturbative completion to M-theory
and of the brane physics this subsumes (see \cite{Duff99B}\cite{BeckerBeckerSchwarz06}),
remains an open problem
(e.g.
\cite[6]{Duff96}\cite[p. 2]{HoweLambertWest97}\cite[p. 2]{NicolaiHelling98}
\cite[p. 6]{Duff98}
\cite[p. 330]{Duff99B}
\cite[12]{Moore14}\cite[p. 2]{ChesterPerlmutter18}\cite
{Witten19}
\footnote{
\cite{Witten19} at 21:15: ``I actually believe that string/M-theory is on the right
track toward a deeper explanation. But at a very fundamental level it's not
well understood. And I'm not even confident that we have a good concept of
what sort of thing is missing or where to find it."}
).
The lack of an actual set of fundamental laws of non-perturbative brane physics
has recently surfaced in a debate on the extent of validity of
the brane uplifts that have been widely discussed for 15 years
\cite{DanielssonVanRiet18}\cite[p. 14-22]{Banks19}.

\vspace{-.1cm}

\medskip

Besides the field of gravity, the only other field in
M-theory at low-energy is the C-field \cite{CJS}.
A list of cohomological conditions on the C-field,
including those shown in \hyperlink{Table1}{Table 1},
have been derived as plausible consistency conditions in various
expected limiting cases of M-theory (effective field theory limits,
decoupling limits etc.) assuming the conjectural string dualities to hold.
One imagines that if M-theory exists
then thereby it must be consistent, and hence ought to imply
all these expected consistency conditions.
In order to  make this actually happen,
the first step in formulating M-theory ought to be the identification of the
generalized cohomology theory that charge-quantizes the C-field,
just as the first step in formulating a quantum consistent
theory of electromagnetism was Dirac's \emph{charge quantization} of
the electromagnetic field:
as a cocycle in (differential) ordinary cohomology (see \cite{Freed00}),

\vspace{-.1cm}

\medskip

The string theory literature has mostly regarded the M-theory C-field as a cocycle in ordinary
4-cohomology, with extra constraints imposed on it by hand. A proposal to build at least one of
these conditions, the shifted flux quantization condition (\cref{HalfIntegralCFieldFluxQuantization}),
into the definition of the cohomology theory (making it a ``mildly generalized cohomology theory'')
has been considered in \cite{DFM03}\cite{HopkinsSinger05}\cite{SatiSchreiberStasheff12}\cite{FSS14a}.
Another condition, the ``integral equation of motion''
(\cref{IntegralEquationOfMotionFromCohomotopy}) has been
argued in \cite{DMWa}\cite{DMWb} to be in correspondence
with one differential of specific degree in the Atiyah-Hirzebruch
spectral sequence for K-theory.
In reaction to this state of affairs, it has been suggested
\cite{Sa1}\cite{Sa2}\cite{Sa3}\cite{tcu}
that the C-field
should be regarded as a cocycle in some genuine generalized cohomology theory,
such as Cohomotopy theory \cite[2.5]{S-top}.
Indeed, if M-theory is as fundamental to physics as it should be, one
may expect the generalized cohomology theory that charge quantizes the C-field
to be more fundamental to mathematics than ordinary cohomology with
some modifications.

\vspace{-.1cm}

\medskip
In order to derive what this generalized cohomology theory
actually is, we had initiated a systematic analysis of the
bifermionic super $p$-brane charges from the point of view of
super rational homotopy theory \cite{FSS13};
see \cite{FSS19a} for review.
We proved in this fully super-geometric setting, albeit in rational approximation,
that the expected charge quantization of the RR-field in twisted K-theory
follows from
systematic analysis
of the D-brane super WZW terms
\cite{FSS16a}\cite{FSS16b}\cite{GaugeEnhancement}.
Then we showed that the exact same
logic applies to the super WZW terms of the M-branes \cite{FSS15}.
The analysis in this case reveals their cohomology theory
to be \cite[3]{FSS15}\cite[2]{FSS16a}:
Cohomotopy cohomology theory in compatible degrees $(4, 7)$, related
by the quaternionic Hopf fibration;
see \cite[7]{FSS19a} for review of this super rational analysis.
This proves that if full M-theory retains the super-space structure
of its low-energy limit,
then the cohomology theory that charge-quantizes the C-field must
be such that its rationalization coincides with
that of Cohomotopy cohomology theory in degrees $(4, 7)$.
While there are many different cohomology theories with the same
rationalization as Cohomotopy theory, one of these
is \emph{minimal} in number of CW-cells:
This is Cohomotopy theory itself.

\vspace{-.1cm}

\medskip
What we have shown in this article is that assuming,
with \hyperlink{HypothesisH}{\it Hypothesis H}, that
Cohomotopy cohomology theory in compatible degrees $(4, 7)$
indeed encodes
the charge-quantization of the C-field even beyond the
rational approximation, then the list in \hyperlink{Table1}{Table 1}
of expected consistency conditions is implied.
Further checks of \hyperlink{HypothesisH}{\it Hypothesis H}
are presented in
\cite{SS19a}\cite{SS19b} (for the case of M-theory orbifolds)
and in \cite{SS19c} (for intersecting branes).

\vspace{-.1cm}

\medskip

\medskip
\noindent {\bf Acknowledgements}
\\
D. F. would like to thank NYU Abu Dhabi for hospitality during the writing of this paper.
We thank Paolo Piccinni for useful discussion and Martin {\v C}adek
for useful communication.  We also thank
Mike Duff,
for comments on an earlier version.


\vspace{1cm}
\noindent Domenico Fiorenza,  {\it Dipartimento di Matematica, La Sapienza Universit\`a di Roma, Piazzale Aldo Moro 2, 00185 Rome, Italy.}

 \medskip
\noindent Hisham Sati, {\it Mathematics, Division of Science, New York University Abu Dhabi, UAE.}

 \medskip
\noindent Urs Schreiber,  {\it Mathematics, Division of Science, New York University Abu Dhabi, UAE, on leave from Czech Academy of Science, Prague.}

\end{document}